\documentclass[conference]{IEEEtran}
\IEEEoverridecommandlockouts
\usepackage{cite}
\usepackage{balance}
\usepackage{amsmath,amssymb,amsfonts}
\usepackage[linesnumbered,ruled,vlined]{algorithm2e}
\usepackage[noend]{algpseudocode}
\usepackage{graphicx}
\usepackage{textcomp}
\usepackage{xcolor}
\usepackage[section]{placeins}
\usepackage{float}
\usepackage{epstopdf}
\usepackage{svg}
\usepackage{subfigure}
\usepackage{mathrsfs}
\usepackage{upgreek}

\usepackage{times}
\usepackage{booktabs} % For formal tables
\usepackage{bbm}
\usepackage{graphicx}
\usepackage{epstopdf}
\usepackage{booktabs}
\usepackage{url}
\usepackage{multirow}
\usepackage{enumerate}
\usepackage{titlesec}

\newtheorem{defn}{Definition}
\newenvironment{proof}{\quad{\it Proof:}}{\hfill $\square$\par}

\newtheorem{lemma}{Lemma}
\newtheorem{example}{Example}
\newtheorem{theo}{Theorem}

\def\BibTeX{{\rm B\kern-.05em{\sc i\kern-.025em b}\kern-.08em
    T\kern-.1667em\lower.7ex\hbox{E}\kern-.125emX}}

\newcommand{\kw}[1]{{\ensuremath {\mathsf{#1}}}\xspace}

\newcommand{\stitle}[1]{\vspace{0.5ex} \noindent{\bf #1}}
\long\def\comment#1{}
\newcommand{\eop}{\hspace*{\fill}\mbox{$\Box$}}

\newcommand{\etal}{\emph{et~al.}\xspace}

\newcommand\figref[1]{Fig.~\ref{#1}}

\newcommand\tabref[1]{Table~\ref{#1}}

\newcommand\secref[1]{Section~\ref{#1}}

\newcommand{\connectedcpn}{\kw{ConnectedGraph}}

\newcommand{\colorful}{\kw{ColorfulCore}}
\newcommand{\encolorful}{\kw{EnColorfulCore}}
\newcommand{\colorfultruss}{\kw{ColorfulSup}}
\newcommand{\encolorfultruss}{\kw{EnColorfulSup}}

\newcommand{\maxrfclique}{\kw{MaxRFC}}

\newcommand{\colorfulpathdp}{\kw{ColorfulPathDP}}
\newcommand{\colorfultriangledp}{\kw{ColorfulTranglePathDP}}

\newcommand{\heur}{\kw{HeurRFC}}
\newcommand{\heurdeg}{\kw{DegHeur}}
\newcommand{\heurcolorfuldeg}{\kw{ColorfulDegHeur}}

\newcommand{\wforder}{\kw{CalColorOD}}

\newcommand{\branch}{\kw{Branch}}

\newcommand{\heurbranch}{\kw{HeurBranch}}

\newcommand{\heurbranchpp}{{\kw{HeurBranch}}++}

\newcommand{\themarker}{\kw{Themarker}}
\newcommand{\wikitalk}{\kw{Wikitalk}}
\newcommand{\flixster}{\kw{Flixster}}
\newcommand{\pokec}{\kw{Pokec}}
\newcommand{\linkedin}{\kw{Linkedin}}
\newcommand{\google}{\kw{Google}}

\newcommand{\aminer}{\kw{Aminer}}
\newcommand{\dbai}{\kw{DBAI}}
\newcommand{\nba}{\kw{NBA}}
\newcommand{\imdb}{\kw{IMDB}}

\newcommand{\DBLP}{\kw{DBLP}}

\renewcommand{\baselinestretch}{0.925
}

\begin{document}

\title{Efficient Maximum Fair Clique Search \\over Large Networks\\
}
\author{{{Qi Zhang}$\scriptsize^{\dag}$, {Rong-Hua Li}$\scriptsize^{\dag}$, {Zifan Zheng}$\scriptsize^{\dag}$, {Hongchao Qin}$\scriptsize^{\dag}$, {Ye Yuan}$\scriptsize^{\dag}$, {Guoren Wang}$\scriptsize^{\dag}$}
	\vspace{1.6mm}\\
	\fontsize{9}{9}\selectfont\itshape
	{$\scriptsize^{\dag}$Beijing Institute of Technology, Beijing, China;} \\
	\fontsize{8}{8}\selectfont\ttfamily\upshape
	{qizhangcs@bit.edu.cn; lironghuabit@126.com; stevenzzf0926@gmail.com;}\\
	{qhc.neu@gmail.com; yuan-ye@bit.edu.cn; wanggrbit@126.com}
}

\maketitle
\begin{abstract}
Mining cohesive subgraphs in attributed graphs is an essential problem in the domain of graph data analysis. The integration of fairness considerations significantly fuels interest in models and algorithms for mining fairness-aware cohesive subgraphs. Notably, the relative fair clique emerges as a robust model, ensuring not only comprehensive attribute coverage but also greater flexibility in distributing attribute vertices. Motivated by the strength of this model, we for the first time pioneer an investigation into the identification of the maximum relative fair clique in large-scale graphs. We introduce a novel concept of colorful support, which serves as the foundation for two innovative graph reduction techniques. These techniques effectively narrow the graph's size by iteratively removing edges that do not belong to relative fair cliques. Furthermore, a series of upper bounds of the maximum relative fair clique size is proposed by incorporating consideration of vertex attributes and colors. The pruning techniques derived from these upper bounds can significantly trim unnecessary search space during the branch-and-bound procedure. Adding to this, we present a heuristic algorithm with a linear time complexity, employing both a degree-based greedy strategy and a colored degree-based greedy strategy to identify a larger relative fair clique. This heuristic algorithm can serve a dual purpose by aiding in branch pruning, thereby enhancing overall search efficiency. Extensive experiments conducted on six real-life datasets demonstrate the efficiency, scalability, and effectiveness of our algorithms.
\end{abstract}

\section{Introduction} \label{sec:introduction}
Graph, consisting of a collection of vertices and edges connecting these vertices, has gained widespread use in representing intricate real-world networks. Graph analysis stands as a crucial tool for understanding network structures and revealing underlying relationships. One fundamental task of graph analysis is cohesive subgraph computation, which aims to identify locally well-connected structures in graphs \cite{chang2018cohesive}. A clique, which requires that every pair of vertices within it must be connected by an edge, represents the most basic form of a cohesive subgraph. The computation of cohesive subgraph related to clique has drawn extensive attention in both academia and industry spheres, resulting in many notable research outcomes such as those highlighted in \cite{bron1973finding, chang2019efficient, eppstein2013listing, li2017minimization, san2016new}. 

%Traditional models have been found to exhibit inherent biases, resulting in unfair decision processes, such as gender barriers, racial discrimination, and age prejudice. Consequently, 

Recently, the concept of fairness has garnered substantial attention within the area of artificial intelligence \cite{DBLP:conf/alt/CotterJS19, DBLP:conf/aistats/Narasimhan18, DBLP:conf/colt/WoodworthGOS17, DBLP:conf/icml/ZemelWSPD13, DBLP:conf/kdd/SinghJ18, DBLP:conf/nips/SinghJ19, DBLP:conf/sigmod/AsudehJS019, DBLP:conf/kdd/BeutelCDQWWHZHC19}. Numerous research endeavors have been initiated to explore methods addressing inherent biases in traditional models, including gender barriers, racial discrimination, and age bias \cite{mehrabi2019debiasing, lipton2018does, louizos2015variational, du2019learning, ross2017right, elazar2018adversarial, zhang2018mitigating, wang2019balanced}. Inspired by these efforts, Pan \etal blazed a trail by introducing fairness into the clique model, and proposed the weak fair clique and strong fair clique models in the field of data mining \cite{DBLP:conf/icde/PanLZDTW22}. Specifically, a weak fair clique is a maximal clique ensuring that the number of vertices for each attribute is at least $k$. On the other hand, a strong fair clique not only requires that the number of vertices with different attributes no less than $k$ but also must be strictly equal. Subsequently, various works on fair cliques are investigated, including the relative fair clique \cite{zhang2023fairness}, absolute fair clique \cite{hao2023afcminer}, fair clique for bipartite graphs \cite{DBLP:journals/corr/abs-2303-03705}, and fair community for heterogeneous graphs \cite{qlp2022community}. The relative fair clique, in particular, mandates that the number of vertices for each attribute is at least $k$, with the difference in the vertex number for different attributes not exceeding $\delta$. Clearly, this model strikes a balance between a weak fair clique and a strong fair clique, ensuring comprehensive attribute coverage while allowing for a more flexible distribution of vertices among attributes. With this robust cohesive subgraph model, we embark on the inaugural investigation of finding the maximum relative fair clique in large-scale graphs.

Identifying the maximum relative fair clique holds significant applications across diverse domains in graph analysis. For example, in collaboration networks, finding the largest team with a small difference in the number of males and females can enhance project creativity by leveraging the distinct strengths that different genders bring to problem-solving, decision-making, and various domains. Similarly, when a project necessitates the convergence of two distinct research domains, it is often imperative to assemble a team that encompasses both areas in a balanced manner, while also being of the maximum size. In social networks, the pursuit of larger and well-connected teams, including both local and foreign members, can significantly enhance product promotion, facilitating the attainment of global brand exposure and influence. In the domain of film, discovering and investing in a substantial team comprising both young talent and seasoned actors is likely to yield higher returns, given that such a team typically possesses a high level of experience and creativity, among other valuable attributes.

% a substantial team with a near-equivalent representation of youthful talents and seasoned actors can bestow a film with a harmonious blend of resources and proficiencies, encompassing creativity, expertise, and experience. Consequently, identifying and investing in such a cinematic endeavor is apt to yield augmented revenues and multifaceted returns. 

%Similarly, when a project needs to bring together two different areas of research, the common desire is to assemble a team that is both balanced between the two areas and the largest possible size to generate comprehensive perspectives and solutions when tackling intricate challenges.

%This expansion facilitates the attainment of global brand exposure and influence, thus fostering a conducive environment for market expansion and successful promotional endeavors.

To address the problem of maximum fair clique search, an intuitive approach is to enumerate all relative fair cliques and output the one with the largest number of vertices. Nevertheless, this approach is computationally expensive, especially for large graphs, as finding all relative fair cliques is NP-hard \cite{DBLP:conf/icde/PanLZDTW22}. Given our goal of finding the relative fair clique with the largest size, a more efficient approach is typically developed with a focus on three crucial aspects: (i) introducing efficient graph reduction techniques to narrow the size of the graph before performing the branch-and-bound search; (ii) designing effective upper bounds on the size of relative fair clique, enabling the pruning of branches that are unlikely to contain the maximum relative fair clique; (iii) devising heuristic algorithms that quickly identify a larger relative fair clique to prune branches further. In alignment with these three aspects, we make the following contributions.

%we introduce two novel graph reduction techniques based on the concept of colorful support and derive several non-trivial upper bounds for pruning the search branches, taking into account diverse perspectives such as the color and attributes of vertices. Furthermore, a greedy algorithm is introduced, relying on both degree and colorful degeneracy, capable of efficiently identifying a larger fair clique within a linear time. In summary, 

\stitle{Novel~graph~reduction techniques.} We introduce a novel concept called ``colored support'' and use it to define a specific subgraph, which is demonstrated to encompass all relative fair cliques. To compute this subgraph, the \colorfultruss algorithm is presented with a peeling strategy to iteratively remove edges that are not permissible within relative fair cliques. Additionally, the enhanced colorful support based reduction is provided to further reduce the graph size.

\stitle{A series~of~upper~bounds~for~branch~pruning.} We concentrate on the colors and attributes of vertices and devise several intuitive upper bounds with low computational complexity, such as the attribute-color-based upper bound and the enhanced attribute-color-based upper bound. To enhance pruning capability further, we develop the colorful degeneracy-based upper bound, the colorful h-index-based upper bound, and the colorful path-based upper bound. Despite the potential for slightly increased computational costs, the superior pruning performance of these advanced upper bounds ultimately contributes to the search efficiency of the maximum relative fair clique.

\stitle{Efficient~heuristic~search~algorithms.} We present a heuristic algorithm combining the degree greedy and color degree greedy strategies. This algorithm produces a larger relative fair clique with linear time complexity, contributing to pruning the search branches.

\stitle{Extensive~experiments.} We conduct comprehensive experimental studies to evaluate the proposed algorithms using six real-world datasets. The results demonstrate that: (i) the colorful support based reduction and its enhanced version significantly remove edges not contained in relative fair cliques; (ii) the proposed upper bounds markedly reduce the runtime for the maximum relative fair clique search; (iii) the relative fair clique size yielded by our heuristic algorithm closely align with the size of the maximum relative fair clique. In most datasets, the difference does not exceed 6. Additionally, we conduct four case studies on real-life graphs with different attributes. The results show that our algorithms can identify the maximum relative fair clique, making it a versatile tool applicable in various domains including product marketing, team formation, business investment, and more.

\comment{
\stitle{Organization.} We introduce some important notations and formulate our problem in \secref{sec:preliminaries}. \secref{sec:graphreduction} presents the novel graph reduction techniques to exclude the vertices that are definitely not in fair cliques. \secref{sec:branchboundalg} presents the branch-reduce-bound algorithm with several carefully-designed upper bounds. The heuristic algorithms for finding a larger relative fair clique are developed in \secref{sec:heuristicalg}. \secref{sec:experiments} reports the experimental results. We survey related studies in \secref{sec:relatedwork} and conclude this work in \secref{sec:conclusion}.
}

\section{Preliminaries} \label{sec:preliminaries}
In this paper, we focus on an undirected and unweighted attributed graph $G = (V, E, A)$, where $V$ represents the set of vertices, $E$ stands for the set of edges, and $A$ is the set of vertex attributes. Let $n=|V|$, $m=|E|$ be the number of vertices and edges, respectively. We specifically concentrate on the scenario of two-dimensional attributes, i.e., $A=\{a, b\}$, and the number of attributes is $A_n=2$. Given a vertex $v$, its attribute is denoted as $A(v)$. The set of $v$'s neighbors is denoted as $N_G(v)$, i.e., $N_G(v) = \{ u \in V | (u, v) \in E\}$, and $deg_G(v) = |{N_G(v)}|$ represents the degree of $v$. Denote by $d_{max}$ the maximum degree of the vertices in $G$. For a subset $S \subseteq V$, the subgraph of $G$ induced by $S$ is defined as ${G_S} = ({V_S}, {E_S})$ where ${V_S} = S$ and ${E_S} = \{(u, v) | u, v \in S, (u, v) \in E\}$. Given an attribute $a$ (resp., $b$), we use $cnt_S(a)$ (resp., $cnt_S(b)$) to indicate the number of vertices in $S$ whose attribute is $a$ (resp., $b$), i.e., $cnt_S(a)=|\{v \in S |A(v)=a\}|$ (resp., $cnt_S(b)=|\{v \in S |A(v)=b\}|$). The subscript $G,S$ in the notations $N_G(v), deg_G(v)$, $cnt_S(a)$ and $cnt_S(b)$ are omitted when the context is self-evident. 

%\vspace*{0.05cm}
\begin{defn}\label{def:relativefairclique}
	\kw{(Relative~fair~clique)} \cite{zhang2023fairness} Given an attributed graph $G=(V, E, A)$ with $A=\{a, b\}$ and two integers $k, \delta$, a clique $C$ of $G$ is a $(k, \delta)$-relative fair clique satisfying the following conditions:
	\begin{enumerate}[(i)]
		{\item \kw{Fairness}}: The number of vertices associated with attribute $a$ and attribute $b$ is no less than $k$, and the difference in their vertex counts is no more than $\delta$, i.e., $cnt_C(a) \ge k$, $cnt_C(b) \ge k$ and $|cnt_C(a)-cnt_C(b)| \le \delta$.  
		{\item \kw{Maximal}}: There is no clique $C' \supset C$ in $G$ satisfying (i). 
	\end{enumerate}
\end{defn}
%\vspace*{0.05cm}

\comment{
\begin{example}
	Consider a graph $G$ in \figref{fig:expgraph}, and suppose that $k=3$ and $\delta=1$. According to Definition \ref{def:relativefairclique}, the subgraph $C_1$ induced by $\{v_1, v_5, v_6, v_7, v_8, v_9\}$ qualifies as a relative fair clique, containing 3 vertices with attribute $a$ and 3 vertices with attribute $b$. On the other hand, clique $C_2=\{v_7, v_8, v_{10}, v_{11}, v_{12}, v_{13}, v_{14}, v_{15}\}$ does not meet the criteria for a relative fair clique, as the difference between $cnt_{C_2}(a)$ and $cnt_{C_2}(b)$ equals to $2 \ge \delta=1$. Similarly, clique $C_3$ induced by $\{v_7, v_8, v_{10}, v_{11}, v_{12}, v_{13}\}$ also fails to be a relative fair clique due to the violation of the maximal requirement. 
	It is evident that by including vertex $v_{14}$ into $C_3$, we can derive a relative fair clique. \eop
\end{example}
}

%\vspace*{-0.2cm}
\begin{figure}[t]
	\centering
	\includegraphics[width=0.275\textwidth]{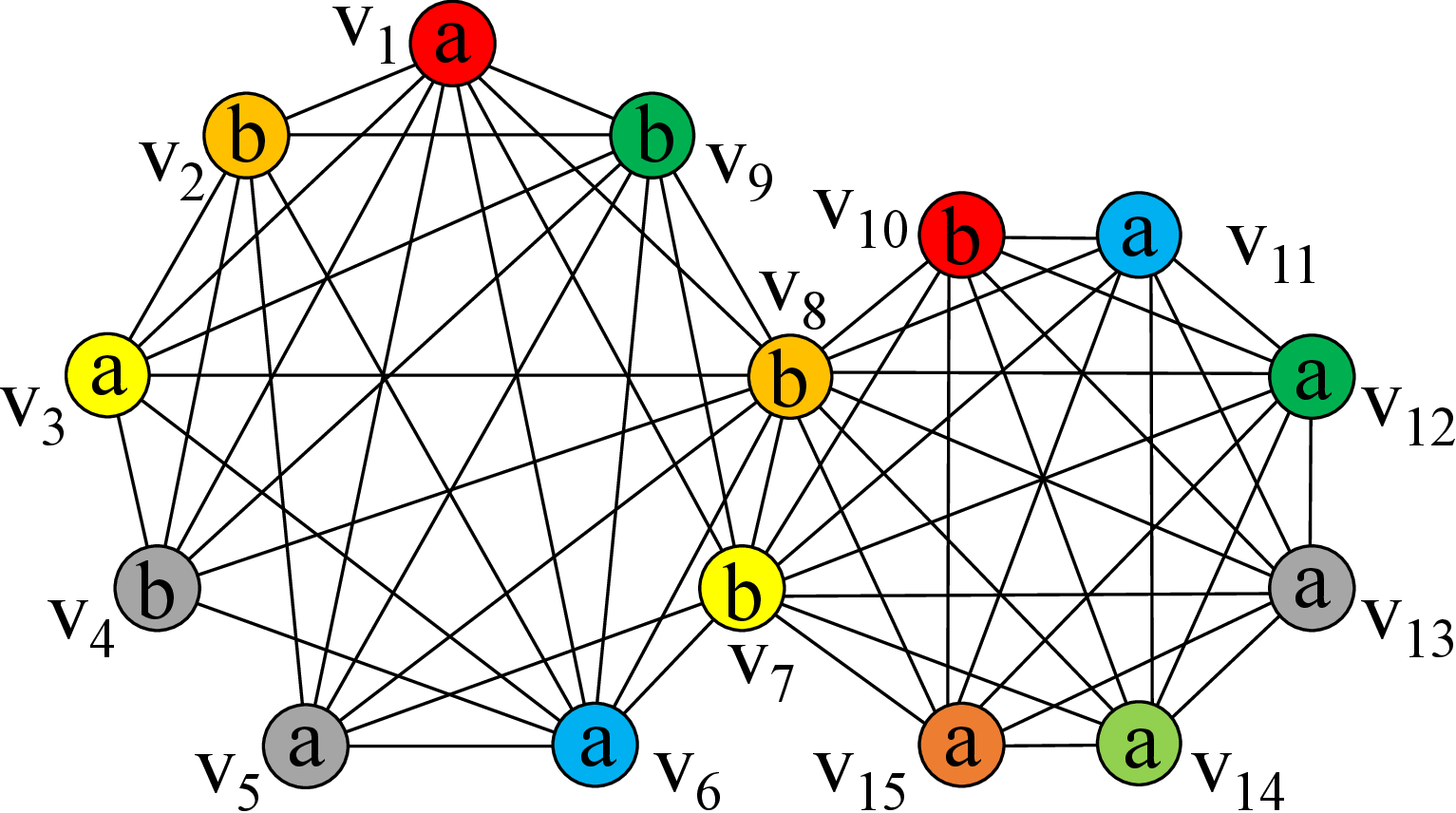}	
	\vspace*{-0.2cm}
	\caption{The example graph $G$}
	\vspace*{-0.4cm}
	\label{fig:expgraph}
\end{figure}

Below, we present the problem formulation of the maximum relative fair clique search, followed by an example to illustrate our problem. Note that, for brevity, we refer to the relative fair clique as a fair clique and use them interchangeably throughout the rest of the paper.

\stitle{Problem formulation.} Given an attributed graph $G=(V, E, A)$ with $A=\{a, b\}$, and two integers $k$, $\delta$, our goal is to identify a relative fair clique in $G$ with the maximum number of vertices.

\comment{\begin{example}
	Consider a graph $G$ depicted in \figref{fig:expgraph}, and suppose that $k=3$ and $\delta=1$. According to the Definition \ref{def:relativefairclique}, there are six relative fair cliques in $G$. Specifically, the subgraph $FC_1$ induced by $\{v_1, v_5, v_6, v_7, v_8, v_9\}$ is a relative fair clique. Given a vertex set $S=\{v_7, v_8, v_{10}, v_{11}, v_{12}, v_{13}, v_{14}, v_{15}\}$, and then $FC_2=S-v_{11}$, $FC_3=S-v_{12}$, $FC_4=S-v_{13}$, $FC_5=S-v_{14}$, $FC_6=S-v_{15}$ are the remaining relative fair cliques. Clearly, the answer of the maximum relative fair clique search problem is $FC_2$ (resp., $FC_3, FC_4, FC_5, FC_6$) because it has the largest number of vertices. \eop
\end{example}
}

\begin{example}
	Consider a graph $G$ shown in \figref{fig:expgraph}, and suppose the parameters $k=3$ and $\delta=1$. Given a vertex set $S=\{v_7, v_8, v_{10}, v_{11}, v_{12}, v_{13}, v_{14}, v_{15}\}$, then the answer to the maximum relative fair clique search problem is $S-v_{11}$ (or $S-v_{12}, S-v_{13}, S-v_{14}, S-v_{15}$). 
\end{example}

\stitle{Challenges.} To address the maximum fair clique search problem, a straightforward approach is to identify all fair cliques and then output the one with the largest number of vertices. However, this approach is fraught with inefficiency, particularly when dealing with large-scale graphs, due to the NP-hard nature of finding all fair cliques. The problem presents several challenges: (i) How to devise effective graph reduction techniques to shrink the size of graphs before initiating the branch-and-bound search; (ii) How to design upper bounding techniques that minimize the exploration of undesirable branches during the branch-and-bound search procedure; (iii) How to develop efficient heuristic algorithms that can rapidly identify a larger fair clique, enabling the efficient pruning of search branches. To tackle these challenges, we introduce novel colorful support based reduction techniques, leveraging insights from truss decomposition. These techniques are capable of significantly reducing the size of the graph by excluding vertices and edges that cannot form a fair clique. Additionally, a series of powerful upper bound based pruning techniques are developed to steer clear of needless branch exploration in the branch-and-bound search process. To further improve efficiency, a heuristic algorithm with linear time complexity is presented, efficiently computing a larger fair clique to facilitate more vigorous branch pruning.

\section{The graph reduction techniques}\label{sec:graphreduction}
This section emphasizes graph reduction techniques as a preliminary step to performing the branch-and-bound search for the maximum fair clique. We initially introduce existing graph reduction methods, and subsequently, explore novel techniques based on the concept of ``colorful support'' to effectively reduce the graph's size.

\subsection{Existing techniques} \label{subsec:colorfulattrdeg}
Existing graph reduction techniques stem from graph coloring, which aims to assign colors to vertices to ensure that connected vertices have distinct colors \cite{matula1972graph, jensen2011graph}. Given a graph $G=(V, E)$, we denote the color of a vertex $u \in V$ by $color(u)$. With graph coloring, Pan \etal introduced two essential concepts: the \emph{colorful degree} and \emph{colorful $k$-core}, forming the basis of their graph reduction techniques.

\begin{defn}
	\label{def:colorfuldeg}
    \kw{(Colorful}~\kw{degree)} \cite{DBLP:conf/icde/PanLZDTW22,zhang2023fairness} Given an attributed graph $G = (V, E, A)$ with $A=\{a,b\}$. For an attribute $a$ (or $b$), the colorful degree of vertex $u$ based on $a$ (or $b$), denoted by $D_{a}(u, G)$ (or $D_{b}(u, G)$), refers to the count of distinct colors among $u$'s neighbors associated with attribute $a$ (or $b$), i.e., $D_{a}(u, G) = | \lbrace color(v) |v \in N(u), A(v) = a \rbrace|$ (or $D_{b}(u, G) = | \lbrace color(v) |v \in N(u), A(v) = b \rbrace|$).
\end{defn}

\begin{defn}
	\label{def:colorfulcore}
	\kw{(Colorful} \kw{k}-\kw{core)} \cite{DBLP:conf/icde/PanLZDTW22,zhang2023fairness} Given an attributed graph $G = (V, E, A)$ with $A=\{a,b\}$ and an integer $k$, a subgraph $H = (V_H, E_H, A)$ of $G$ is a colorful $k$-core if: (i) for each $u \in V_H$, $D_{min}(u, H)=\min\{D_{a}(u, H), D_{b}(u, H)\}\ge k$; (ii) there is no subgraph $H' \subseteq G$ satisfying (i) and $H \subset H'$.
\end{defn}

With these concepts, the colorful $k$-core based graph reduction, namely, \colorful, is shown in Lemma \ref{lem:colorfulkcoreprune} \cite{DBLP:conf/icde/PanLZDTW22,zhang2023fairness}.

\begin{lemma}
	\label{lem:colorfulkcoreprune}
	Given an attributed graph $G = (V, E, A)$ and an integer $k$, any relative fair clique must be contained in the colorful $(k-1)$-core of $G$ \cite{DBLP:conf/icde/PanLZDTW22,zhang2023fairness}.
\end{lemma}

The \colorful reduction considers the attributes of $u$'s neighbors individually, potentially assigning the same color to vertices with attributes a and b. However, this scenario is improbable in a fair clique. Addressing this, Zhang \etal \cite{zhang2023fairness} proposed the enhanced colorful $k$-core based reduction, known as \encolorful, by allocating each color to a specific attribute. Before introducing \encolorful, we give the following important concepts.

\begin{defn}
	\label{def:encolorfuldeg}
	\kw{(Enhanced} \kw{colorful} \kw{degree)} Given a colored attributed graph $G =(V, E, A)$ with $A=\{a, b\}$, the enhanced colorful degree of $u$, denoted as $ED(u)$, is defined as the minimum number of colors assigned exclusively to either attribute $a$ or attribute $b$.
\end{defn}

\begin{defn}
	\label{def:enhancedcolorcore}
	\kw{(Enhanced} \kw{colorful} \kw{k}-\kw{core)} Given an attributed graph ${G}= (V, E, A)$ with $A=\{a, b\}$ and an integer $k$, a subgraph $H = (V_H, E_H, A)$ of $G$ is an enhanced colorful $k$-core if: (i) for each $u \in V_H$, $ED(u) \ge k$; (ii) there is no subgraph $H' \subseteq G$ that satisfies (i) and $H \subset H'$.
\end{defn}

Lemma \ref{lem:encolorfulkcoreprune} details the reduction technique based on the enhanced colorful $k$-core, denoted as \encolorful \cite{zhang2023fairness}.

\begin{lemma}
	\label{lem:encolorfulkcoreprune}
	Given an attributed graph $G = (V, E, A)$ with $A = \{a, b\}$ and an integer $k$, any relative fair clique must be contained in the enhanced colorful $(k-1)$-core of $G$.
\end{lemma}

\subsection{The colorful support based reduction} \label{subsec:colorfulattrsup}
The existing graph reduction techniques focus on eliminating unpromising vertices, offering limited capability to significantly reduce the graph size. To achieve more substantial graph reduction, we introduce the novel concept of ``colorful support''. Building upon this concept, we develop a reduction technique that iteratively deletes edges unlikely to form fair cliques. The concept of \emph{colorful support} for an edge $(u, v)$ is outlined as follows.

\begin{defn}
	\label{def:colorfulsup}
	{\kw{(Colorful} \kw{support)}} Given an attributed graph $G = (V, E, A)$, an edge $(u, v)$, and an attribute $a_i \in A=\{a ,b\}$. The colorful support of $(u, v)$ based on $a_i$, denoted by $\overline{sup}_{a_i}(u,v)$, is the number of distinct colors within the common neighbors of $u$ and $v$ having attribute $a_i$, i.e., ${\overline{sup}}_{a_i}(u,v) = |\{color(w)|w \in N(u) \cap N(v), A(w)= a_i\}|$.
\end{defn}

Below, we introduce the colorful support based reduction technique, namely, \colorfultruss, elaborated in Lemma \ref{lem:colorfulktruss}.

\begin{lemma}
	\label{lem:colorfulktruss}
	Given an attributed graph $G = (V, E, A)$ with $A=\{a, b\}$ and an integer $k$, let $G'$ be the maximal subgraph of $G$, s.t.,
	\begin{enumerate}[(i)]
		\item $\forall (u,v) \in E_{G'}$ with $A(u)=A(v)=a$, ${\overline{sup}}_{a}(u, v)\ge k-2$ and ${\overline{sup}}_{b}(u,v) \ge k$;
		\item $\forall (u,v) \in E_{G'}$ with $A(u)=A(v)=b$, ${\overline{sup}}_{a}(u,v) \ge k$ and ${\overline{sup}}_{b}(u,v) \ge k-2$;
		\item $\forall (u,v) \in E_{G'}$ with $A(u)=a, A(v)=b$ or $A(u)=b, A(v)=a$, ${\overline{sup}}_{a}(u,v) \ge k-1$ and ${\overline{sup}}_{b}(u,v) \ge k-1$;
	\end{enumerate}
	then, any fair clique $C$ in $G$ that adheres to the size constraint of $k$ is encompassed within $G'$.
\end{lemma}

\begin{proof}
	Let's consider an edge $(u,v)$ in the fair clique $C$ with $A(u)=A(v)=a$. According to Definition \ref{def:relativefairclique}, $u$ and $v$ must have at least $k-2$ common neighbors with attribute $a$ and at least $k$ common neighbors with attribute $b$ in $C$. Since vertices with the same color cannot be adjacent, it follows that ${\overline{sup}}_{a}(u, v)\ge k-2$ and ${\overline{sup}}_{b}(u,v) \ge k$. Similar arguments apply to $(u,v)$ in $C$ with $A(u)=A(v)=b$, or $A(u)=a, A(v)=b$, or $A(u)=b, A(v)=a$. Due to space limitations, we omit the proofs for these cases. Hence, it can be concluded that $C$ must be included in the maximal subgraph $G'$.
\end{proof}

\begin{algorithm}[t]
	%\small
	\scriptsize
	\caption{$\colorfultruss(G,k)$}
	\label{alg:colorfulattrsupgraph}
	\KwIn{$G = (V, E, A)$, an integer $k$}
	\KwOut{The maximal subgraph $G'$}
	Color all vertices with a degree-based greedy coloring algorithm\;
	\For{$(u,v) \in E$}
	{
		\For{$w \in N(u) \cap N(v)$}
		{
			{\bf {if}} $M_{(u,v)}(A(w), color(w)) = 0$ {\bf {then}} ${\overline{sup}}_{A(w)}(u,v)$++\;
			$M_{(u,v)}(A(w), color(w))\text{++}$\;
		}
	}
	Let ${\mathcal Q} $ be a priority queue; ${\mathcal Q} \leftarrow \emptyset $\;
	\For{$(u, v) \in E$}
	{
		\If{$A(u)=a$ and $A(v)=a$}
		{	
			\If{${\overline{sup}}_{a}(u,v)<k-2$ or ${\overline{sup}}_{b}(u,v) <k$}
			{
				${\mathcal Q}.push(u,v)$; Remove $(u, v)$ from $G$\;
			}
			
		}
		\ElseIf{$A(u)=b$ and $A(v)=b$}
		{
			\If{${\overline{sup}}_{a}(u,v)<k$ or ${\overline{sup}}_{b}(u,v)<k-2$}
			{
				${\mathcal Q}.push(u,v)$; Remove $(u, v)$ from $G$\;
			}
		}
		\Else{
			\If{${\overline{sup}}_{a}(u,v)<k-1$ or ${\overline{sup}}_{b}(u,v)<k-1$}
			{
				${\mathcal Q}.push(u,v)$; Remove $(u, v)$ from $G$\;
			}
		}
	}
	\While{${\mathcal Q} \neq \emptyset$}
	{
		$(u, v) \leftarrow {\mathcal Q}.pop()$\;
		\For{$w \in N(u) \cap N(v)$}
		{
			\If{$(u, w)$ is not removed}
			{
				$M_{(u,w)}(A(v), color(v))$\text{-}\text{-}\;
				\If{$M_{(u,w)}(A(v), color(v)) \le 0$}
				{
					${\overline{sup}}_{A(v)}(u, w) \leftarrow {\overline{sup}}_{A(v)}(u,w) - 1$\;
					Perform the operations as lines 8-16 for edge $(u, w)$\;
				}
			}
			Perform the operations as lines 20-24 for edge $(v,w)$;\\
		}
	}
	$G' \leftarrow$ the remaining graph of $G$\;
	{\bf return} $G'$\;
\end{algorithm}

\comment{\If{$A(u)=a$ and $A(w)=a$}
	{	
		\If{${\overline{sup}}_{a}(u,w)<k-2$ or ${\overline{sup}}_{b}(u,w) <k$}
		{
			${\mathcal Q}.push(u,w)$; Remove $(u, w)$ from $G$\;
		}
		
	}
	\ElseIf{$A(u)=b$ and $A(w)=b$}
	{
		\If{${\overline{sup}}_{a}(u,w)<k$ or ${\overline{sup}}_{b}(u,w)<k-2$}
		{
			${\mathcal Q}.push(u,w)$; Remove $(u, w)$ from $G$\;
		}
	}
	\Else{
		\If{${\overline{sup}}_{a}(u,w)<k-1$ or ${\overline{sup}}_{b}(u,w)<k-1$}
		{
			${\mathcal Q}.push(u,w)$; Remove $(u, w)$ from $G$\;
		}
}}

Algorithm \ref{alg:colorfulattrsupgraph} depicts the pseudo-code of the colorful support reduction technique \colorfultruss, a variant of the truss decomposition. The main idea is to iteratively delete edges failing to satisfy any of the three conditions in Lemma \ref{lem:colorfulktruss} to reduce the graph size. Specifically, it first performs graph coloring by degree-based greedy method, thereby calculating the colorful support for each edge (lines 1-5). A priority queue ${\mathcal Q}$ maintains edges that violate one of the three conditions in Lemma \ref{lem:colorfulktruss}, which will be removed during the peeling procedure (line 6). The data structure $M_{(u,v)}$ keeps track of the count of common neighbors of $u$ and $v$ with identical attributes and colors (lines 7-16). Subsequently, \colorfultruss iteratively peels edges from the remaining graph according to Lemma \ref{lem:colorfulktruss} (lines 17-25). Finally, the algorithm outputs the remaining graph $G'$ as the maximal subgraph defined in Lemma \ref{lem:colorfulktruss} (lines 26-27). 

\begin{example}
	Consider a graph $G$ in \figref{fig:expgraph}, and suppose that $k=3$ and $\delta=1$. It is evident that $G$ qualifies as a colorful 2-core as $D_{min}(u, G) \ge 2$ for every vertex $u$ in $G$. Meanwhile, $G$ is also an enhanced colorful 2-core. For edge $(v_2, v_5)$, the common neighbors with attribute $a$ are $v_1$ and $v_6$, while the remaining $v_9$ is associated with attribute $b$. Therefore, we have ${\overline{sup}}_{a}(v_2, v_5)=2$ and ${\overline{sup}}_{b}(v_2, v_5)=1$. Clearly, $(v_2, v_5)$ violates condition (iii) in Lemma \ref{lem:colorfulktruss} because of $A(v_2)=b$, $A(v_5)=a$ and ${\overline{sup}}_{b}(v_2, v_5)<3-1=2$, thus it cannot form a fair clique and can be safely removed from $G$. Following this deletion, the remaining graph satisfies Lemma \ref{lem:colorfulktruss}, containing all fair cliques in $G$ with the size constraint $k$.
\end{example}

Below, we analyze the complexity of Algorithm \ref{alg:colorfulattrsupgraph}.

\begin{theo}
	\label{theo:colorfultrusstime}
	Algorithm~\ref{alg:colorfulattrsupgraph} consumes $O(\alpha \times |E|+|V|)$ time using $O(|E| \times |A| \times {{color(G)}})$ space, where $\alpha$ is the arboricity of graph $G$, and ${{color(G)}}$ denotes the number of colors in $G$.
\end{theo}

\begin{proof}
	In line~1, the greedy coloring procedure takes $O(|E|+|V|)$ time \cite{DBLP:conf/spaa/HasenplaughKSL14}. In lines~2-5, it is clear that the algorithm takes $O(\sum_{(u,v) \in E}{\min\{deg(u), deg(v)\}})=O(\alpha \times |E|)$ time. Regarding lines 17-25, the algorithm can update $M_{(u,w)}$ and $M_{(v,w)}$ for each $w\in N(u) \cap N(v)$ in $O(1)$ time. For each triangle $(u, v, w)$, the update operator only performs once, thus the total time complexity of Algorithm~\ref{alg:colorfulattrsupgraph} is bounded by $O(\alpha \times |E|+|V|)$. In terms of space complexity, the algorithm maintains $M_{(u,v)}$ for each edge, resulting in a total space requirement bounded by $O(|E| \times |A| \times {{color(G)}})$.
\end{proof}

\comment{
	\stitle{\color{red}{for code}}.
	Given an attributed graph $G = (V, E, A)$ with $A=\{a, b\}$ and a parameter $k$, let $G'$ be the maximal sub-network of $G$, s.t.,
	\begin{enumerate}[(1)]
		\item $\forall (u,v) \in E$ with $A(u)=A(v)=a$, ${\overline{sup}}_{a}(u, v)\ge k-2$ and ${\overline{sup}}_{b}(u,v) \ge k$;
		\item $\forall (u,v) \in E$ with $A(u)=A(v)=b$, ${\overline{sup}}_{a}(u,v) \ge k$ and ${\overline{sup}}_{b}(u,v) \ge k-2$;
		\item $\forall (u,v) \in E$ with $A(u)=a, A(v)=b$ or $A(u)=b, A(v)=a$, ${\overline{sup}}_{a}(u,v) \ge k-1$ and ${\overline{sup}}_{b}(u,v) \ge k-1$;
	\end{enumerate}
	then, every maximal relative fair clique $C$ in $G$ satisfying the size constraint with $k$ is contained in $G'$.
}

\subsection{The enhanced colorful support based reduction} \label{subsec:encolorfulattrsup}
However, the \colorfultruss technique still exhibits flaws in graph reduction. Take, for instance, an edge $(u,v)$ in \figref{fig:exampleedge}, where $k=4$. The common neighbors of $u$ and $v$ are depicted in \figref{fig:exampleattribute}. According to Definition \ref{def:colorfulcore}, we determine ${\overline{sup}}_{a}(u, v)=3$ and ${\overline{sup}}_{b}(u, v)=4$, implying that $(u,v)$ is preserved after executing \colorfultruss. Nevertheless, it is worth noting that neighbors with attribute $a$ share colors with those bearing attribute $b$. Thus, these seven neighbors are unlikely to coexist within a fair clique. Given these limitations, we draw inspiration from the enhanced colorful degree and propose an alternative: the enhanced colorful support as presented below.

\begin{defn}
	\label{def:encolorfulattrsup}  
	\kw{(Enhanced}~\kw{colorful}~\kw{support)} Given an attributed graph $G = (V, E, A)$, an edge $(u, v)$, and an attribute value $a_i \in A=\{a ,b\}$. The enhanced colorful support of $(u, v)$ based on $a_i$, denoted as $\widetilde{sup}_{a_i}(u,v)$, is the count of colors designated with attribute $a_i$.
\end{defn}

The enhanced colorful support is determined by associating each color with a specific attribute. For instance, when considering an edge $(u,v)$ with $A(u)=A(v)=a$, the process unfolds as follows. The common neighbors of $u$ and $v$ are partitioned into three groups based on their colors: \kw{Group~a}, \kw{Group~b} and \kw{Mixed~group}. Let $c_a$, $c_b$ and $c_m$ be the number of colors within these three respective groups. In case $c_m=0$, we set $\widetilde{sup}_{a}(u,v)=c_a$ and $\widetilde{sup}_{b}(u,v)=c_b$. On the other hand, when $c_a < k-2$, we select $\gamma=\min\{(k-2-c_a), c_m\}$ colors from the \kw{Mixed~group} and assign them to attribute $a$, resulting in $\widetilde{sup}_{a}(u,v)=c_a+\gamma$; otherwise, we set $\widetilde{sup}_{a}(u,v)=c_a$. Next, we update the remaining ${\hat c}_m=c_m-\gamma$ and repeat the color assignment process for attribute $b$. Thus, $\widetilde{sup}_{b}(u,v)=c_b+\min\{(k-c_b), {\hat c}_m\}$ holds when $c_b < k$, while it remains at $c_b$ otherwise. The calculation of $\widetilde{sup}_{a}(u,v)$ and $\widetilde{sup}_{b}(u,v)$ in the scenario where the edge's endpoints possess other attributes can be inferred similarly, although not elaborated due to space constraints. With the definition and calculation method of enhanced colorful support established, we proceed to the subsequent lemma, which contributes to further reducing the graph size.

\begin{lemma}
	\label{lem:enhancedcolorfulktruss}
	Given an attributed graph $G = (V, E, A)$ with $A=\{a, b\}$ and an integer $k$, let $G'$ be the maximal subgraph of $G$, s.t.,
	\begin{enumerate}[(i)]
		\item $\forall (u,v) \in E_{G'}$ with $A(u)=A(v)=a$, ${\widetilde{sup}}_{a}(u, v)\ge k-2$ and ${\widetilde{sup}}_{b}(u,v) \ge k$;
		\item $\forall (u,v) \in E_{G'}$ with $A(u)=A(v)=b$, ${\widetilde{sup}}_{a}(u,v) \ge k$ and ${\widetilde{sup}}_{b}(u,v) \ge k-2$;
		\item $\forall (u,v) \in E_{G'}$ with $A(u)=a, A(v)=b$ or $A(u)=b, A(v)=a$, ${\widetilde{sup}}_{a}(u,v) \ge k-1$ and ${\widetilde{sup}}_{b}(u,v) \ge k-1$;
	\end{enumerate}
	then, every fair clique $C$ in $G$ that satisfies the size constraint with $k$ is contained in $G'$.
\end{lemma}

\begin{figure}[t!]\vspace*{-0.2cm}
	\begin{center}
		\subfigure[$(u,v)$]{
			\label{fig:exampleedge}
			\centering
			\includegraphics[width=0.053\textwidth]{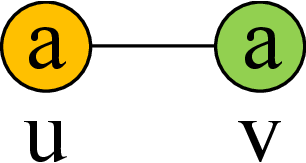}
		}
		\hspace{0.15cm}
		\subfigure[The common neighbors of $u$ and $v$]{
			\label{fig:exampleattribute}
			\centering
			\includegraphics[width=0.23\textwidth]{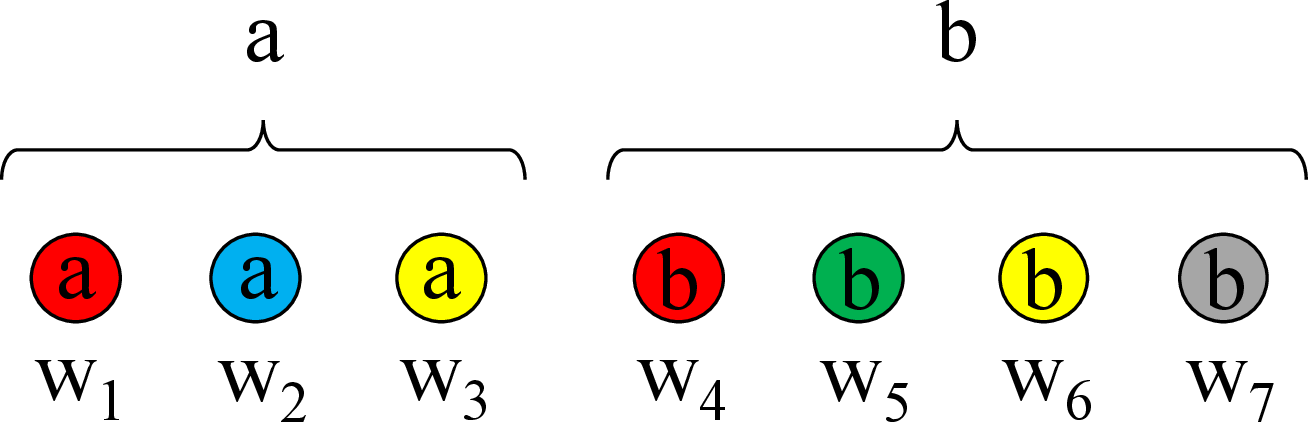}
		}

        \vspace*{-0.15cm}
		\subfigure[The groups of common neighbors]{
			\label{fig:examplegroup}
			\centering
			\includegraphics[width=0.31\textwidth]{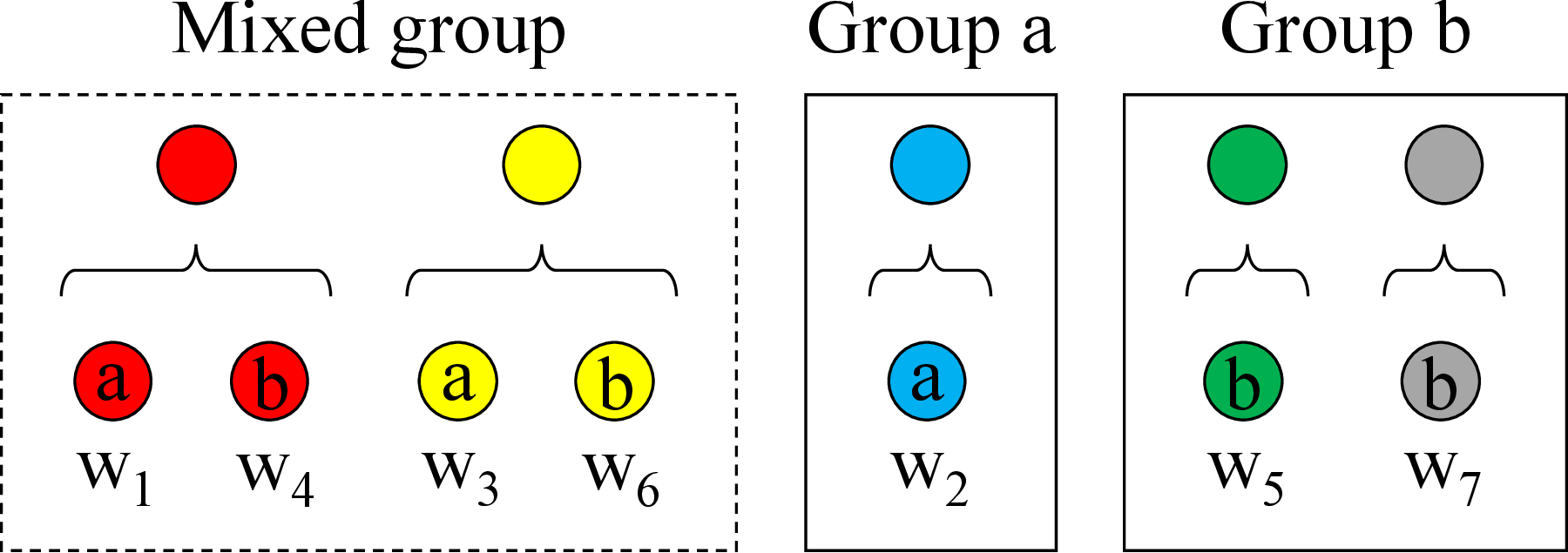}
		}
	\end{center}
	\vspace*{-0.4cm}
	\caption{The shortcoming of \colorfultruss}
	\vspace*{-0.3cm}
	\label{fig:examplecolorultrussshort}
\end{figure}

\comment{\begin{proof}
		Assume that $C$ is a relative fair clique. We first consider an edge $(u,v)$ in $C$ with $A(u)=A(v)=a$. Based on Definition 1, $u, v$ has at least $k-2$ common neighbors in $C$ with attribute value $a$ and at least $k$ neighbors in $C$ whose attribute value is $b$. Since the vertices with the same color must not be adjacent, we have ${\overline{sup}}_{a}(u, v)\ge k-2$ and ${\overline{sup}}_{b}(u,v) \ge k$. The proofs for $(u,v)$ in $C$ with $A(u)=A(v)=b$ or $A(u)=a,A(v)=b$ or $A(u)=b, A(v)=a$ is analogous and we omit them due to the space limit. Hence, $C$ must not be included in the maximal sub-network $G'$.
\end{proof}}

\begin{example}
	Consider the edge $(u,v)$ with $A(u)=A(v)=a$ in \figref{fig:exampleedge} as an illustration. The common neighbors of $u$ and $v$ can be divided into three groups as shown in \figref{fig:examplegroup}. Here, attribute $a$ is uniquely associated with blue, and attribute $b$ is exclusive to dark green and grey. The colors red and yellow, on the other hand, are common to both attributes $a$ and $b$. Thus, we have $c_a=1$, $c_b=2$ and $c_m=2$. For a fair clique with a size constraint of $k=4$ that includes $(u,v)$, it needs to be supplemented with at least $2$ vertices with $a$ and $4$ vertices with $b$. Consider the first attribute $a$. As $a$ is exclusively blue, we must choose $\gamma=\min\{(4-2-1), 2\}=1$ color from the \kw{Mixed~group} to assign to attribute $a$, which is assumed to be red. For attribute $b$, only yellow remains in the \kw{Mixed~group} at this point, so we assign it to $b$. Thus, we have $\widetilde{sup}_{a}(u,v)=2$ and $\widetilde{sup}_{b}(u,v)=3$. Evidently, $(u, v)$ obey condition (i) in Lemma \ref{lem:enhancedcolorfulktruss}, indicating it must not form a fair clique and can therefore be safely removed. 
\end{example}

To derive the maximal subgraph $G'$ in Lemma \ref{lem:enhancedcolorfulktruss}, we employ the peeling strategy and make the following simple adaptation of Algorithm \ref{alg:colorfulattrsupgraph}. Specifically, in lines 2-5, instead of calculating the colorful support for each edge, we compute the enhanced colorful support. Then, we initialize the priority queue $\mathcal Q$ and eliminate unpromising edges based on Lemma \ref{lem:enhancedcolorfulktruss} in lines 7-25. This adapted version, utilizing enhanced colorful support, is named \encolorfultruss and its pseudo-code is omitted due to space limit. Theorem \ref{theo:encolorfultrusstime} shows the complexity of \encolorfultruss.

\begin{theo}
	\label{theo:encolorfultrusstime}
	The \encolorfultruss algorithm's time complexity is $O(\alpha \times |E| \times {{color(G)}})$, utilizing $O(|E| \times {{color(G)}})$ space.
\end{theo}

\begin{proof}
	As mentioned, the greedy coloring procedure takes $O(|E|+|V|)$ time \cite{DBLP:conf/spaa/HasenplaughKSL14}. The algorithm takes $O(\sum_{(u,v) \in E}{\min\{deg(u), deg(v)\}}+|E| \times color(G))=O((\alpha+color(G)) \times |E|)$ time to initialize $Group_{(u,v)}$ and calculate $\widetilde{sup}_{a}$ and $\widetilde{sup}_{b}(u,v)$ for each edge. For each triangle $(u, v, w)$, the update cost is bounded by $O(color)$. Thus the total time complexity amounts to $O(\alpha \times |E| \times {{color(G)}})$. Regarding space complexity, the algorithm maintains the structure $Group_{(u,v)}$ for each color, resulting in a total space requirement $O(|E| \times {{color(G)}})$.
\end{proof}

\section{A branch-and-bound framework} \label{sec:branchboundalg}
This section introduces the basic framework for identifying the maximum fair clique, i.e., \maxrfclique. Following this, we introduce a series of simple yet effective upper-bound techniques designed to curtail the search space. Additionally, we propose more stringent upper bounds aimed at further enhancing the efficiency of the maximum fair clique search algorithm.

\begin{algorithm}[t]
	%\small
	\scriptsize
	\caption{$\maxrfclique(G, k, \delta)$}
	\label{alg:maxfaircliquebasic}
	\KwIn{$G = (V, E, A)$ with $A=\{a,b\}$, two integers $k$, $\delta$}
	\KwOut{The fair clique with the largest size $R^*$}
	$\ddot{G}=(\ddot{V}, \ddot{E}) \leftarrow \encolorful(G, k)$\;
	$\hat{G}=(\hat{V}, \hat{E}) \leftarrow \colorfultruss(\ddot{G}, k)$\;
	$\bar{G}=(\bar{V}, \bar{E}) \leftarrow \encolorfultruss(\hat{G}, k)$\;
	Initialize an array $B$ with $B(i)=false,1 \le i 
	\le \bar{V}$\;
	$R^* \leftarrow \emptyset$\;
	%Heuristically compute a relative fair clique $R^*$ in $\bar G$\;
	%\If{$|R^*|<ub$}
	%{
		\For{$u \in {\bar V}$}
		{
			\If{$B(u) = false$}
			{	
				$C \leftarrow \connectedcpn(u, B)$\;
				${\mathcal{O}} \leftarrow \text{\wforder}(C)$\;
				$R \leftarrow \emptyset$\; %$X \leftarrow \emptyset$\; 
				%$\branch(R, C, X, {\mathcal {O}}, a, -1)$;
				$\branch(R, C, {\mathcal {O}}, a, -1)$;
			}
		} 
		%}
	{\bf return} $R^*$\;
\end{algorithm}

\subsection{The basic framework} \label{subsec:framework}
Here, we present a basic framework, namely, \maxrfclique, for the maximum fair clique search problem. The main idea of \maxrfclique involves employing a branch-and-bound framework along with a simple upper bound derived from set size to prune unpromising branches.

The workflow of \maxrfclique is detailed in Algorithm \ref{alg:maxfaircliquebasic}. $R$ represents an identified clique with the potential for expansion into a fair clique. $C$ denotes a candidate set with $C \cap R = \emptyset$, containing vertices used to extend set $R$. $R^*$ signifies the maximum fair clique discovered thus far. Algorithm \ref{alg:maxfaircliquebasic} initially performs \encolorful, \colorfultruss, and \encolorfultruss sequentially to exclude vertices and edges that are unlikely to be included in fair cliques, thus reducing the graph size (lines 1-3). Then, the algorithm invokes the \branch procedure to find the maximum fair clique in the reduced graph $\bar{G}$ (lines 6-11). Since $\bar{G}$ may be disconnected, we perform \branch on each connected component. For vertex selection order, in line with the method outlined in \cite{DBLP:conf/icde/PanLZDTW22,zhang2023fairness}, the algorithm utilizes the colorful core based ordering, i.e., \wforder (line 9). Finally, \maxrfclique outputs $R^*$ as a result (line 12). 

The \branch procedure, described in Algorithm \ref{alg:relativebacktrack}, alternatively picks a vertex of a particular attribute during the backtracking process to find a fair clique. When the candidate set $C$ becomes empty, it signifies the discovery of a fair clique. At this point, \branch compares the current clique $R$ with the existing optimal solution $R^*$, determining whether an update to $R^*$ is warranted (line 11). Additionally, a basic upper bounding pruning technique, expressed as $|\hat C|+|\hat R|$, is integrated into \branch to reduce the number of branches (line 19). 

It is noteworthy that in Algorithm \ref{alg:maxfaircliquebasic} and Algorithm \ref{alg:relativebacktrack}, we abstain from using a set, often denoted as $X$, to keep track of vertices that could be added to $R$ and have been traversed in earlier search paths. This choice is made due to the fact that $X$ is utilized to prevent redundant enumerations of fair cliques. Its absence does not impact the determination of the maximum fair clique, and the operations on $X$ even introduce an additional time cost. 

\begin{algorithm}[t]
	\scriptsize
	\caption{$\branch(R, C, {\mathcal {O}}, attr{\_}choose, attr{\_}max)$}
	\label{alg:relativebacktrack}
	{\bf Procedure} $\branch(R, C, {\mathcal {O}}, attr{\_}choose, attr{\_}max)$\\
	{\bf for} $u \in C$ {\bf do} $C_{A(u)} \leftarrow C_{A(u)} \cup u$\;
	{\bf for} $u \in R$ {\bf do} $R_{A(u)} \leftarrow R_{A(u)} \cup u$\;
	\If{$C_{attr{\_}choose}=\emptyset$ and $a_{max}=-1$}
	{
		$a_{min} \leftarrow |R_{attr{\_}choose}|$\;
		$a_{max} \leftarrow a_{min}+\delta$\;
	}
	
	{\bf {if}} $|R_{a}|=a_{max}$ {\bf then} $C \leftarrow C-C_{a}$; $C_{a} \leftarrow \emptyset$\;
	{\bf {if}} $|R_{b}|=a_{max}$ {\bf then} $C \leftarrow C-C_{b}$; $C_{b} \leftarrow \emptyset$\;
	
	\If{$C=\emptyset$}
	{
		\If{$|R^*| < |R|$}{
			$R^* \leftarrow R$; {\bf return}\;
		}
	}
	\If{$C_{attr{\_}choose}=\emptyset$}{$\branch({R}, {C}, {\mathcal {O}}, A-attr{\_}choose, a_{max}$); {\bf return}\;}
	\For{$u \in C_{attr{\_}choose}$}
	{
		$\hat{R} \leftarrow R \cup u$; ${\hat C} \leftarrow \emptyset$; $flag \leftarrow false$\;
		\For{$v \in C$}
		{
			\If{$v \in N(u)$ and ${{\mathcal{O}}(v) > {\mathcal O}(u)}$}
			{
				${\hat C} \leftarrow$ ${\hat C} \cup v$; $cnt_{\hat C}(A(v))$++\;
			}
		}
		
		{\bf {if}} $|{\hat C}| + |{\hat R}| < |R^*|$ {\bf then continue}\;
		{\bf {if}} $|{\hat C}| + |{\hat R}| < 2k$ {\bf then continue}\;
		{\bf {for}} $v \in {\hat R}$ {\bf {do}} $cnt_{\hat R}(A(v))$++\;
		\If{$cnt_{\hat R}(a) + cnt_{\hat C}(a) < k$ or $cnt_{\hat R}(b) + cnt_{\hat C}(b) < k$}
		{
			{\bf continue}\;
		}
		%${\hat X} \leftarrow X \cap N(u)$\;
		$\branch({\hat R}, {\hat C}, {\mathcal {O}}, A-{attr{\_}choose}, a_{max}$)\;
		%$X \leftarrow X \cup \{u\}$\;
	}
\end{algorithm}

\subsection{The intuitive and effective upper bounds} \label{subsec:upperboundtech}
In this subsection, our goal is to establish upper bounds for the size of fair cliques within the search instance $(R, C)$. Let $MRFC(R, C)$ denote the size of the maximum fair clique in the instance $(R, C)$, and $(R, C)$ can be entirely pruned if the upper bounds are no larger than $2 \times k+\delta$ or $|R^*|$. An intuitive upper bound of $MRFC(R, C)$ asserts that a fair clique contains all the vertices in the instance $(R, C)$, i.e., Lemma \ref{lem:sizeub}, which is applied in the basic framework \maxrfclique (line 19 in Algorithm \ref{alg:maxfaircliquebasic}).

\begin{lemma}
	\label{lem:sizeub}
	(Size-based Upper Bound) Given an instance $(R, C)$, $ub_s=|R|+|C|$ is an upper bound of $MRFC(R,C)$.
\end{lemma}

\comment{
	\begin{theo}
		The time complexity of Lemma \ref{lem:sizeub} is $O(1)$.
	\end{theo}
}

The size-based upper bound is straightforward. By factoring in the constraint regarding the number of attributes within a fair clique, we can derive a tighter upper bound of $MRFC(R, C)$, as demonstrated in Lemma \ref{lem:attrub}.

\begin{lemma}
	\label{lem:attrub}
	(Attribute-based Upper Bound) Given an instance $(R, C)$, if $|cnt_{R \cup C}(a)-cnt_{R \cup C}(b)| < \delta$ holds, then $ub_a=cnt_{R \cup C}(a)+cnt_{R \cup C}(b)$ is an upper bound of $MRFC(R,C)$; otherwise, $ub_a=2 \times \min\{cnt_{R \cup C}(a), cnt_{R \cup C}(b)\} + \delta$ is an upper bound of $MRFC(R,C)$.
\end{lemma}

\comment{
	\begin{theo}
		The time complexity of Lemma \ref{lem:attrub} is $O(|V|)$.
	\end{theo}
}

On the other hand, we employ the graph coloring technique to deduce upper bounds for $MRFC(R, C)$. Let $G'$ represent the subgraph induced by the vertices in $R \cup C$. We apply a degree-based greedy coloring approach to assign colors to the vertices of $G'$ and denote the number of colors in $G'$ as $color(R \cup C)$. By leveraging the vertex coloring, the ensuing upper bounds can be established.

\begin{lemma}
	\label{lem:colorub}
	(Color-based Upper Bound) Given an instance $(R, C)$, $ub_c=color(R \cup C)$ serves as an upper bound of $MRFC(R, C)$.
\end{lemma}

\comment{
	\begin{theo}
		The time complexity of Lemma \ref{lem:colorub} is $O(|V|+|E|)$.
	\end{theo}
}

Lemma \ref{lem:attrub} and Lemma \ref{lem:colorub} individually focus on either the vertices' attributes or their colors. To achieve a more comprehensive approach, we integrate both attributes and colors to derive a tighter upper bound for $MRFC(R, C)$. Denote $color_{R \cup C}(a)$ (resp., $color_{R \cup C}(b)$) as the count of colors assigned to vertices with attribute $a$ (resp., $b$) within $G'$. The refined attribute-color-based upper bound is outlined as follows.

\begin{lemma}
	\label{lem:attrcolorub}
	(Attribute-color-based Upper Bound) Given an instance $(R, C)$, if $|color_{R \cup C}(a)-color_{R \cup C}(b)| < \delta$, then $ub_{ac}=color_{R \cup C}(a)+color_{R \cup C}(b)$ stands as an upper bound for $MRFC(R,C)$; otherwise, $ub_{ac}=2 \times \min\{color_{R \cup C}(a), color_{R \cup C}(b)\} + \delta$ serves as an upper bound of $MRFC(R,C)$.
\end{lemma}

\comment{
	\begin{theo}
		The time complexity of Lemma \ref{lem:attrcolorub} is $O(|V|+|E|)$.
	\end{theo}
}

In Lemma \ref{lem:attrcolorub}, it is possible for color intersections between vertices with attribute $a$ and those with attribute $b$. Drawing inspiration from the concept of enhanced colorful support, we introduce a tighter upper bound $ub_{eac}$ for $MRFC(R, C)$ by categorizing vertices based on their colors.

\begin{lemma}
	\label{lem:enattrcolorub}
	(Enhance-attribute-color-based Upper Bound) Given an instance $(R, C)$, if $\min\{c_a, c_b\}+c_m < \max\{c_a, c_b\}-\delta$, then $ub_{eac}=2 \times \min\{c_a, c_b\}+c_m + \delta$ serves as an upper bound of $MRFC(R,C)$. Here $c_a$, $c_b$ and $c_m$ are the number of colors in the \kw{Group~a}, \kw{Group~b} and \kw{Mixed~group}, respectively. 
\end{lemma}

\comment{
	\begin{theo}
		The time complexity of Lemma \ref{lem:lem:enattrcolorub} is $O(|V|+|E|)$.
	\end{theo}
}

\comment{
	\begin{proof}
		Given an instance $(R, C)$, if $\min\{c_a, c_b\}+c_m < \max\{c_a, c_b\}-\delta$, then $ub_{eac}=2 \times \min\{c_a, c_b\}+c_m + \delta$ is an upper bound of $MRFC(R,C)$. Here $c_a$, $c_b$, and $c_m$ are the number of colors in the \kw{Group a}, \kw{Group b}, and \kw{Mixed Group}. 
	\end{proof}
}

\begin{theo}
	Computing $ub_s$ has a time complexity of $O(1)$, and computing $ub_a$/$ub_c$/$ub_{ac}$/$ub_{eac}$ carries a time complexity of $O(|V(G')|)$.
\end{theo}

Beyond the mentioned upper bounds, those bounds for the maximum clique size can also serve as constraints for the maximum fair clique size. This is because a fair clique represents a specific instance of a clique, and its size cannot exceed the number of vertices in the maximum clique. The upper bounds of the maximum clique size typically encompass the degeneracy of a graph \cite{lick1970k, seidman1983network}, and the h-index of a graph \cite{DBLP:journals/pnas/Hirsch05}, as illustrated in Lemma \ref{lem:degeneracyub} and Lemma \ref{lem:hindexub}. 

%in accordance with the Definition \ref{def:relativefairclique}, 
\begin{lemma}
	\label{lem:degeneracyub}
	(Degeneracy-based Upper Bound \cite{DBLP:conf/kdd/WangCF13}) Given an instance $(R, C)$, $ub_{\triangle}=\triangle(G')$ is an upper bound of $MRFC(R,C)$ where $\triangle(G')$ denotes the degeneracy of $G'$ (i.e., the maximum core number of $G'$). 
\end{lemma}

\begin{lemma}
	\label{lem:hindexub}
	(H-index-based Upper Bound \cite{DBLP:conf/kdd/WangCF13}) Given an instance $(R, C)$, $ub_{h}=h(G')$ is an upper bound of $MRFC(R, C)$ where $h(G')$ is the maximum value of $h$ such that there exist $h$ vertices with degree no less than $h$ in $G'$.
\end{lemma}

It is proved that $MRFC(R, C) \le ub_{\triangle} \le ub_{h}$, with the computation of degeneracy having a higher time complexity compared to that of h-index of a graph (i.e., Theorem \ref{theo:deghindextime}).

\begin{theo}
	\label{theo:deghindextime}
	The time complexity of computing $ub_{\triangle}$ and $ub_h$ are $O(|E(G')|)$ and $O(|V(G')|)$, respectively \cite{DBLP:conf/kdd/WangCF13}.
\end{theo}

\subsection{The non-trivial upper bounds} \label{subsec:upperboundtech}
In this subsection, we present three novel concepts: ``colorful degeneracy'', ``colorful h-index'', and ``colorful path''. These concepts provide corresponding upper bounds to bound the size of the maximum fair clique within the search branch $(R, C)$. We introduce each of these three non-trivial upper bounds in turn below.

%Given an instance $(R, C)$, coloring the subgraph $G'$ induced by the vertices in $R \cup C$. For a vertex $v$ in $G'$, calculate the colorful degree $D_{a}(v, G')$, $D_{b}(v, G')$, and $D_{min}(v, G')=\min(D_{a}(v, G'), D_{b}(v, G'))$. Based on $D_{min}(v, G')$, we can calculate the attribute colorful core number $accore_{G'}(v)$. 

\stitle{Colorful degeneracy based upper bound.} Building upon the colorful $k$-core concept, the colorful core number and colorful degeneracy are defined as follows.

\begin{defn}
	\label{defn:colordegecore}
	(Colorful core number) Given a colored graph $G$, the colorful core number of a vertex $v$ in $G$, denoted as $ccore(v)$, is the largest $k$ such that the colorful $k$-core of $G$ contains $v$.
\end{defn}

\begin{defn}
	\label{defn:colordege}
	(Colorful degeneracy) The color degeneracy of $G$ is the maximum value among colorful core numbers of vertices in $G$, i.e., $\overline \triangle(G)=\max_{v\in G}ccore(v)$. 
\end{defn}	

With Definition \ref{defn:colordege}, the upper bound of the maximum fair clique derived by colorful degeneracy is given in Lemma \ref{lem:colordegeub}.

\begin{lemma}
	\label{lem:colordegeub}
	(Colorful-degeneracy-based Upper Bound) Given an instance $(R, C)$, let $u$ be the vertex with the largest colorful core number, i.e., $u=\mathop{\arg\max}_{v \in G'}{ccore(v)}$. If $|D_{a}(u, G')-D_{b}(u, G')| < \delta$, then $ub_{cd}= D_{a}(u, G')+D_{b}(u, G')$ is an upper bound of $MRFC(R,C)$; otherwise, $ub_{cd}=2\times{\min\{D_{a}(u, G'), D_{b}(u, G')\}}+\delta$ stands as an upper bound of $MRFC(R,C)$.
\end{lemma}

The time complexity for computing the colorful degeneracy is outlined in Theorem \ref{theo:cdubtime}, which aligns its proof with the complexity of computing the colorful $k$-core, as detailed in \cite{DBLP:conf/icde/PanLZDTW22,zhang2023fairness}.

\begin{theo}
	\label{theo:cdubtime}
	The time complexity of computing $ub_{cd}$ is $O(|E(G')|+|V(G')|)$.
\end{theo}

\stitle{Colorful h-index based upper bound.} Here, we introduce the definition of the colorful h-index, followed by the derived upper bound governing the size of a maximum fair clique within $(R, C)$.

\begin{defn}
	\label{defn:colorhindex}
	(Colorful h-index) Given a colored graph $G$, for a vertex $v$ in $G$, let $D_{min}(v, G)=\min\{D_{a}(v, G), D_{b}(v, G)\}$. We construct a sequence $L={D_{min}(v_1, G), D_{min}(v_2, G), ..., D_{min}(v_t, G)}$. The colorful $h$-index of $G$, denoted as $\overline h(G)$, is the maximum integer $h$ such that there exist at least $h$ vertices with $D_{min}(v, G) \ge h$. 
\end{defn}

Utilizing Definition \ref{defn:colorhindex}, we establish an upper bound for $MRFC(R, C)$ through the colorful h-index, as shown in Lemma \ref{lem:colorhindexub}.

\begin{lemma}
	\label{lem:colorhindexub}
	(Colorful-h-index-based Upper Bound) Given an instance $(R, C)$, coloring the subgraph $G'$ induced by the vertices in $R \cup C$. Let $u$ represent the vertex with $\overline h(G')=D_{min}(u, G')$. If $|D_{a}(u, G')-D_{b}(u, G')| < \delta$, then $ub_{ch}= D_{a}(u, G')+D_{b}(u, G')$ is an upper bound of $MRFC(R,C)$; else, $ub_{ch}=2\times{\min\{D_{a}(u, G'), D_{b}(u, G')\}}+\delta$ is an upper bound of $MRFC(R,C)$.
\end{lemma}

\begin{figure}[t!]\vspace*{-0.2cm}
	\begin{center}
		\subfigure[The colored graph $G'$]{
			\label{fig:pathexp1}
			\centering
			\includegraphics[width=0.18\textwidth]{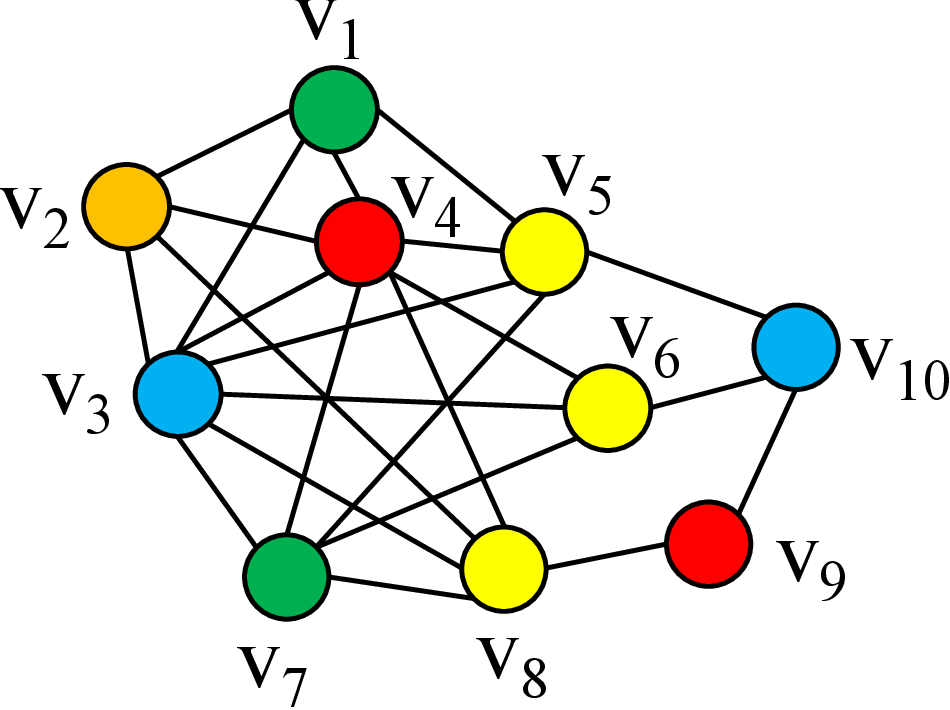}
		}
		%\hspace{0.1cm}
		\subfigure[The DAG graph $\vec G'$]{
			\label{fig:pathexp2}
			\centering
			\includegraphics[width=0.18\textwidth]{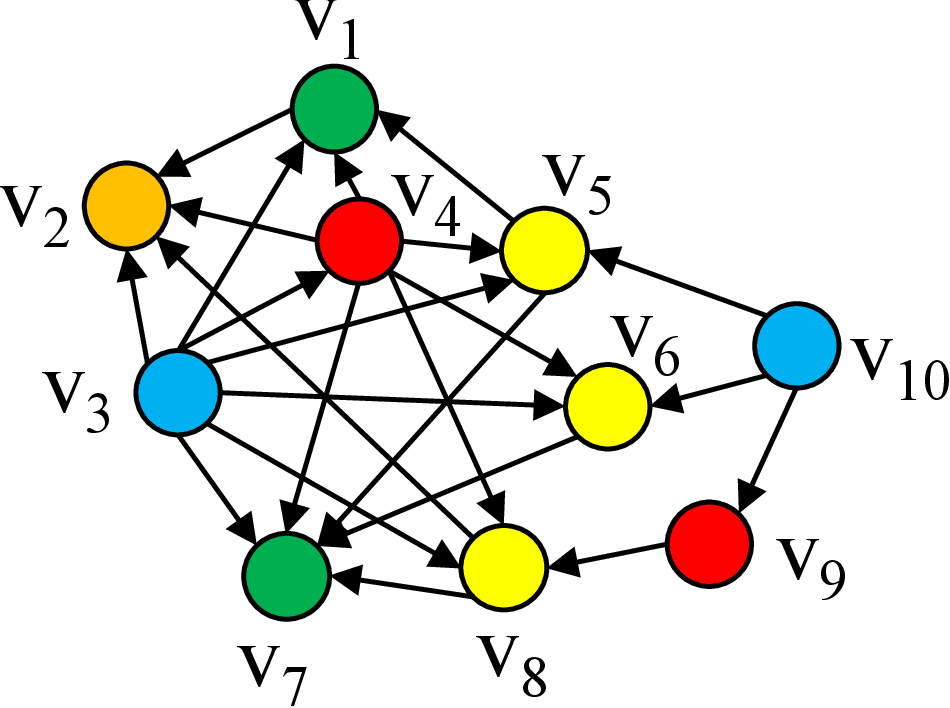}
		}
		\comment{
		\subfigure[The 5-colorful path]{
			\label{fig:pathexp3}
			\centering
			\includegraphics[width=0.120\textwidth]{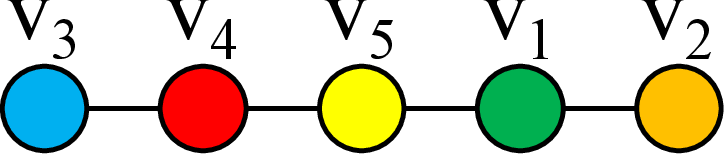}
		}
		\hspace{0.15cm}
		\subfigure[The 4-colorful paths]{
			\label{fig:pathexp4}
			\centering
			\includegraphics[width=0.3\textwidth]{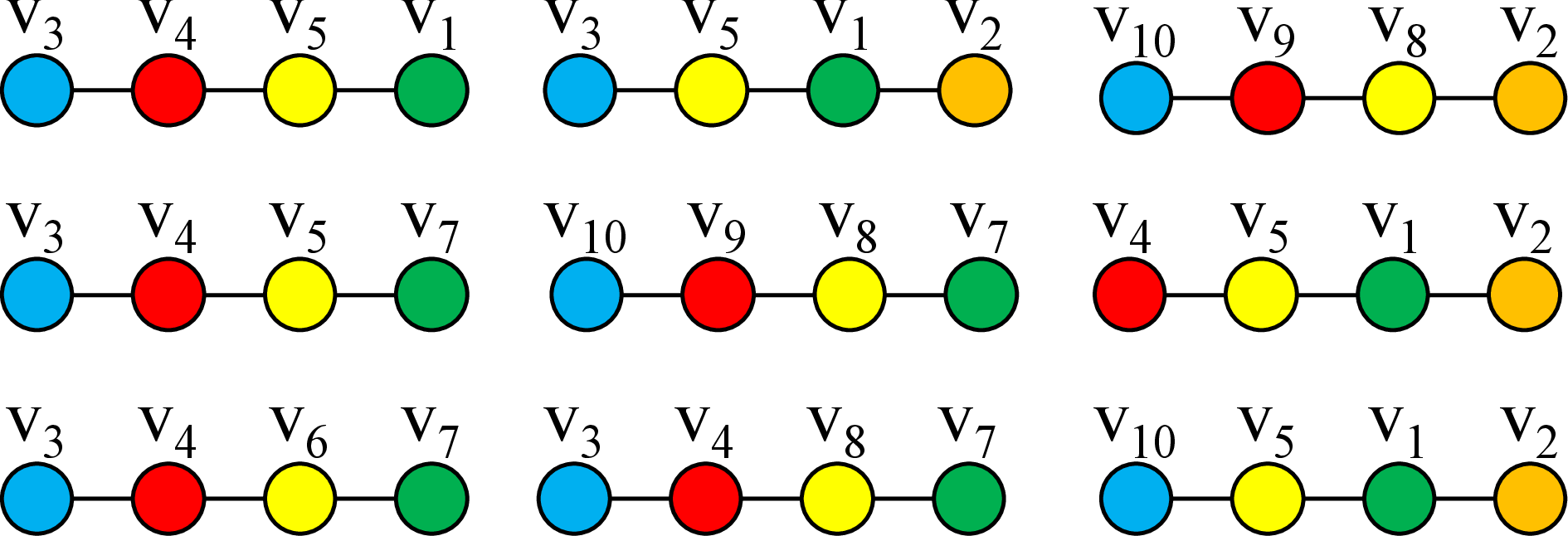}
		}
    }
	\end{center}
	\vspace*{-0.4cm}
	\caption{Running example of the colorful-path-based upper bound}
	\vspace*{-0.3cm}
	\label{fig:pathexp}
\end{figure}

\begin{theo}
	The time complexity of computing $ub_{ch}$ is $O(|E(G')|+|V(G')|)$.
\end{theo}

\begin{proof}
	The calculation of $D_{min}(u, G)$ for each vertex $u$ in $G$ requires $O(\sum_{u \in V(G')}{deg(u)}+|V(G')|) = O(|E(G')|+|V(G')|)$ time. Following this, the computation of the h-index consumes $O(|V(G')|)$ time. Thus, the time complexity for computing $ub_{ch}$ amounts to $O(|E(G')|+|V(G')|)$.
\end{proof}

\stitle{Colorful Path based Upper Bound.} Given an instance $(R, C)$ and the colored subgraph $G'$. Let $CL(G')=\{c_1, c_2, ..., c_p\}$ denote the color set of $G'$, and $V_{c_i}$ represent the vertices with color $c_i$, i.e., $V_{c_i}=\{v \in R \cup C|color(v)=c_i\}$. By utilizing the color ID and vertex ID, a total ordering $\prec$ on $R \cup C$ can be defined with the following rule. For any two vertices $u$ and $v$ in $R \cup C$, $u \prec v$ if and only if: (i) $color(u) < color(v)$; or (ii) $color(u)=color(v)$ and $u_{ID} < v_{ID}$ \cite{DBLP:conf/stoc/EdenRS18}. Based on this total ordering, each edge $(u,v)$ can be oriented from the low-ranked vertex to the high-ranked vertex, resulting in a Directed Acyclic Graph (DAG) $\vec G'$. Below, we provide the definition of a colorful path.

\begin{defn}\label{def:colorfulpath}
	(\kw{Colorful~path}) Given a colored graph $G=(V, E)$, a colorful path $P=\{v_1, v_2, ..., v_p\}$ is a path where each vertex possesses a unique color, i.e., $\forall v_i \in P, \nexists v_j \in P-\{v_i\}, color(v_i)=color(v_j)$.
\end{defn}

Within the (fair) clique, every pair of vertices is connected by edges. Due to the principle of graph coloring, the vertices in the (fair) clique hold different colors, thereby forming a colorful path. It is evident that the largest colorful path can be used to establish an upper bound for the size of the maximum (fair) clique, as detailed in Lemma \ref{lem:colorpathub}.

\begin{lemma}
	\label{lem:colorpathub}
	(Colorful-path-based Upper Bound) Given an instance $(R, C)$, coloring the subgraph $G'$ induced by the vertices in $R \cup C$. We construct its DAG $\vec G'$ using the total ordering and let $CP(G')$ be the largest colorful path in $G'$. Then, $ub_{cp}=|CP(G')|$ is an upper bound of $MRFC(R,C)$.
\end{lemma}

\begin{example}
	Consider a colored graph $G'$ shown in \figref{fig:pathexp1}. We can easily check that $CL=\{c_1=blue, c_2=red, c_3=yellow, c_4=green, c_5=orange\}$. Assuming $k=3$ and $\delta=1$, let's consider the edge $(v_3, v_5)$. Since $color(v_3)=c_1 < color(v_5)=c_3$, we conclude that $v_3 \prec v_5$ based on the total ordering, resulting in the directed edge $<v_3, v_5>$ in $\vec G'$. The DAG $\vec G'$ of $G'$ is depicted in \figref{fig:pathexp2}. Within $\vec G'$, there exists a 5-colorful path $P=\{v_3, v_4, v_5, v_1, v_2\}$ and nine 4-colorful paths. It is evident that $CP(G')=P$, thus rendering $ub_{cp}=|CP(G')|=5$ as an upper bound for $MRFC(R,C)$.
\end{example}

\begin{algorithm}[t]
	% \small
	\scriptsize
	\caption{$\colorfulpathdp(G, R, C)$}
	\label{alg:colorfulpath}
	\KwIn{The graph $G = (V, E)$, an instance $(R, C)$}
	\KwOut{The largest length of colorful paths in $(R,C)$}
	Color all vertices in $R \cup C$ by a degree-based greedy coloring algorithm\;
	Construct the DAG $\vec G' = (V'=(R \cup C), \vec E') $ of $G'$\;
	Let $f(i)$ be the number of vertices in a colorful path ending in $i$ and with the maximum size;\\
    Let $B$ be an array of size $|V'|$ constructed according to the total ordering $\prec$;\\
	\For{$u \in V'$}{
		$f(u) \leftarrow 1$\; 
	}
	\For{$i$ from $0$ to $|V'|-1$}
	{
        $u \leftarrow B(i)$\;
        \For{$v \in N^{-}_{\vec G'}(u)$}
        {
            $f(u) \leftarrow \max\{f(u), f(v)+1\}$\;
        }
        $maxlen \leftarrow  \max\{maxlen, f(u)\}$\;
	}
	{\bf return} $maxlen$\;
\end{algorithm} 

To calculate the longest length of colorful paths in a DAG $\vec G'$, we can employ the Dynamic Programming (DP) approach. In particular, let $N^+_{\vec G'}(u)$ and $N^-_{\vec G'}(u)$ represent the outgoing neighbors and incoming neighbors of $u$ in $\vec G'$. The notation $f(i)$ indicates the number of vertices in a colorful path ending in $i$ with the maximum size. Initially, the value of $f(i)$ is set to 1 for every vertex $i \in R \cup C$. Then, $f(i)$ can be calculated using the transition equation: $f(i) = (\max_{u \in N^-_{\vec G'}(i)}f(u))+1$. 

%$deg^+_{\vec G'}(u)$ and $deg^-_{\vec G'}(u)$ represent the out-degree and in-degree of $u$ in $\vec G'$, respectively. 

The DP-based algorithm for calculating the largest size of colorful paths, referred to as \colorfulpathdp, is detailed in Algorithm \ref{alg:colorfulpath}. It commences by employing the degree-based greedy coloring algorithm to assign colors to the vertices in the graph $G'$. Subsequently, it constructs the DAG $\vec G'$ using a total ordering $\prec$ (lines 1-2). Following this, the algorithm initializes $f(i)$ to $1$ for each vertex (line 3, lines 5-6) and computes $f(i)$ using a DP approach to yield the length of the longest colorful path ending at vertex $i$ within $\vec G'$ (lines 7-11). During the DP process, \colorfulpathdp uses a variable $maxlen$ to maintain the number of vertices in a colorful path with the largest size in $\vec G'$, i.e., $maxlen=|CP(\vec G')|$. Finally, the algorithm outputs $maxlen$ as an upper bound of $MRFC(R, C)$ (line 12). The complexity of Algorithm \ref{alg:colorfulpath} is presented in Theorem \ref{theo:colorpathtime}.

\begin{theo}
	\label{theo:colorpathtime}
	The \colorfulpathdp algorithm requires  $O(|V(G')|+|E(G')|)$ time for calculating $ub_{cp}$. 
\end{theo}

\section{Heuristic algorithms} \label{sec:heuristicalg}

\begin{algorithm}[t]
	%\small
	\scriptsize
	\caption{$\heurdeg(G, k, \delta)$}
	\label{alg:degbasedHeur}
	\KwIn{$G = (V, E, A)$, two integers $k$ and $\delta$}
	\KwOut{The fair clique $R^*$}
	$R^* \leftarrow \emptyset$\;
	$v \leftarrow \max_{v \in V} deg(v)$\;
	$attr\_choose \leftarrow a \in A-A(v)$\;
	$\heurbranch(\{v\}, N(v), attr{\_}choose, R^*, -1)$\;
	{\bf return} $R^*$;
	
	\vspace{0.1cm}
	{\bf Procedure} $\heurbranch(R, C, attr\_choose, R^*, a_{max})$\\
	{\bf for} $u \in C$ {\bf do} $C_{A(u)} \leftarrow C_{A(u)} \cup \{u\}$\;
	{\bf for} $u \in R$ {\bf do} $R_{A(u)} \leftarrow R_{A(u)} \cup \{u\}$\;
	\If{$C_{attr{\_}choose}=\emptyset$ and $a_{max}=-1$}
	{
		$a_{min} \leftarrow |R_{attr{\_}choose}|$\;
		$a_{max} \leftarrow a_{min}+\delta$\;
	}
	{\bf {if}} $|R_{a}|=a_{max}$ {\bf then} $C \leftarrow C-C_{a}$; $C_{a} \leftarrow \emptyset$\;
	{\bf {if}} $|R_{b}|=a_{max}$ {\bf then} $C \leftarrow C-C_{b}$; $C_{b} \leftarrow \emptyset$\;
	\If{$C=\emptyset$}
	{
		$R^* \leftarrow R$; {\bf return}\;
	}
	\If{$C_{attr\_choose} = \emptyset$}
	{
		$attr\_choose \leftarrow a \in A-attr\_choose$\;
		$\heurbranch(R, C, attr\_choose, R^*, a_{max})$\;
		{\bf continue};\
	}
	$v \leftarrow \max_{v \in C, 
		A(v)=attr\_choose} deg(v)$\;
	$attr\_choose \leftarrow a \in A-A(v)$\;
	$\hat R \leftarrow R \cup \{v\}$\;
	$\hat C \leftarrow C \cap N(v)$\;
	{\bf {if}} $|{\hat C}| + |{\hat R}| < k*2$ {\bf then return}\;
	{\bf {for}} $v \in {\hat R}$ {\bf {do}} ${cnt_{\hat R}}(A(v))$++\;
	{\bf {for}} $v \in {\hat C}$ {\bf {do}} ${cnt_{\hat C}}(A(v))$++\;
    {\bf {if}} $cnt_{{\hat R}}(a) + cnt_{{\hat C}}(a) < k$ \emph{or} $cnt_{{\hat R}}(b) + cnt_{{\hat C}}(b) < k$ {\bf {then}} {\bf return}\;
	$\heurbranch(\hat R, \hat C, attr\_choose, R^*, a_{max})$\;
\end{algorithm}

This section introduces a heuristic framework, namely, \heur, to identify a larger fair clique within linear time. The framework relies on two key procedures: the degree-based greedy procedure, referred to as \heurdeg, and the colorful degree-based greedy procedure, known as \heurcolorfuldeg. We begin by detailing \heurdeg and \heurcolorfuldeg before outlining the heuristic framework \heur.

{\stitle{The degree-based greedy procedure.}} The degree-based greedy algorithm, i.e., \heurdeg, computes a larger fair clique by iteratively selecting the vertex with the highest degree to augment $R$ until further extension is not feasible. The pseudo-code of \heurdeg is depicted in Algorithm \ref{alg:degbasedHeur}. To ensure attribute fairness to the greatest extent feasible, \heurdeg adopts an alternating attribute selection strategy similar to \maxrfclique. However, a fundamental disparity exists: while \maxrfclique endeavors to extend $R$ for every vertex in $C$, \heurdeg incorporates only the vertex with the highest degree to $R$. Specifically, during the iteration when a vertex with attribute $a$ is chosen, \heurdeg adds to $R$ the vertex $v \in C$ that satisfies $v \leftarrow \max_{v \in C, A(v)=a} deg(v)$ (line 20). The algorithm terminates when $C$ is empty, yielding $R^*$ as a larger fair clique (lines 14-15).

{\stitle{The colorful degree-based greedy procedure.}} We introduce the colorful degree-based greedy algorithm \heurcolorfuldeg. Similar to \heurdeg, \heurcolorfuldeg employs a greedy strategy to extend the set $R$ based on the colorful degree (as defined in Definition \ref{def:colorfuldeg}). To implement the \heurcolorfuldeg algorithm, we make a slight modification to Algorithm \ref{alg:degbasedHeur}. Specifically, we replace line 2 with $v \leftarrow \max_{v \in V} {\min\{D_a(v), D_b(v)\}}$ and line 20 with $v \leftarrow \max_{v \in C, A(v)=attr\_choose} {\min\{D_a(v), D_b(v)\}}$. %We omit the pseudo-code of \heurcolorfuldeg due to the space limit.

%In particular, we compute $D_{\min}(v, G')=\min\{D_{a}(u, G'), D_{b}(u, G')\}$ in each iteration for each $v \in C$. And then we select the vertex with the largest degree $D_{\min}(v, G')$ to extend $R$ until no vertex can add to $R$.

\comment{
	\begin{algorithm}[t]
		\small
		%\scriptsize
		\caption{The colorful degree-based heuristic algorithm}
		\label{alg:colorfuldegbasedHeur}
		\KwIn{$G = (V, E, A)$, two integers $k$ and $\delta$}
		\KwOut{The relative fair clique $R^*$}
		$R^* \leftarrow \emptyset$\;
		$v \leftarrow \max_{v \in V} {\min\{D^a_G(v), D^b_G(v)\}}$\;
		$attr\_choose \leftarrow a \in A-A(v)$\;
		$\heurbranchpp(\{v\}, N(v), attr{\_}choose,R^*,-1)$\;
		{\bf return} $R^*$;
		
		\vspace{0.1cm}
		{\bf Procedure} $\heurbranchpp(R, C, attr\_choose,C^*,a_{max})$\\
		lines 7-15 in Algorithm \ref{alg:degbasedHeur}\;
		\If{$C_{attr\_choose} = \emptyset$}
		{
			$attr\_choose \leftarrow a \in A-attr\_choose$\;
			$\heurbranchpp(R, C, attr\_choose, R^*, a_{max})$\;
			{\bf continue};\
		}
		
		$v \leftarrow \max_{v \in C, 
			A(v)=attr\_choose} {\min\{D^a_G(v),D^b_G(v)\}}$\;
		$attr\_choose \leftarrow a\in A-A(v)$\;
		lines 22-28 in Algorithm \ref{alg:degbasedHeur}\;
		$\heurbranchpp(R, C, attr\_choose, R^*, a_{max})$\;
	\end{algorithm}
}

{\stitle{The heuristic framework.}} Algorithm \ref{alg:Heurframework} outlines the heuristic framework \heur, encompassing both the degree-based and colorful degree-based procedures. The main idea is to compute two fair cliques by invoking \heurdeg and \heurcolorfuldeg and then select the one with a larger cardinality. It is important to note that upon obtaining a fair clique $R^*$, its size can aid in graph pruning, as a larger fair clique is guaranteed to be within the $(|R^*|-1)$-core subgraph (line 3 and line 8 in Algorithm \ref{alg:Heurframework}). After performing \heurdeg and \heurcolorfuldeg, the \heur algorithm recolors the remaining graph and establishes the upper bound of the maximum fair clique as the number of colors (lines 9-10). Finally, it outputs $R^*$, $ub$, and $color$ and terminates. 

\begin{algorithm}[t]
	% \small
	\scriptsize
	\caption{$\heur(G,k,\delta)$}
	\label{alg:Heurframework}
	\KwIn{$G = (V, E, A)$, two integers $k$ and $\delta$}
	\KwOut{The fair clique $R^*$, the upper bound $ub$, the color array $color$}
	$R^* \leftarrow$ the fair clique by performing \heurdeg on $G$\;
	$k^* \leftarrow |R^*|-1$\;
	$G \leftarrow$ the $k^*$-core of $G$\;
	$\hat R \leftarrow$ the fair clique by performing \heurcolorfuldeg on $G$\;
	\If{$|\hat R| > |R^*|$}{
		$R^* \leftarrow \hat R$\;
		$k^* \leftarrow |R^*|-1$\;
		$G \leftarrow$ the $k^*$-core of $G$\;       
	}
	Color the graph $G$\;
	$ub \leftarrow$ the number of colors in $G$\;
	{\bf return} $(R^*,ub, color(\cdot))$;
\end{algorithm}

{\stitle{Remark.}} \heur can be integrated into the branch-and-bound search algorithm \maxrfclique to improve the efficiency for finding the maximum fair clique. Specifically, after \maxrfclique performs \encolorfultruss for graph reduction, it can invoke the \heur algorithm to yield a larger fair clique $R^*$. Then $R^*$ can be utilized to prune the branch $(R, C)$ during the processing of \branch when the upper bound of $MRFC(R, C)$ does not exceed $|R^*|$. Undoubtedly, a high-quality solution from \heur significantly prunes search branches, thereby reducing the time consumption of the \maxrfclique algorithm. In the experiments, we will compare the sizes of fair cliques found by \heur and \maxrfclique to demonstrate the effectiveness of the proposed heuristic framework.

\begin{theo}
	\label{theo:heurtimecomplexity}
	The \heur algorithm takes $O(|E|+|V|)$ time to output a fair clique with a larger size. 
\end{theo}

%The degeneracy-based clique is obtained as the longest suffix of the degeneracy ordering that is a clique. The pseudocode of MC-DD is shown and discussed in Section A.2 in Appendix. Note that, besides a clique C∗, MC-DD also outputs an upper bound ω of the clique number of G, such that C∗ is certified to be a maximum clique  MC-DD runs in O(m) time as the degeneracy ordering can be computed in O(m) time [5].

\section{Experiments} \label{sec:experiments}
%In this section, we conduct extensive experiments to evaluate the efficiency and effectiveness of the proposed algorithms. We start by describing the experimental setup and then present the results.

\subsection{Experimental setup} \label{sec:setup}

\begin{table}[t]\vspace*{-0.2cm}
\caption{Datasets}
\label{tab:datasets}
\vspace*{-0.5cm}
\begin{center}
{
    \scriptsize
    \setlength{\tabcolsep}{1mm}
	\begin{tabular}{c|c|c|c|c} 
		\hline
		\rule{0pt}{6pt}{\bf Dataset} & $n=|V|$ & $m=|E|$ & $d_{max}$ &Description\\
		\hline
		\rule{0pt}{6pt} \themarker & 69,414	& 3,289,686	& 8,930 & Social network\\
	    \rule{0pt}{6pt} \google &	875,713	&	8,644,102 & 6,332 & Web network\\
	    \rule{0pt}{6pt} \DBLP &	1,843,615	& 16,700,518 & 2,213 & Collaboration network\\
	     \rule{0pt}{6pt} \flixster &	2,523,387	& 15,837,602 & 1,474 & Social network\\
	    \rule{0pt}{6pt} \pokec & 1,632,803	&44,603,928&14,854 & Social network\\
		\hline
		\rule{0pt}{6pt} \aminer& 423,469 & 2,462,224 &712  & Collaboration network\\
		\hline
	\end{tabular}
}
\end{center}
\vspace*{-0.4cm}
\end{table}

{\stitle{Algorithms.}} We implement the colorful support based pruning algorithms, \colorfultruss (Algorithm \ref{alg:colorfulattrsupgraph}) and \encolorfultruss, for graph reduction. We categorize the upper bounds $ub_s$, $ub_a$, $ub_c$, $ub_{ac}$ and $ub_{eac}$ into a group, denoted by $ub_{AD}$, called the advanced upper bound of $MRFC(R,C)$. For the maximum fair clique search problem, we implement the basic framework \maxrfclique (Algorithm \ref{alg:maxfaircliquebasic}) equipped with the following upper bounds to prune unpromising branches: (1) $ub_{AD}$; (2) $ub_{AD}+ub_{\triangle}$; (3) $ub_{AD}+ub_{h}$; (4) $ub_{AD}+ub_{cd}$; (5) $ub_{AD}+ub_{ch}$; (6) $ub_{AD}+ub_{cp}$. Furthermore, the heuristic framework \heur is implemented (Algorithm \ref{alg:Heurframework}) integrating both the degree-based greed method (Algorithm \ref{alg:degbasedHeur}) and colorful degree-based greed method. Additionally, we implement the versions of \maxrfclique equipped with \heur and the aforementioned upper bounds. All algorithms are implemented in C++. We conduct all experiments on a PC with a 2.10GHz Inter Xeon CPU and 256GB memory. We set the time limit to $12$ hours for all algorithms, and use the symbol ``INF'' to denote cases where the algorithm cannot terminate within 12 hours or run out of memory. For reproducibility, the source code of this paper is released on GitHub: {\url{https://github.com/fan2goa1/MaximumFairClique}}.

\begin{figure*}[t!]\vspace*{-0.2cm}
	\begin{center}
		\subfigure[{\themarker (vary $k$)}]{
			\label{fig:pruning-cmp-varyk-nodenum-themarker-log}
			\begin{minipage}{3.2cm}%{2.5cm}
				\centering
				\includegraphics[width=\textwidth]{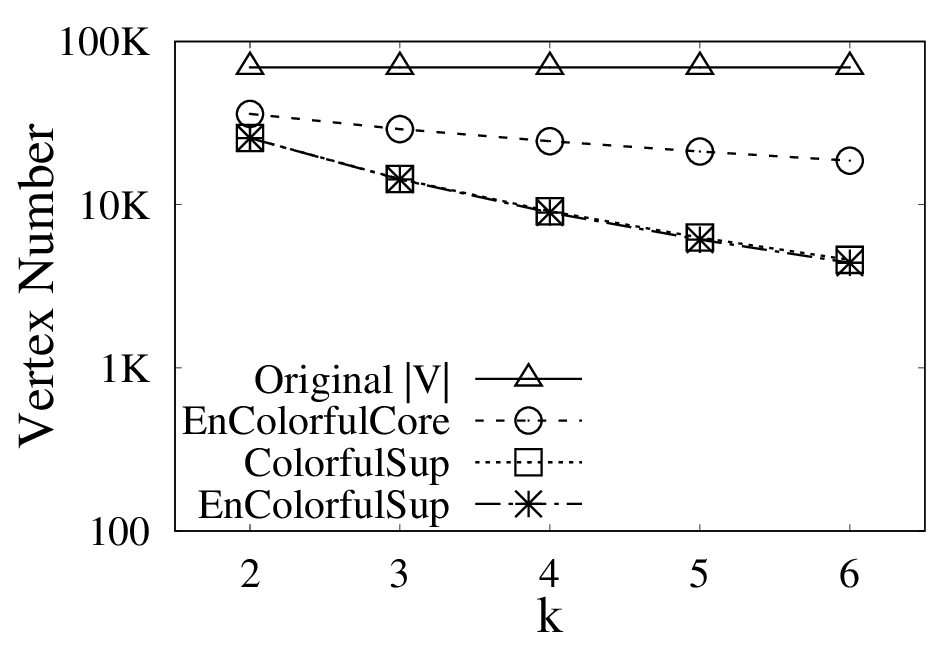}
			\end{minipage}
		}
		\subfigure[{\google (vary $k$)}]{
			\label{fig:pruning-cmp-varyk-nodenum-google-log}
			\begin{minipage}{3.2cm}%{2.5cm}
				\centering
				\includegraphics[width=\textwidth]{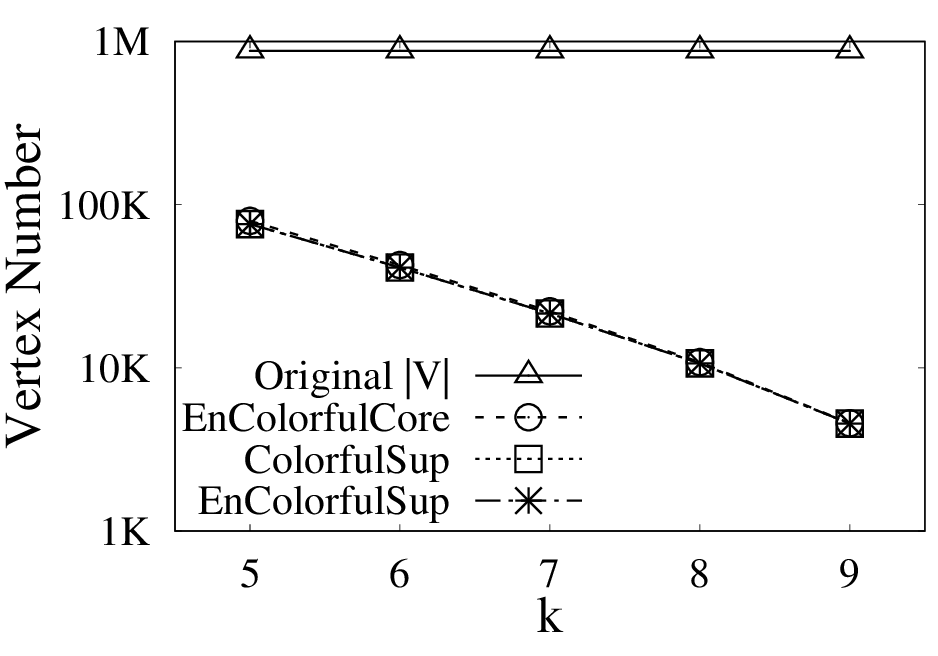}
			\end{minipage}
		}
		\subfigure[{\DBLP (vary $k$)}]{
			\label{fig:pruning-cmp-varyk-nodenum-dblp-log}
			\begin{minipage}{3.2cm}%{2.5cm}
				\centering
				\includegraphics[width=\textwidth]{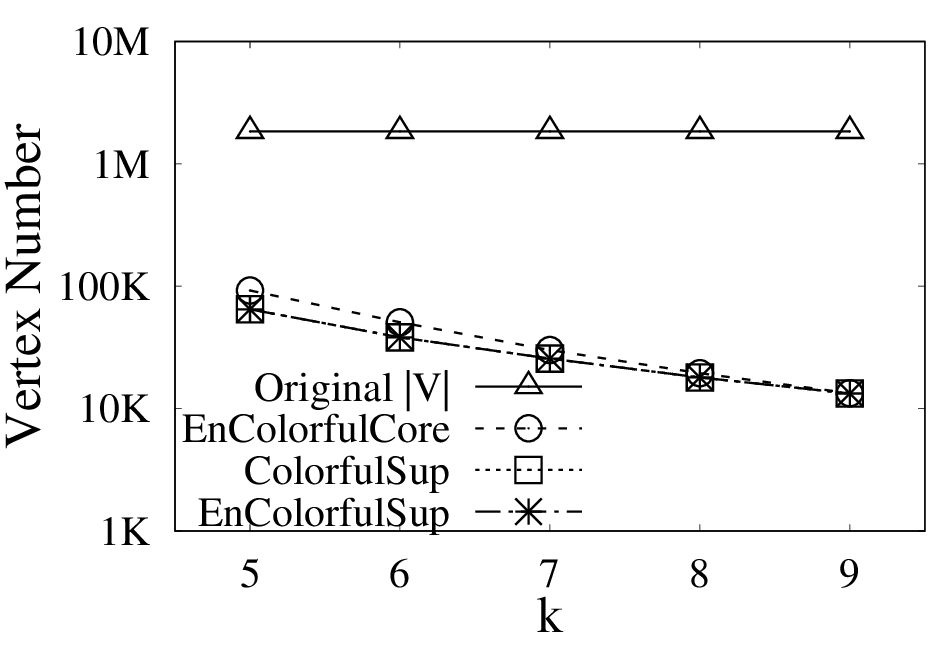}
			\end{minipage}
		}
		\subfigure[{\flixster (vary $k$)}]{
			\label{fig:pruning-cmp-varyk-nodenum-flixster-log}
			\begin{minipage}{3.2cm}%{2.5cm}
				\centering
				\includegraphics[width=\textwidth]{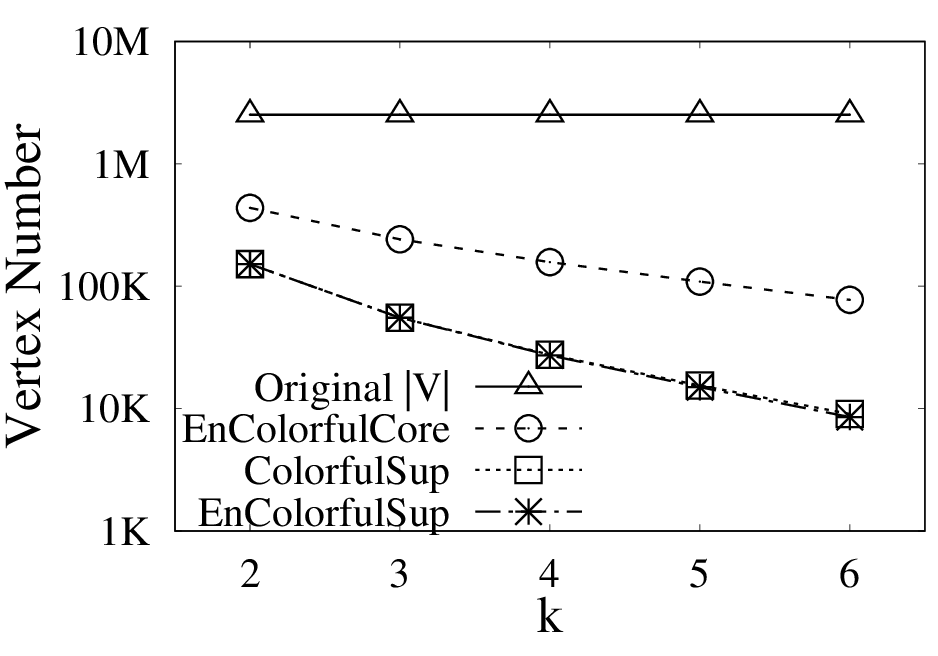}
			\end{minipage}
		}
		\subfigure[{\pokec (vary $k$)}]{
			\label{fig:pruning-cmp-varyk-nodenum-pokec-log}
			\begin{minipage}{3.2cm}%{2.5cm}
				\centering
				\includegraphics[width=\textwidth]{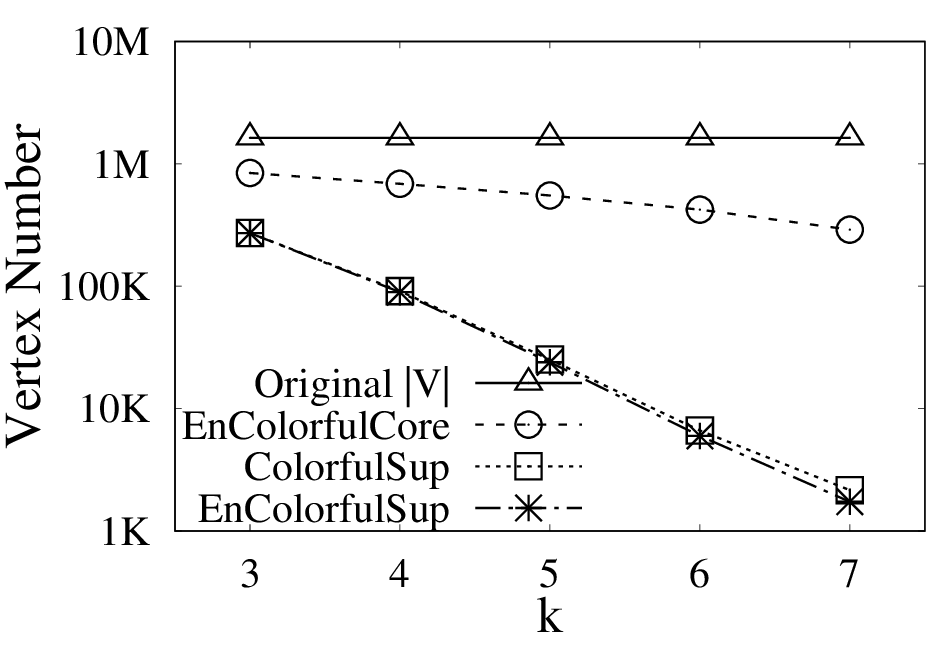}
			\end{minipage}
		}

        \vspace*{-0.4cm}
		\subfigure[{\themarker (vary $k$)}]{
			\label{fig:pruning-cmp-varyk-edgenum-themarker-log}
			\begin{minipage}{3.2cm}%{2.5cm}
				\centering
				\includegraphics[width=\textwidth]{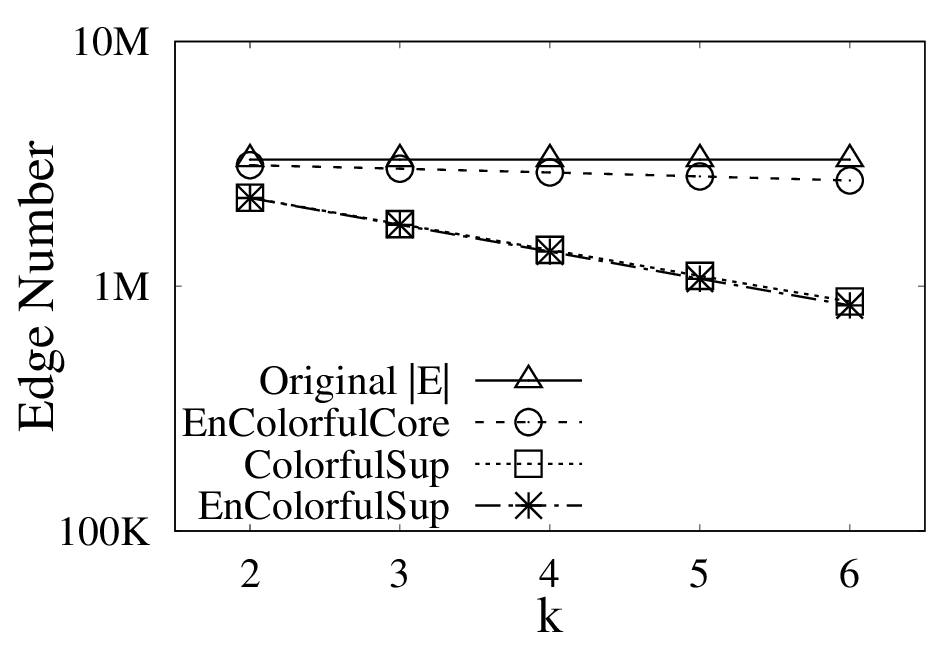}
			\end{minipage}
		}
		\subfigure[{\google (vary $k$)}]{
			\label{fig:pruning-cmp-varyk-edgenum-google-log}
			\begin{minipage}{3.2cm}%{2.5cm}
				\centering
				\includegraphics[width=\textwidth]{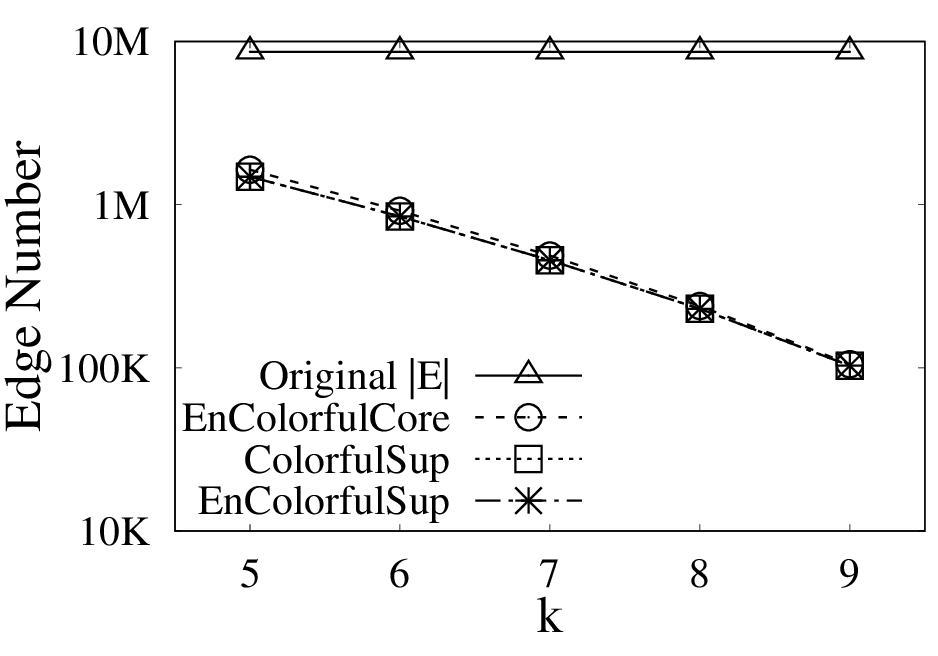}
			\end{minipage}
		}
		\subfigure[{\DBLP (vary $k$)}]{
			\label{fig:pruning-cmp-varyk-edgenum-dblp-log}
			\begin{minipage}{3.2cm}%{2.5cm}
				\centering
				\includegraphics[width=\textwidth]{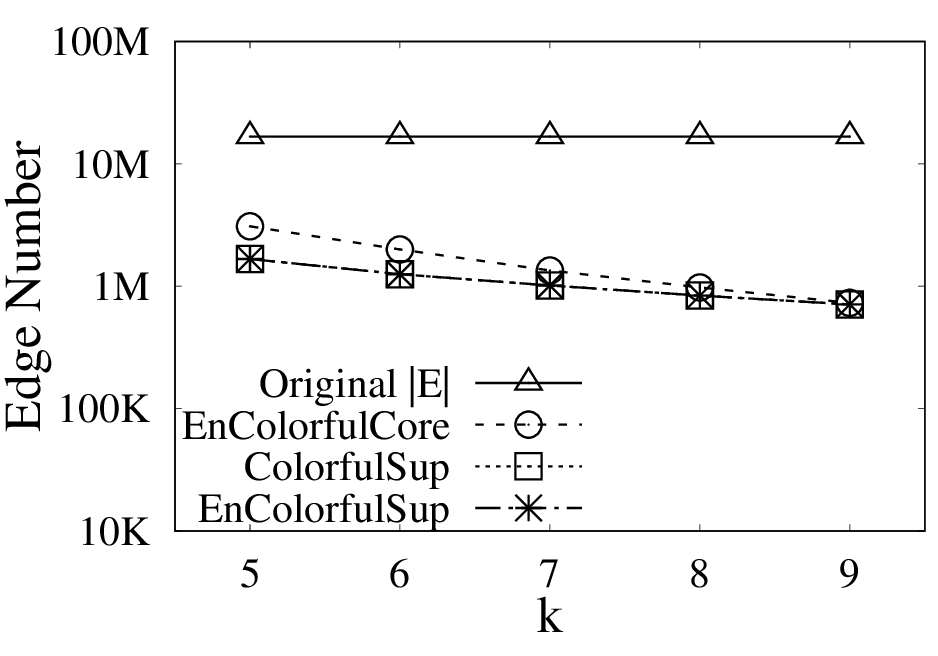}
			\end{minipage}
		}
		\subfigure[{\flixster (vary $k$)}]{
			\label{fig:pruning-cmp-varyk-edgenum-flixster-log}
			\begin{minipage}{3.2cm}%{2.5cm}
				\centering
				\includegraphics[width=\textwidth]{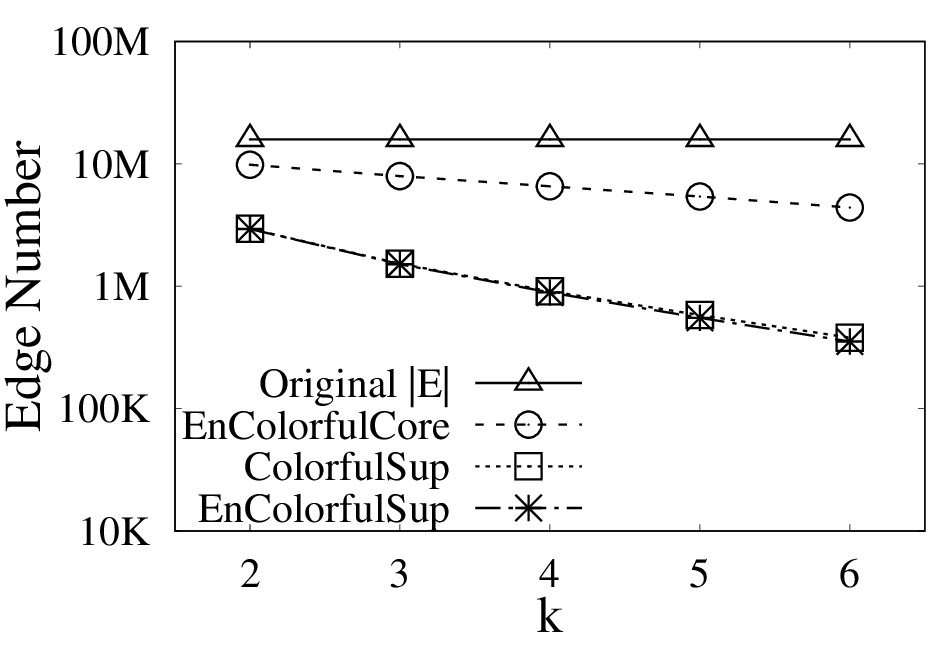}
			\end{minipage}
		}
		\subfigure[{\pokec (vary $k$)}]{
			\label{fig:pruning-cmp-varyk-edgenum-pokec-log}
			\begin{minipage}{3.2cm}%{2.5cm}
				\centering
				\includegraphics[width=\textwidth]{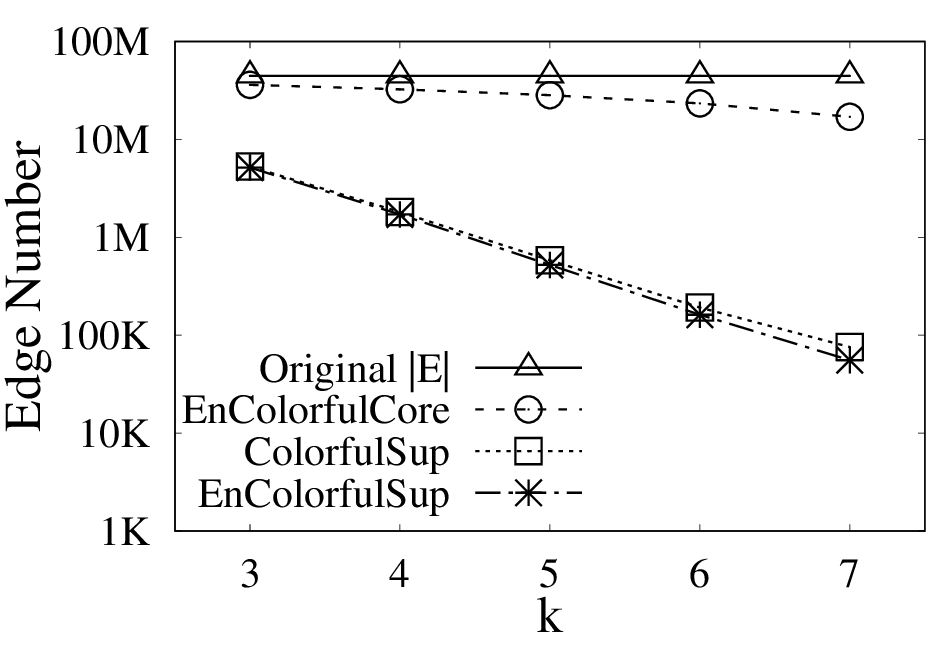}
			\end{minipage}
		}
		
		\comment{
		\subfigure[{\themarker (vary $k$)}]{
			\label{fig:pruning-cmp-varyk-time-themarker-log}
			\begin{minipage}{3.2cm}%{2.5cm}
				\centering
				\includegraphics[width=\textwidth]{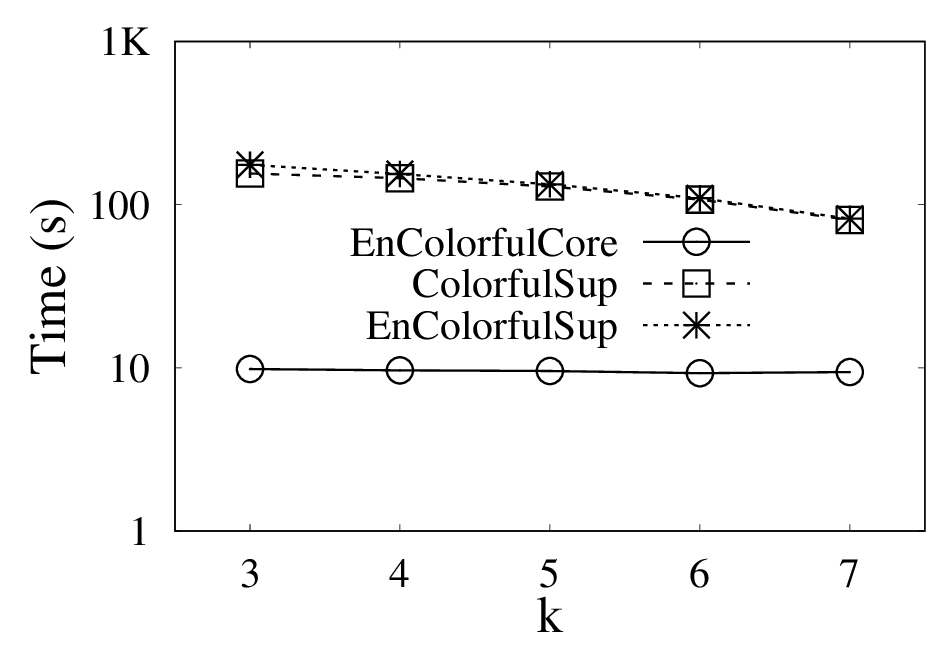}
			\end{minipage}
		}
		\subfigure[{\google (vary $k$)}]{
			\label{fig:pruning-cmp-varyk-time-google-log}
			\begin{minipage}{3.2cm}%{2.5cm}
				\centering
				\includegraphics[width=\textwidth]{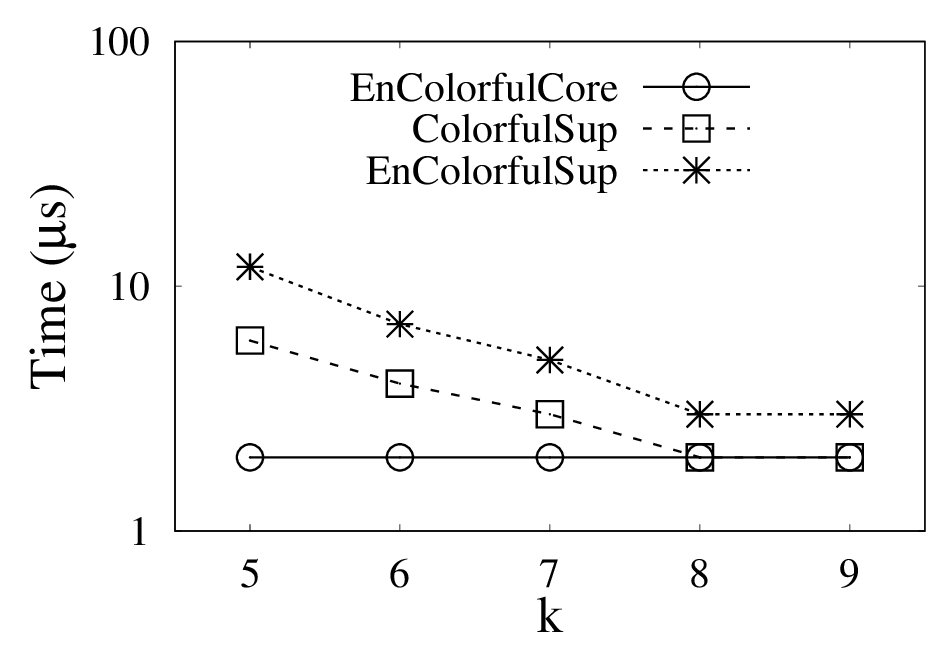}
			\end{minipage}
		}
		\subfigure[{\wikitalk (vary $k$)}]{
			\label{fig:pruning-cmp-varyk-time-wikitalk-log}
			\begin{minipage}{3.2cm}%{2.5cm}
				\centering
				\includegraphics[width=\textwidth]{figures/exp/pruning-cmp-varyk-time-pokec-log.eps}
			\end{minipage}
		}
		\subfigure[{\DBLP (vary $k$)}]{
			\label{fig:pruning-cmp-varyk-time-dblp-log}
			\begin{minipage}{3.2cm}%{2.5cm}
				\centering
				\includegraphics[width=\textwidth]{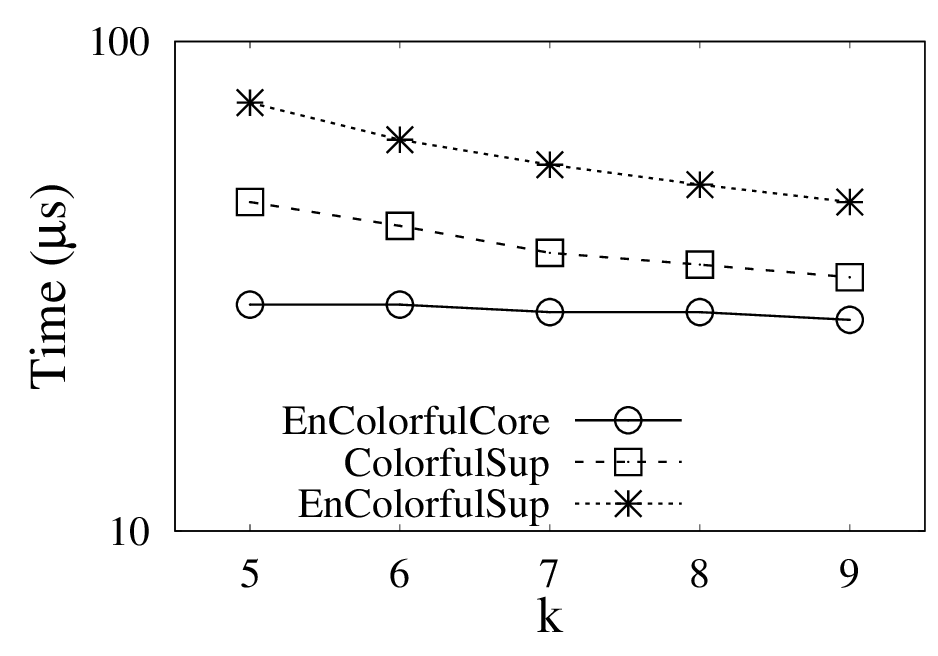}
			\end{minipage}
		}
		\subfigure[{\pokec (vary $k$)}]{
			\label{fig:pruning-cmp-varyk-time-pokec-log}
			\begin{minipage}{3.2cm}%{2.5cm}
				\centering
				\includegraphics[width=\textwidth]{figures/exp/pruning-cmp-varyk-time-pokec-log.eps}
			\end{minipage}
		}
	}	
	\end{center}
	\vspace*{-0.4cm}
	\caption{Comparison of graph reduction techniques: \encolorful, \colorfultruss and \encolorfultruss}
	\vspace*{-0.25cm}
	\label{fig:exp:graphreductions}
\end{figure*}

\stitle{Datasets.} We utilize six real-world graphs to evaluate the efficiency of the proposed algorithms and the dataset statistics are summarized in \tabref{tab:datasets}. Among these datasets, \aminer is an attributed graph where the attribute indicates the gender of scholars, available for download from \url{https://github.com/SotirisTsioutsiouliklis/FairLaR/}. The remaining datasets consist of non-attributed graphs accessible from \url{networkrepository.com/} and \url{snap.stanford.edu}. For these non-attributed graphs, we generate attribute graphs by randomly assigning attributes to vertices with approximately equal probability to evaluate the efficiency of all algorithms.

% \google is a web graph, \aminer and \DBLP are collaboration networks, and the remaining datasets are social networks. 

\stitle{Parameters.} In the maximum fair clique search problem, two parameters, $k$ and $\delta$, require consideration. Due to variations in dataset scales, we adjust the parameter $k$ to different integers for each dataset. Specifically, for \aminer, $k$ is chosen in the range of $[4, 8]$ with a default value of $k=6$. For \themarker, we select $k$ from the interval $[2, 6]$ with a default value of $k=6$. For \google and \DBLP, $k$ ranges between $[5, 9]$, and the default value is $k=7$. For \flixster, we consider $k$ from $[2, 6]$, setting the default value as $k=3$. Regarding \pokec, $k$ varies within $[3, 7]$, with the default set to $k=4$. As for the parameter $\delta$, integer values within the range of $[1, 5]$ are considered, with a default value assigned as $\delta = 4$. In particular, for \themarker and \flixster, we set the default value of $\delta$ to be $3$. During the variation of one parameter, the value of another parameter is maintained at its default setting.

\subsection{Performance studies} \label{sec:results}

\begin{table*}[!htb]\vspace*{-0.2cm}
	\centering
	\caption{Running times of the \maxrfclique algorithms with different upper bounds}
    \vspace*{-0.25cm}
	\scriptsize
	\setlength{\tabcolsep}{1pt}
    \scalebox{0.85}{
	\begin{tabular}{c||c|c|c|c|c|c|c||c|c|c|c|c|c|c}
		\hline
		\hline
		\rule{0pt}{6pt}\multirow{2}{*}{\textbf{Dataset}} &
		\multirow{2}{*}{\textbf{$k$}} &
		\multicolumn{6}{c||}{\rule{0pt}{6pt}\textbf{The \maxrfclique algorithms with different upper bounds ($\upmu$s)}} & \multirow{2}{*}{\textbf{$\delta$}} &  \multicolumn{6}{c}{\rule{0pt}{6pt}\textbf{The \maxrfclique algorithms with different upper bounds ($\upmu$s)}}\\
		\cline{3-8}
		\cline{10-15}
		\rule{0pt}{6pt} & & $ub_{AD}$ &$ub_{AD}+ub_{\triangle}$ &$ub_{AD}+ub_{h}$ &$ub_{AD}+ub_{cd}$ &$ub_{AD}+ub_{ch}$ &$ub_{AD}+ub_{cp}$ & &$ub_{AD}$ &$ub_{AD}+ub_{\triangle}$ &$ub_{AD}+ub_{h}$ &$ub_{AD}+ub_{cd}$ &$ub_{AD}+ub_{ch}$ &$ub_{AD}+ub_{cp}$
		\\ 
		\hline
		\hline
		\rule{0pt}{6pt}\multirow{5}{*}{\textbf{\themarker}} 
		& 2 & 164,020,093	& 164,222,230	& \bf{163,785,612}	& 164,051,886	& 164,191,208	& 164,073,654
		 & 1 & 90,597,328	& 89,731,778	& 91,511,428	& 91,809,020	& 89,826,042	& \bf{89,395,928}
		 \\
		& 3 & 156,447,185	& 156,455,589	& \bf{155,891,523}	& 156,514,092	& 156,114,447	& 156,206,675
		 & 2 & \bf{94,772,436}	&96,119,426	& 95,119,312	& 94,986,264	& 98,534,905	& 95,162,120
		\\
		& 4 & \bf{133,397,225}	& 133,598,283	& 133,501,854	& 133,536,721	& 133,408,072	& 133,555,517
		 & 3 & 95,690,748	& 95,773,560	& 95,812,487	& 95,818,086	& 95,825,326	& \bf{ 95,608,156}
		 \\
		& 5 & 111,368,194	&\bf{ 111,170,802}	& 111,552,467	& 111,195,109	& 111,248,057	& 111,198,219
		 & 4 & 94,292,236	& \bf{94,244,774}	& 97,198,799	& 99,452,775	& 98,857,373	& 101,292,596
		 \\
		& 6 & 95,690,748	& 95,773,560	& 95,812,487	& 95,818,086	& 95,825,326	& \bf{95,608,156} & 5 & 106,183,433	& 104,451,294	& 105,621,450	& 107,220,150	& 103,967,715	& \bf{103,817,481}
		\\
		\hline
		\hline
		\rule{0pt}{6pt}\multirow{5}{*}{\textbf{\google}} 
		& 5 & 13,296,055 &	13,221,049 &	13,219,447 &	\bf{13,197,031} &	13,200,569 &	13,207,587
		 & 1 & 5,615,173	& 5,595,760	& 5,598,803	& 5,596,462	& \bf{5,588,777}	& 5,596,590
		  \\
		& 6 & 8,438,944	& 8,418,007	& 8,400,184	& 8,408,693	& 8,410,402	&\bf{8,399,664} 
		 & 2 & 5,615,501	& \bf{5,592,032}	& 5,596,900	& 5,594,423	& 5,597,983	& 5,597,964
		 \\
		& 7 & 5,608,029	& 5,594,834	& \bf{5,593,969}	& 5,595,033	& 5,599,214	& 5,598,307
		 & 3 & 5,614,339 &	5,597,891	& \bf{5,595,395}	& 5,599,553	& 5,596,092	& 5,595,905
		  \\
		& 8 & 3,963,008	& 3,952,311	& 3,953,677	& 3,951,426	& 3,951,837	& \bf{3,951,291}
		 & 4 & 5,608,029	& 5,594,834	& \bf{5,593,969}	& 5,595,033	& 5,599,214	& 5,598,307
		  \\
		& 9 & 3,112,872	& 3,109,023	& 3,108,917	& \bf{3,108,193}	& 3,108,959	& 3,108,725
		 & 5 & 5,610,155	& 5,596,797	& 5,598,475	& 5,597,520	& 5,595,962	&\bf{5,594,828}
		 \\

		\hline
		\hline
		\rule{0pt}{6pt}\multirow{5}{*}{\textbf{\DBLP}} 
		& 5 & 79,231,788	& 79,242,860	& 79,298,483	& \bf{79,215,098}	& 79,324,955	& 79,278,494
		 & 1 & 57,813,767	& 57,814,992	& 57,820,346	& \bf{57,798,765}	& 57,813,522	& 57,809,781
		  \\
		& 6 & 65,693,405	& 65,693,550	& 65,717,597	& \bf{65,671,062}	& 65,693,812	& 65,691,195
		 & 2 & 57,817,797	& 57,817,311	& 57,821,456	& \bf{57,801,427}	& 57,815,387	& 57,805,805
		 \\
		& 7 & 57,826,719	& 57,838,782	& 57,834,688	& \bf{57,806,109}	& 57,825,486	& 57,830,394
		 & 3 & 57,820,374	& 57,831,164	& 57,835,896	& 57,831,067	& 57,833,149	& \bf{57,816,605}
		  \\
		& 8 & 52,304,894	& 52,316,988	& 52,316,524	& \bf{52,300,741}	& 52,310,454	& 52,304,072
		 & 4 & 57,826,719	& 57,838,782	& 57,834,688	& \bf{57,806,109}	& 57,825,486	& 57,830,394
		  \\
		& 9 & 48,244,249	& \bf{48,224,779}	& 48,231,096	& 48,232,376	& 48,224,881	& 48,239,952
		 & 5 & 57,837,625	& 57,827,627	& 57,840,896	& \bf{57,819,094}	& 57,834,107	& 57,824,004
		 \\
		 \hline
		 \hline
		 \rule{0pt}{6pt}\multirow{5}{*}{\textbf{\flixster}} 
		 & 2 & 116,217,884	& 113,383,872	& \bf{111,973,906}	& 114,089,180	& 113,714,458	& 114,033,281
		 & 1 & 51,498,237	& 51,532,023	& 51,582,030	& \bf{51,486,834}	& 51,531,231	& 51,529,883
		 \\
		 & 3 & 51,747,574	& 51,890,463	& 51,798,280	& \bf{51,559,466}	& 51,740,465	& 51,574,520
		 & 2 & 51,540,428	& 51,613,976	& 51,612,262	& \bf{51,534,845}	& 51,579,454	& 51,621,994
		 \\
		 & 4 & 40,146,859	& 40,173,220	& 40,170,887	& \bf{40,135,284}	& 40,155,296	& 40,163,524
		 & 3 & 51,747,574	& 51,890,463	& 51,798,280	& \bf{51,559,466}	& 51,740,465	& 51,574,520
		 \\
		 & 5 & 33,427,015	& 33,424,487	& 33,438,979	& 33,415,927	& 33,419,604	& \bf{33,413,852}
		 & 4 & \bf{51,651,691}	& 51,821,367	& 51,919,363	&51,658,008	& 51,745,928	& 51,771,211
		 \\
		 & 6 & 28,680,932	& 28,699,177	& 28,688,229	& \bf{28,678,719}	& 28,685,011	& 28,690,077
		 & 5 & \bf{51,530,114}	& 51,589,462	& 51,586,067	& 51,536,750	& 51,782,115	& 51,568,905
		 \\
		\hline
		\hline
		\rule{0pt}{6pt}\multirow{5}{*}{\textbf{\pokec}} 
		& 3 & 383,185,110 &	382,281,393	& 383,224,567 &	392,558,892
			& 380,385,663 &	\bf{379,412,432}
		& 1 & 133,659,904 & 133,659,538 & 133,658,234 & \bf{133,647,298} &	133,655,196	& 133,652,383\\
		& 4 & 179,717,984 &	179,859,519	& 180,397,350 &	179,407,754 &	180,891,601	& \bf{179,011,512}
		 & 2 & 133,653,055 &	133,649,730	& 133,646,114	& \bf{133,640,817} &	133,649,393	& 133,641,592
		\\
		& 5 & 133,645,808 &	133,629,147 &	133,627,799 &	133,626,269	& 133,628,578 &	\bf{133,623,669} & 3 & 133,682,815 &	133,672,986	& 133,673,620	& \bf{133,666,725} &	133,671,610	& 133,671,002
		 \\
		& 6 & 123,720,463 &	123,714,946	& 123,714,386 & 	\bf{123,713,901} &	123,714,644 &	123,716,296
		 & 4 & 133,645,808 &	133,629,147 &	133,627,799 &	133,626,269	& 133,628,578 &	\bf{133,623,669}
		 \\
		& 7 & 96,308,417 &	96,306,924	& 96,307,518	& 96,306,542	& 96,307,212	& \bf{96,306,375}
		 & 5 & 133,638,610	& 133,629,536	& 133,632,658 &	\bf{133,619,445}	 & 133,625,968	& 133,624,136
		\\
		\hline
		\hline
		\rule{0pt}{6pt}\multirow{5}{*}{\textbf{\aminer}} 
		& 4 & 1,740,023	&\bf{1,735,615}	& 1,736,197	& 1,735,819	& 1,735,884	& 1,736,148
		& 1 & 1,399,201	& 1,398,504	& 1,398,299	& \bf{1,398,291}	& 1,398,325	& 1,398,557
		\\
		& 5 & 1,468,401	& 1,466,590	& 1,466,970	& \bf{ 1,466,422}	& 1,466,673	& 1,466,564
		& 2 & 1,399,097	& \bf{1,398,257}	& 1,398,417	& 1,398,430	& 1,398,261	& 1,398,443
		\\
		& 6 & 1,399,833	& 1,398,941	& \bf{1,398,720}	& 1,398,862	& 1,399,046	& 1,398,811
		& 3 & 1,399,251	& 1,399,251	& 1,398,751	& 1,398,520	& \bf{1,398,238}	& 1,398,703	
		\\
		& 7 & 1,282,391	& 1,282,148	& 1,282,173	& \bf{1,281,954}	& 1,281,991	& 1,282,156
		& 4 & 1,399,833	& 1,398,941	& \bf{1,398,720}	& 1,398,862	& 1,399,046	& 1,398,811
		\\
		& 8 & 1,251,681	& 1,251,316	& \bf{1,251,297}	& 1,251,460	& 1,251,342	& 1,251,482
		& 5 & 1,399,996	& \bf{1,399,182}	& 1,399,489	& 1,399,286	& 1,399,864	& 1,399,603
		\\
		\hline
		\hline
		\comment{
		\rule{0pt}{6pt}\multirow{5}{*}{\textbf{\linkedin}} 
		& 4 & 383,185,110 &	382,281,393	& 383,224,567 &	392,558,892
		& 380,385,663 &	\bf{379,412,432}
		& 1 & 133,659,904 & 133,659,538 & 133,658,234 & \bf{133,647,298} &	133,655,196	& 133,652,383\\
		& 5 & 179,717,984 &	179,859,519	& 180,397,350 &	179,407,754 &	180,891,601	& \bf{179,011,512}
		& 2 & 133,653,055 &	133,649,730	& 133,646,114	& \bf{133,640,817} &	133,649,393	& 133,641,592
		\\
		& 6 & 133,645,808 &	133,629,147 &	133,627,799 &	133,626,269	& 133,628,578 &	\bf{133,623,669} & 3 & 133,682,815 &	133,672,986	& 133,673,620	& \bf{133,666,725} &	133,671,610	& 133,671,002
		\\
		& 7 & 123,720,463 &	123,714,946	& 123,714,386 & 	\bf{123,713,901} &	123,714,644 &	123,716,296
		& 4 & 133,645,808 &	133,629,147 &	133,627,799 &	133,626,269	& 133,628,578 &	\bf{133,623,669}
		\\
		& 8 & 96,308,417 &	96,306,924	& 96,307,518	& 96,306,542	& 96,307,212	& \bf{96,306,375}
		& 5 & 133,638,610	& 133,629,536	& 133,632,658 &	\bf{133,619,445}	 & 133,625,968	& 133,624,136
		\\}
		\hline	
	\end{tabular}
    }
	\label{tab:exp:upperbounds}
\end{table*}

\stitle{Evaluation of the graph reduction techniques.} In this experiment, we evaluate the graph reduction techniques, namely, \encolorful, \colorfultruss, and \encolorfultruss, by varying the value of $k$. The counts of remaining vertices and edges on datasets with generated attributes are depicted in \figref{fig:exp:graphreductions}. Notably, as the value of $k$ increases, the number of vertices and edges left in the graph decreases across all reduction techniques. This is because, with larger values of $k$, the requirements for the enhanced colorful degree (resp., colorful support, enhanced colorful support) of vertices (resp., edges) within fair cliques become more rigorous. Consequently, only a few vertices and edges are able to fulfill these stringent requirements. Moreover, with a fixed $k$, \encolorful, \colorfultruss and \encolorfultruss significantly reduce the number of vertices and edges compared to the initial graph. Both \colorfultruss and \encolorfultruss exhibit more robust graph reduction capabilities compared to \encolorful, and \encolorfultruss outperforms \colorfultruss. This is owing to the fact that \colorfultruss builds upon \encolorful by incorporating a constraint on the number of common neighbors with a specific attribute at the endpoints of an edge, i.e., the constraint on the colorful support of an edge. \encolorfultruss further extends \colorfultruss by assigning colors to specific attributes, imposing more stringent conditions on edges, and resulting in a more pronounced reduction in nodes and edges. For example, on the \pokec dataset with $k=7$, sequentially applying \encolorful, \colorfultruss and \encolorfultruss leaves 290,258, 2,155, and 1,735 vertices, with remaining edges numbering 17,004,374, 75,652, and 55,536, respectively. In contrast, the original graph contains 1,632,803 vertices and 44,603,928 edges. Additionally, we evaluate the performance of these three reductions using the \aminer dataset with real attributes, and the results shown in \figref{fig:exp:graphreductionsaminer} align consistently with the previous findings.

\begin{figure}[t!]\vspace*{-0.2cm}
	\begin{center}		
		\subfigure[{\aminer (vary $k$)}]{
			\label{fig:pruning-cmp-varyk-nodenum-aminer-log}
			\begin{minipage}{3.2cm}%{2.5cm}
				\centering
				\includegraphics[width=\textwidth]{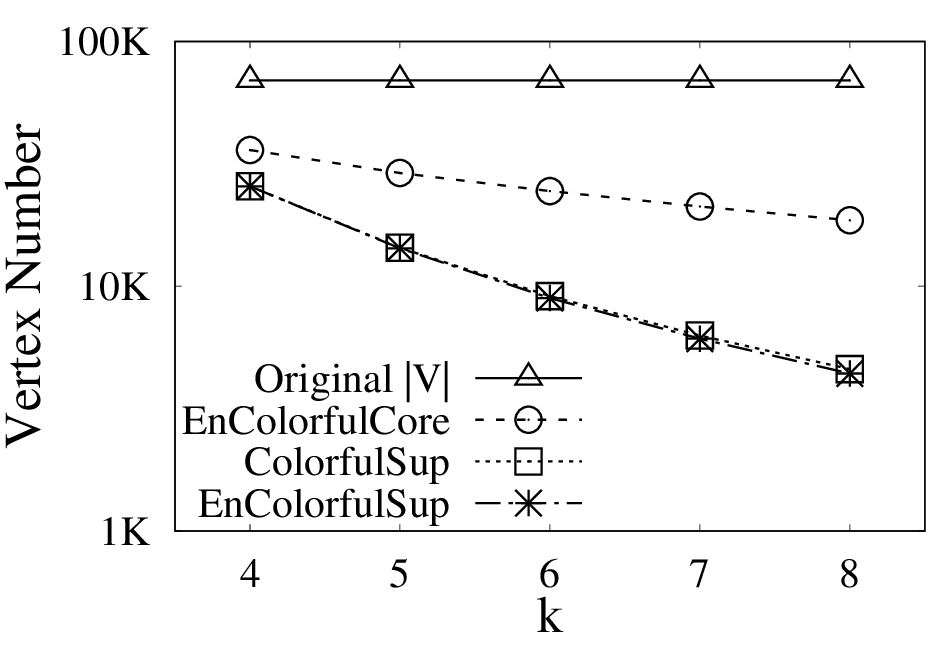}
			\end{minipage}
		}
		\subfigure[{\aminer (vary $k$)}]{
			\label{fig:pruning-cmp-varyk-edgenum-aminer-log}
			\begin{minipage}{3.2cm}%{2.5cm}
				\centering
				\includegraphics[width=\textwidth]{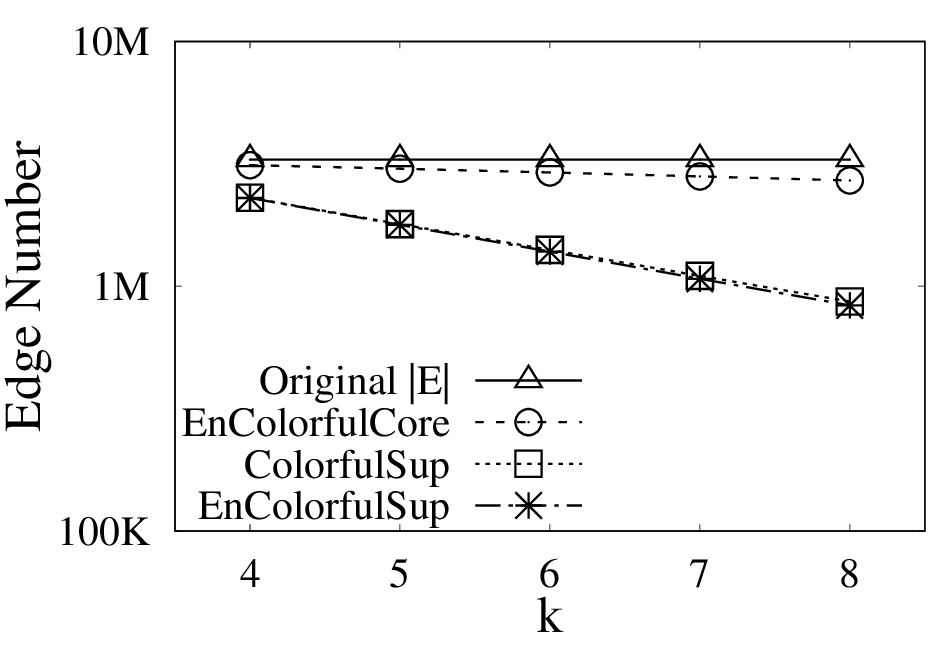}
			\end{minipage}
		}	
	\end{center}
	\vspace*{-0.4cm}
	\caption{Comparison of graph reduction techniques on \aminer}
	\vspace*{-0.25cm}
	\label{fig:exp:graphreductionsaminer}
\end{figure}

\stitle{Evaluation of different upper bounds.} We evaluate the runtime of the \maxrfclique algorithms equipped with different upper bounds with varying $k$ and $\delta$. These upper bounding pruning techniques are applied in \maxrfclique when selecting vertices to be added to $R$ for the first time. The running times of \maxrfclique using various upper bounds are presented in \tabref{tab:exp:upperbounds}, with the minimum time highlighted. It can be observed that diverse datasets exhibit distinct characteristics, resulting in varying optimal upper bounds. Notably, the colorful-degeneracy-based upper bound and colorful-path-based upper bound achieve superior performance across a broader range of experimental settings. Although the running times of \maxrfclique with different upper bounds do not exhibit considerable differences within the same dataset, employing these upper bounds in \maxrfclique significantly reduces the runtime for the maximum fair clique search, as demonstrated in the subsequent experiments.

\begin{figure*}[t!]\vspace*{-0.2cm}
	\begin{center}
		%\hspace*{-0.5cm}
		\subfigure[{\themarker (vary $k$)}]{
			\label{fig:search-cmp-varyk-themarker-log}
			\begin{minipage}{3.2cm}%{2.5cm}
				\centering
				\includegraphics[width=\textwidth]{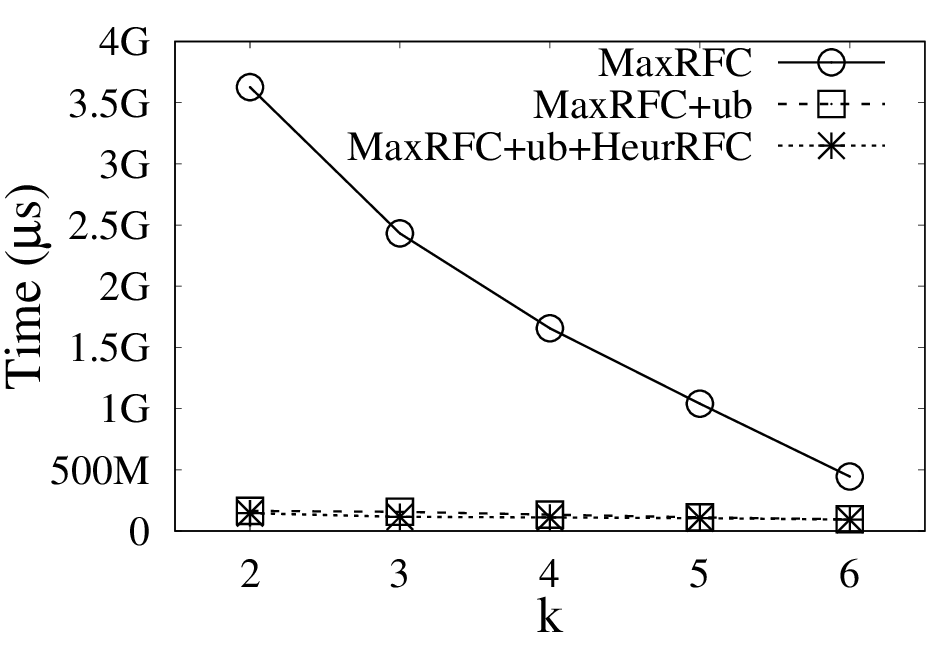}
			\end{minipage}
		}
		\subfigure[{\google (vary $k$)}]{
			\label{fig:search-cmp-varyk-google-log}
			\begin{minipage}{3.2cm}%{2.5cm}
				\centering
				\includegraphics[width=\textwidth]{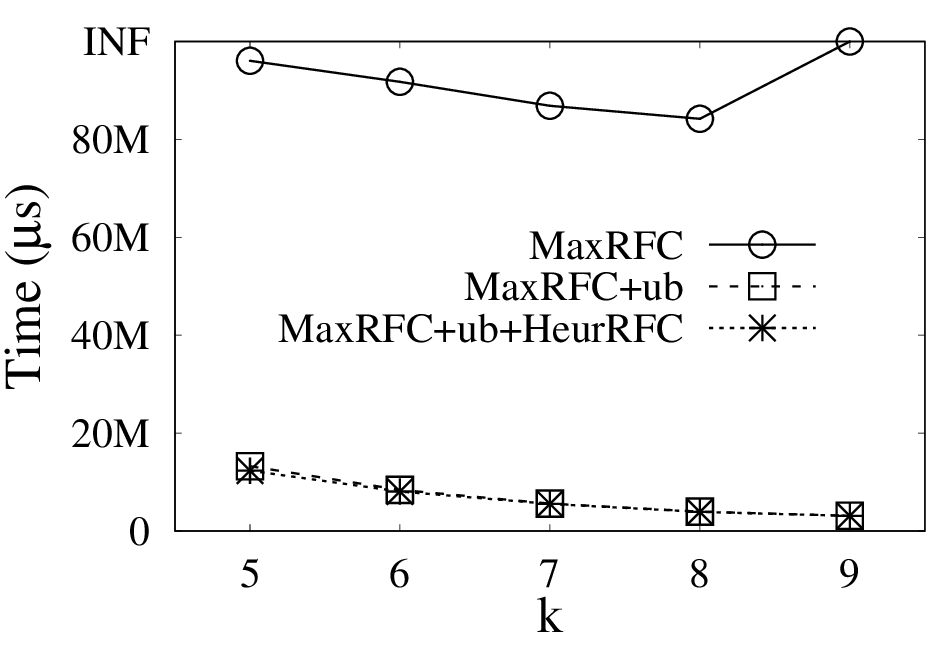}
			\end{minipage}
		}
		\subfigure[{\DBLP (vary $k$)}]{
			\label{fig:search-cmp-varyk-dblp}
			\begin{minipage}{3.2cm}%{2.5cm}
				\centering
				\includegraphics[width=\textwidth]{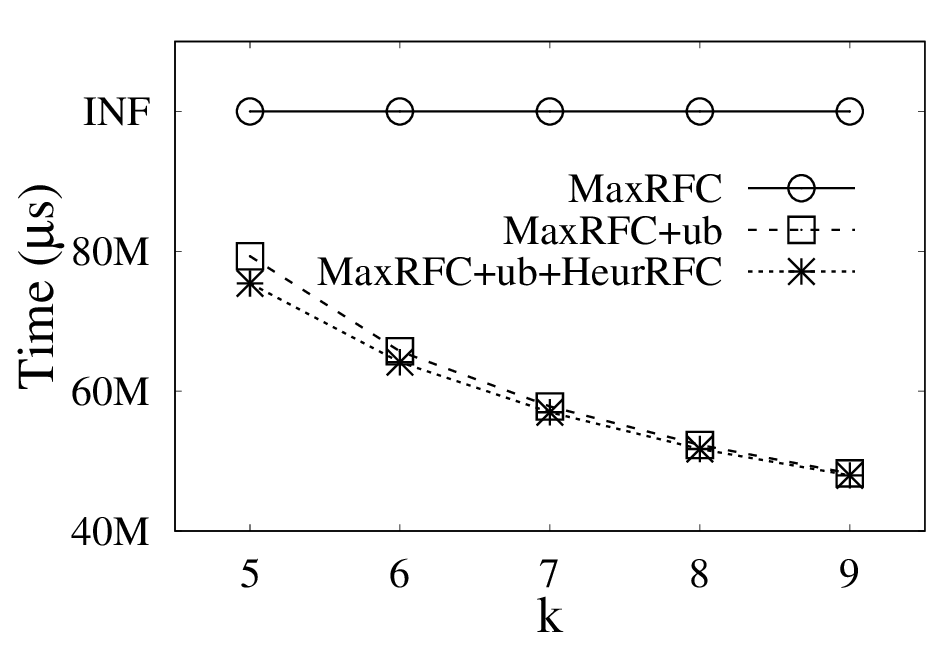}
			\end{minipage}
		}
		\subfigure[{\flixster (vary $k$)}]{
			\label{fig:search-cmp-varyk-flixster}
			\begin{minipage}{3.2cm}%{2.5cm}
				\centering
				\includegraphics[width=\textwidth]{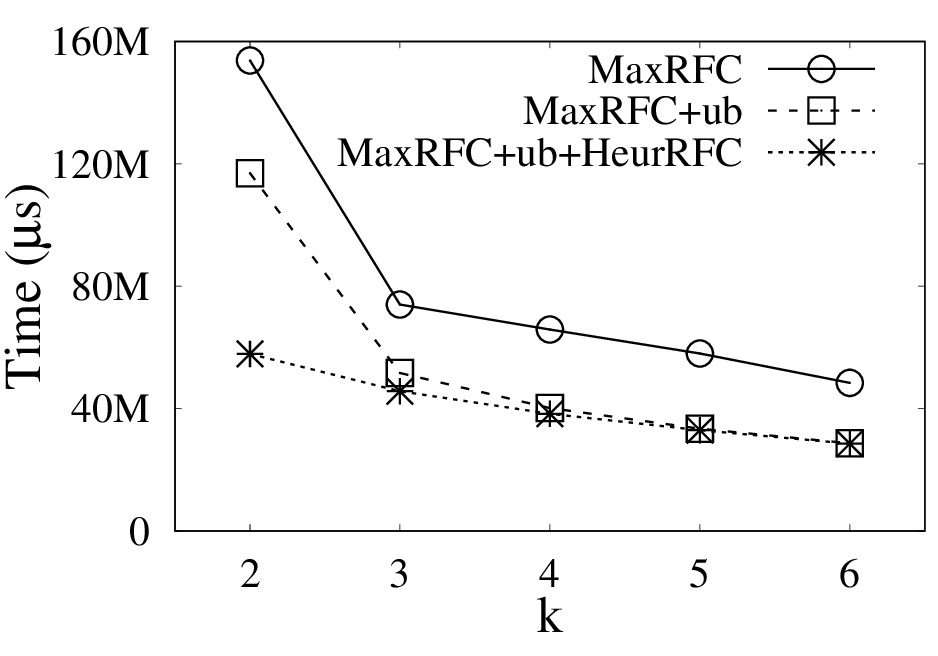}
			\end{minipage}
		}
		\subfigure[{\pokec (vary $k$)}]{
			\label{fig:search-cmp-varyk-time-pokec}
			\begin{minipage}{3.2cm}%{2.5cm}
				\centering
				\includegraphics[width=\textwidth]{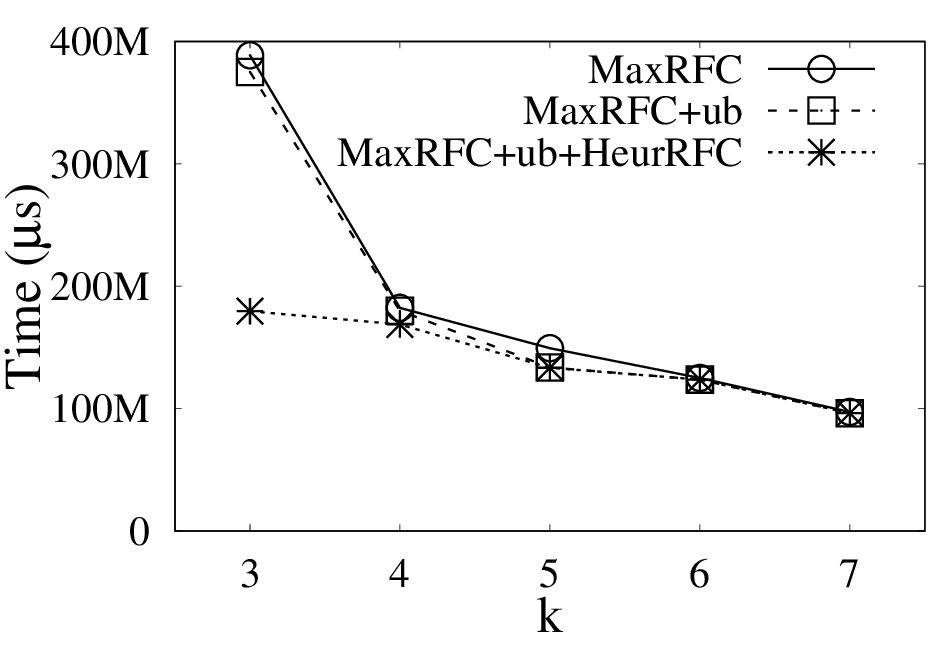}
			\end{minipage}
		}

        \vspace*{-0.4cm}
		\subfigure[{\themarker (vary $\delta$)}]{
			\label{fig:search-cmp-varydelta-themarker-log}
			\begin{minipage}{3.2cm}%{2.5cm}
				\centering
				\includegraphics[width=\textwidth]{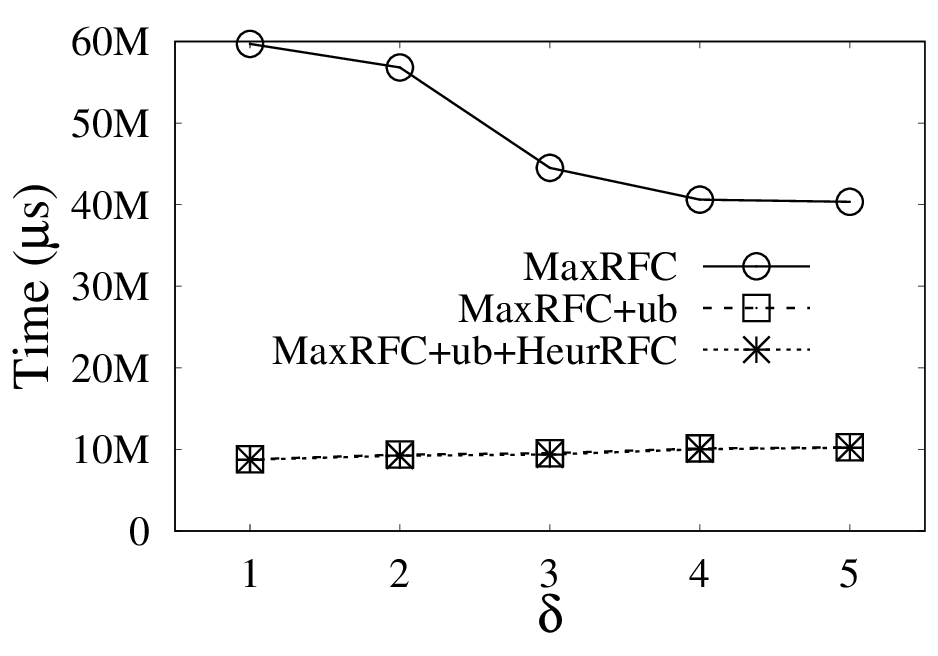}
			\end{minipage}
		}
		\subfigure[{\google (vary $\delta$)}]{
			\label{fig:search-cmp-varydelta-google-log}
			\begin{minipage}{3.2cm}%{2.5cm}
				\centering
				\includegraphics[width=\textwidth]{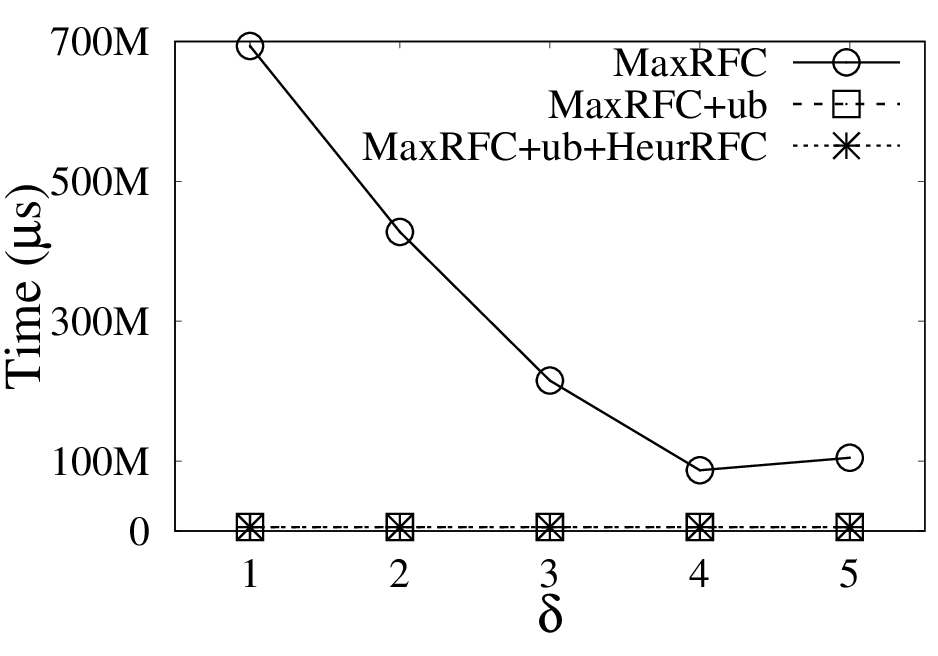}
			\end{minipage}
		}
		\subfigure[{\DBLP (vary $\delta$)}]{
			\label{fig:search-cmp-varydelta-dblp}
			\begin{minipage}{3.2cm}%{2.5cm}
				\centering
				\includegraphics[width=\textwidth]{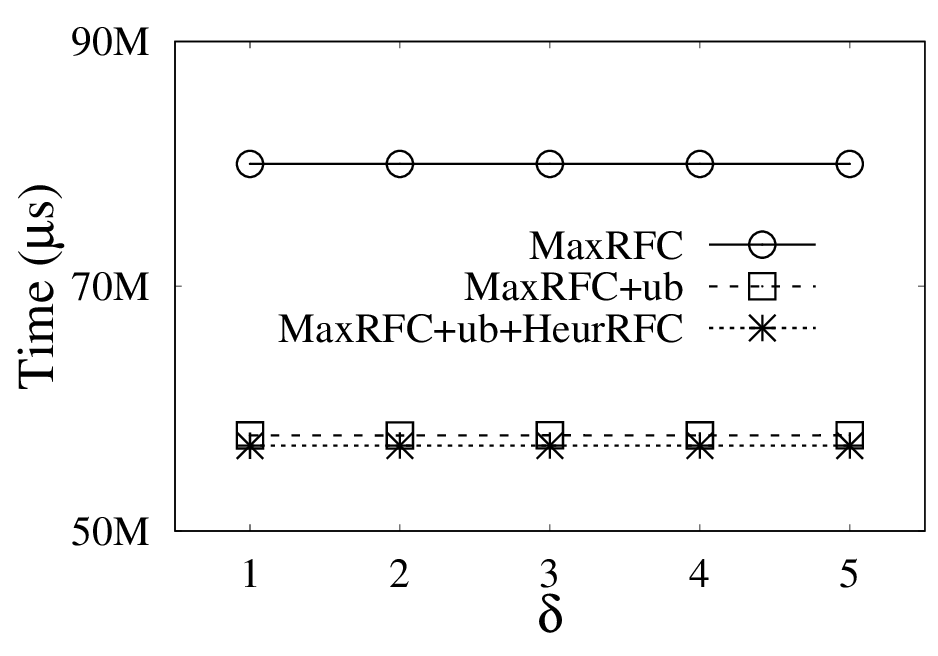}
			\end{minipage}
		}
		\subfigure[{\flixster (vary $\delta$)}]{
			\label{fig:search-cmp-varydelta-flixster}
			\begin{minipage}{3.2cm}%{2.5cm}
				\centering
				\includegraphics[width=\textwidth]{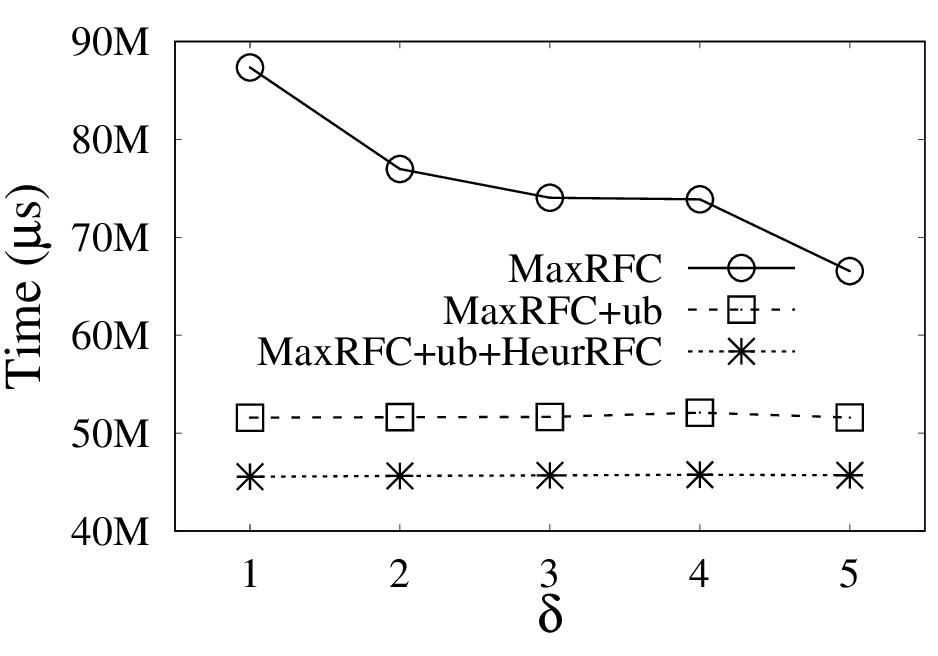}
			\end{minipage}
		}
		\subfigure[{\pokec (vary $\delta$)}]{
			\label{fig:search-cmp-varydelta-time-pokec}
			\begin{minipage}{3.2cm}%{2.5cm}
				\centering
				\includegraphics[width=\textwidth]{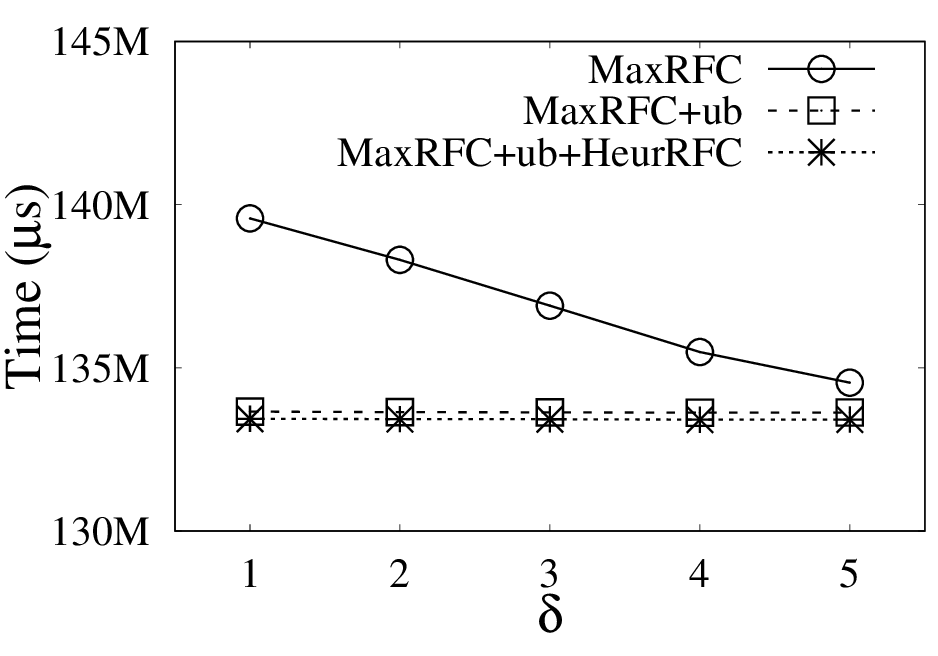}
			\end{minipage}
		}
	\end{center}
	\vspace*{-0.4cm}
	\caption{Comparison of the \maxrfclique algorithms}
	\vspace*{-0.25cm}
	\label{fig:exp:searchalgorithms}
\end{figure*}

\stitle{Evaluation of the maximum fair clique search algorithms.} We establish \maxrfclique as the baseline and conduct a comparative analysis against two variations: \maxrfclique with upper bounding technique, and \maxrfclique with both upper bounding technique and \heur. For each dataset, we select the optimal upper bound from \tabref{tab:exp:upperbounds} to apply as the upper bound in \maxrfclique. Specifically, for \themarker, \google and \pokec, \maxrfclique uses ``$ub_{AD}+ub_{cp}$'' as the upper bound, while for the other datasets, it employs ``$ub_{AD}+ub_{cd}$'' as the upper bound. The runtime of \maxrfclique, {{\kw{MaxRFC}}+$ub$}, and {{\kw{MaxRFC}}+$ub$+{\kw{HeurRFC}}} for finding the maximum fair clique is shown in \figref{fig:exp:searchalgorithms} and \figref{fig:exp:searchalgorithmsaminer}. Note that in \figref{fig:search-cmp-varyk-google-log}, ``INF'' indicates ``Out of memory'', while in \figref{fig:search-cmp-varyk-dblp} and \figref{fig:exp:searchalgorithmsaminer}, ``INF'' represents that the algorithm exceeds the predefined time limit. As can be seen, the running time of \maxrfclique, {{\kw{MaxRFC}}+$ub$}, and {{\kw{MaxRFC}}+$ub$+{\kw{HeurRFC}}} tends to decrease with increasing $k$ due to fewer cliques satisfying fair clique constraints, expediting the identification of the maximum fair clique. Changes in $\delta$ do not exhibit a consistent trend in the runtime of these algorithms; rather, this seems to be influenced by the characteristics of the specific dataset. Notably, both {{\kw{MaxRFC}}+$ub$} and {{\kw{MaxRFC}}+$ub$+{\kw{HeurRFC}}} exhibit significantly faster execution times compared to \maxrfclique. This performance enhancement can be credited to the use of the upper-bound-based pruning techniques and the integration of the heuristic-result-based pruning. Concerning the {{\kw{MaxRFC}}+$ub$+{\kw{HeurRFC}}} algorithm, although its runtime is marginally lower than that of {{\kw{MaxRFC}}+$ub$}, these results suggest the contribution of \heur to the efficiency of the maximum fair clique search process. For instance, on the \flixster, when $k=2$, {{\kw{MaxRFC}}+$ub$} and {{\kw{MaxRFC}}+$ub$+{\kw{HeurRFC}}} run approximately 15 and 20 times faster than \maxrfclique, respectively. These results underscore the efficiency of the proposed upper bound pruning techniques and the heuristic algorithm.

\begin{figure}[t!]\vspace*{-0.2cm}
	\begin{center}		
		\subfigure[{\aminer (vary $k$)}]{
			\label{fig:search-cmp-varyk-time-aminer}
			\begin{minipage}{3.2cm}%{2.5cm}
				\centering
				\includegraphics[width=\textwidth]{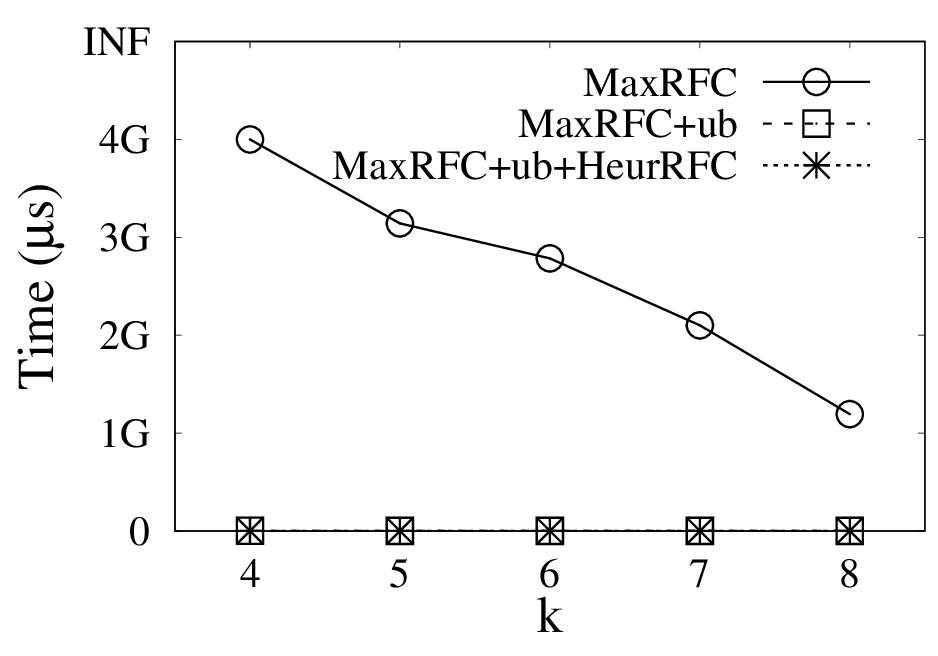}
			\end{minipage}
		}
		\subfigure[{\aminer (vary $\delta$)}]{
			\label{fig:search-cmp-varydelta-time-aminer}
			\begin{minipage}{3.2cm}%{2.5cm}
				\centering
				\includegraphics[width=\textwidth]{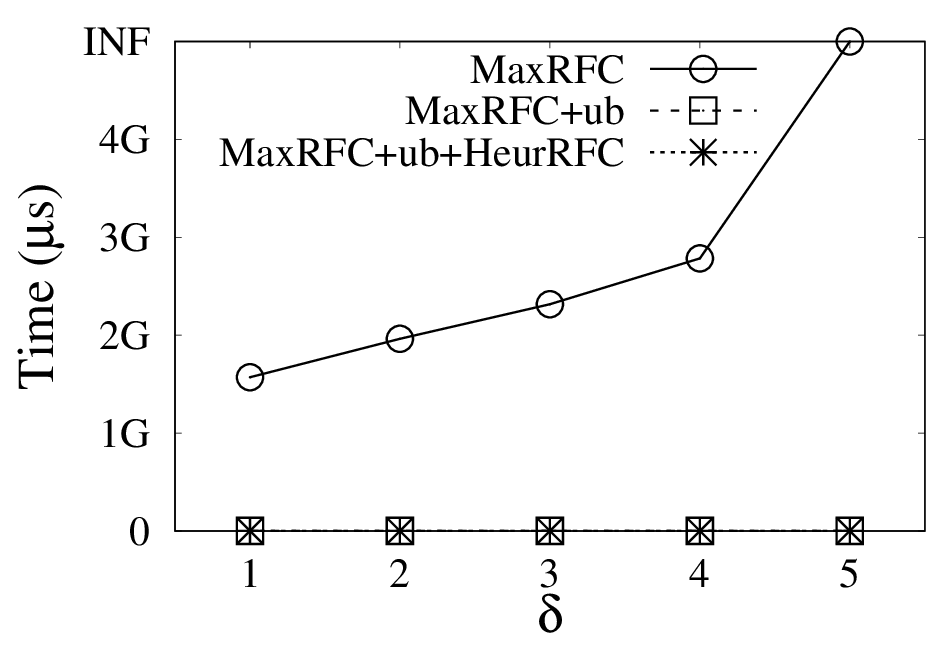}
			\end{minipage}
		}	
	\end{center}
	\vspace*{-0.4cm}
	\caption{Comparison of the \maxrfclique algorithms on \aminer}
	\vspace*{-0.25cm}
	\label{fig:exp:searchalgorithmsaminer}
\end{figure}

\stitle{The effectiveness of the heuristic algorithm.} We evaluate the effectiveness of \heur by comparing the size of the fair clique it finds with the size of the maximum fair clique. The results are depicted in \figref{fig:exp:sizecmp}. Clearly, across most datasets, the fair clique discovered by \heur is very close in size to the maximum fair clique, with differences of no more than 6. Notably, on \DBLP, the \heur algorithm outputs a fair clique of the same size as the maximum fair clique. These results demonstrate that our \heur algorithm can indeed yield a fair clique of larger size within linear time, making it a valuable tool for pruning the search space in \maxrfclique.

\begin{figure}[t]\vspace*{-0.2cm}
	\centering
 \includegraphics[width=0.35\textwidth]{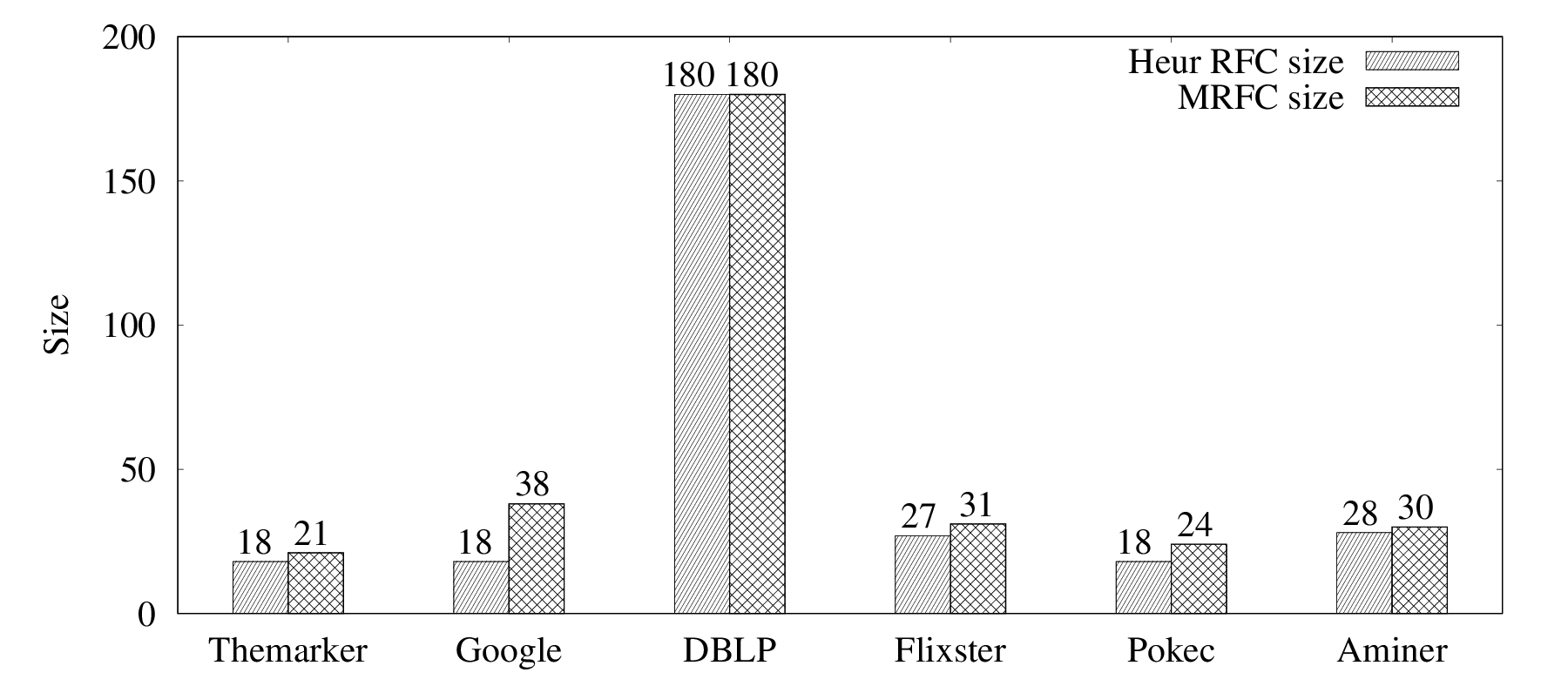}	
	\vspace*{-0.4cm}
	\caption{The sizes of fair cliques found by \maxrfclique and \heur}
	\vspace*{-0.2cm}
	\label{fig:exp:sizecmp}
\end{figure}

\stitle{Scalability testing.} We create four subgraphs for each dataset by randomly selecting 20\%-80\% of vertices and edges to evaluate the scalability of the maximum fair clique search algorithms. The results on \flixster are presented in \figref{fig:exp:scalability}. Similar outcomes are expected for the other datasets, though they are not shown here due to space limits. As can be seen, \maxrfclique exhibits a steep rise in running time with increasing $m$ or $n$, whereas {{\kw{MaxRFC}}+$ub$} and {{\kw{MaxRFC}}+$ub$+{\kw{HeurRFC}}} show a more gradual increase. Again, the runtime of \maxrfclique is notably longer compared to {{\kw{MaxRFC}}+$ub$} and {{\kw{MaxRFC}}+$ub$+{\kw{HeurRFC}}}. These results confirm the superior scalability of the {{\kw{MaxRFC}}+$ub$} and {{\kw{MaxRFC}}+$ub$+{\kw{HeurRFC}}} algorithms in handling large-scale graphs.

\begin{figure}[t!]\vspace*{-0.2cm}
	\begin{center}		
		\subfigure[{\flixster (vary $m$)}]{
			\label{fig:scala-varym-flixster}
			\begin{minipage}{3.2cm}%{2.5cm}
				\centering
				\includegraphics[width=\textwidth]{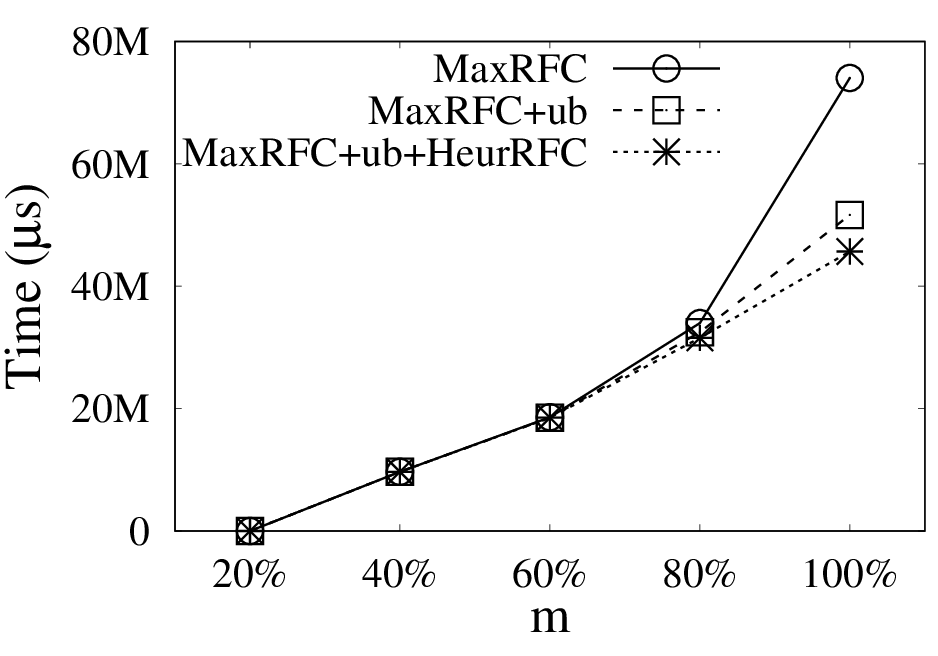}
			\end{minipage}
		}
		\subfigure[{\flixster (vary $n$)}]{
			\label{fig:scala-varyn-flixster}
			\begin{minipage}{3.2cm}%{2.5cm}
				\centering
				\includegraphics[width=\textwidth]{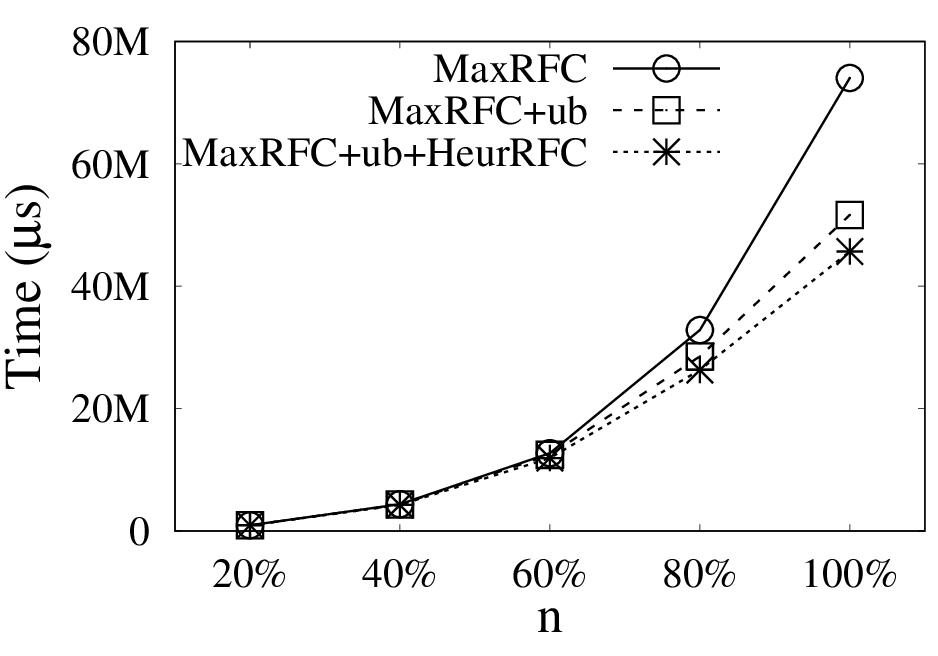}
			\end{minipage}
		}	
	\end{center}
	\vspace*{-0.5cm}
	\caption{Scalability test}
	\vspace*{-0.25cm}
	\label{fig:exp:scalability}
\end{figure}

\subsection{Case study} \label{sec:casestudy}

\stitle{Case study on \aminer.} We conduct a case study on \aminer to evaluate the effectiveness of our algorithms. The attribute $A$ in \aminer indicates the gender of the author, i.e., $A=\{male, female\}$. With $k=5$ and $\delta=3$, we invoke the proposed algorithms to find the maximum fair clique. \figref{fig:caseaminer} shows the result with $13$ males (colored blue) and $16$ females (colored red). It maintains a balance, ensuring the count of males and females is not less than $k$, with a difference between them not exceeding $\delta$. The scholars in \figref{fig:caseaminer} primarily affiliate with two establishments: the smart HCI lab of the ICxT Innovation center at the University of Turin and Telecom Italy Company. Their focus areas span human-computer interaction, information visualization, and multimodal interaction. Notably, five scholars boast a Google Scholar impact exceeding 2,000. Further validation through the HCI Lab's official website confirms a longstanding partnership with Telecom Italy, involving collaborative projects like Personalised Television Services, E-Tourism-Context-Aware Systems, and ICT Converging Technologies 2008-PIEMONTE, among others. These findings underscore the effectiveness of our algorithms in identifying large, well-connected teams renowned in the field of human-computer interaction. Within these collectives, scholars of diverse genders leverage their individual expertise, culminating in a robust and adept collaborative force.

\comment{
\begin{figure}[t]\vspace*{-0.2cm}
	\begin{center}		
        \subfigure[\aminer]{
			\label{fig:caseaminer}
			\begin{minipage}[b]{5cm}%{2.5cm}
				\centering
				\includegraphics[width=\textwidth]{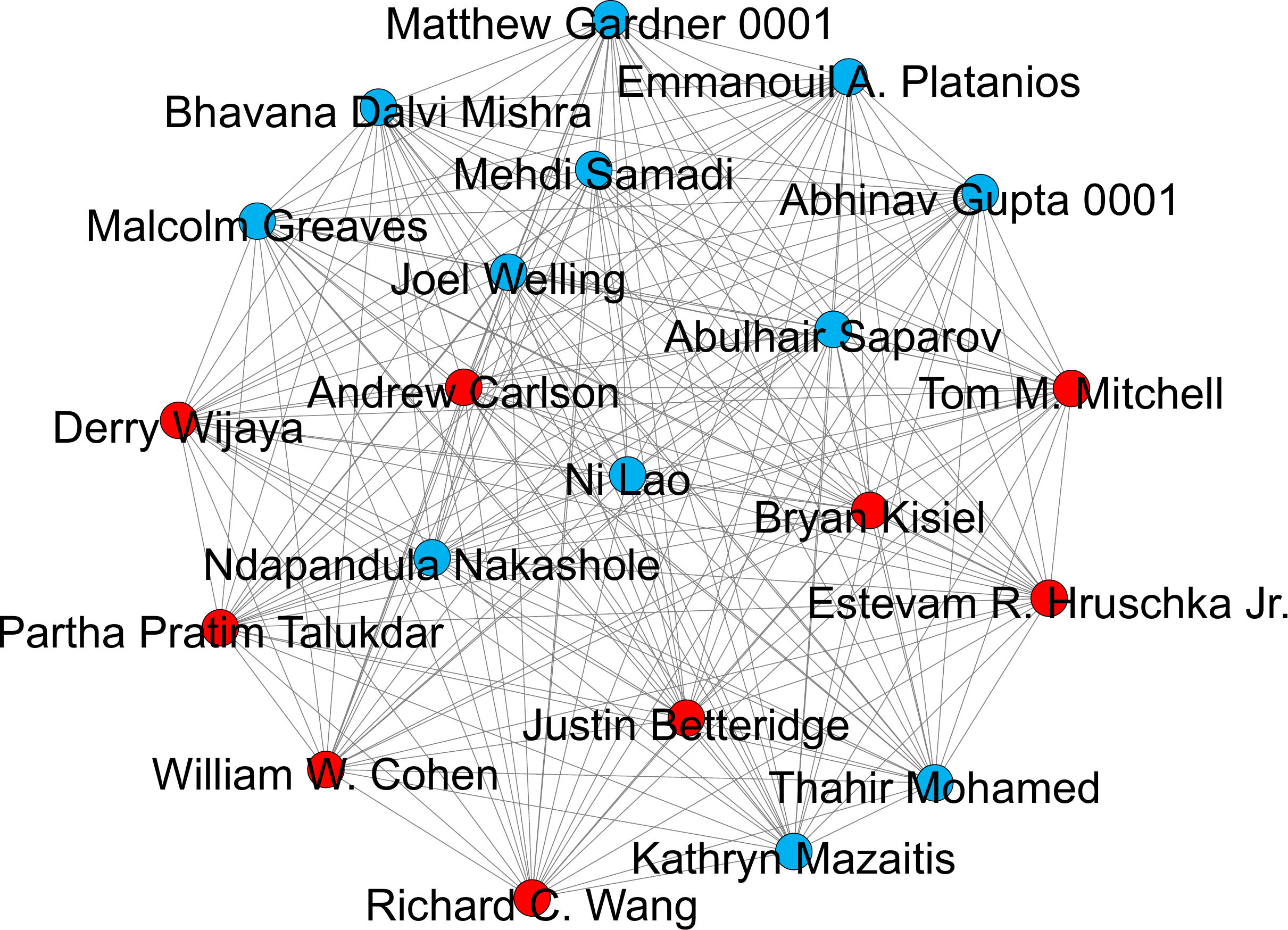}
			\end{minipage}
		}
    \hspace{-0.5cm}
		\subfigure[\dbai]{
			\label{fig:casedbai}
			\begin{minipage}[b]{5cm}%{2.5cm}
				\centering
				\includegraphics[width=\textwidth]{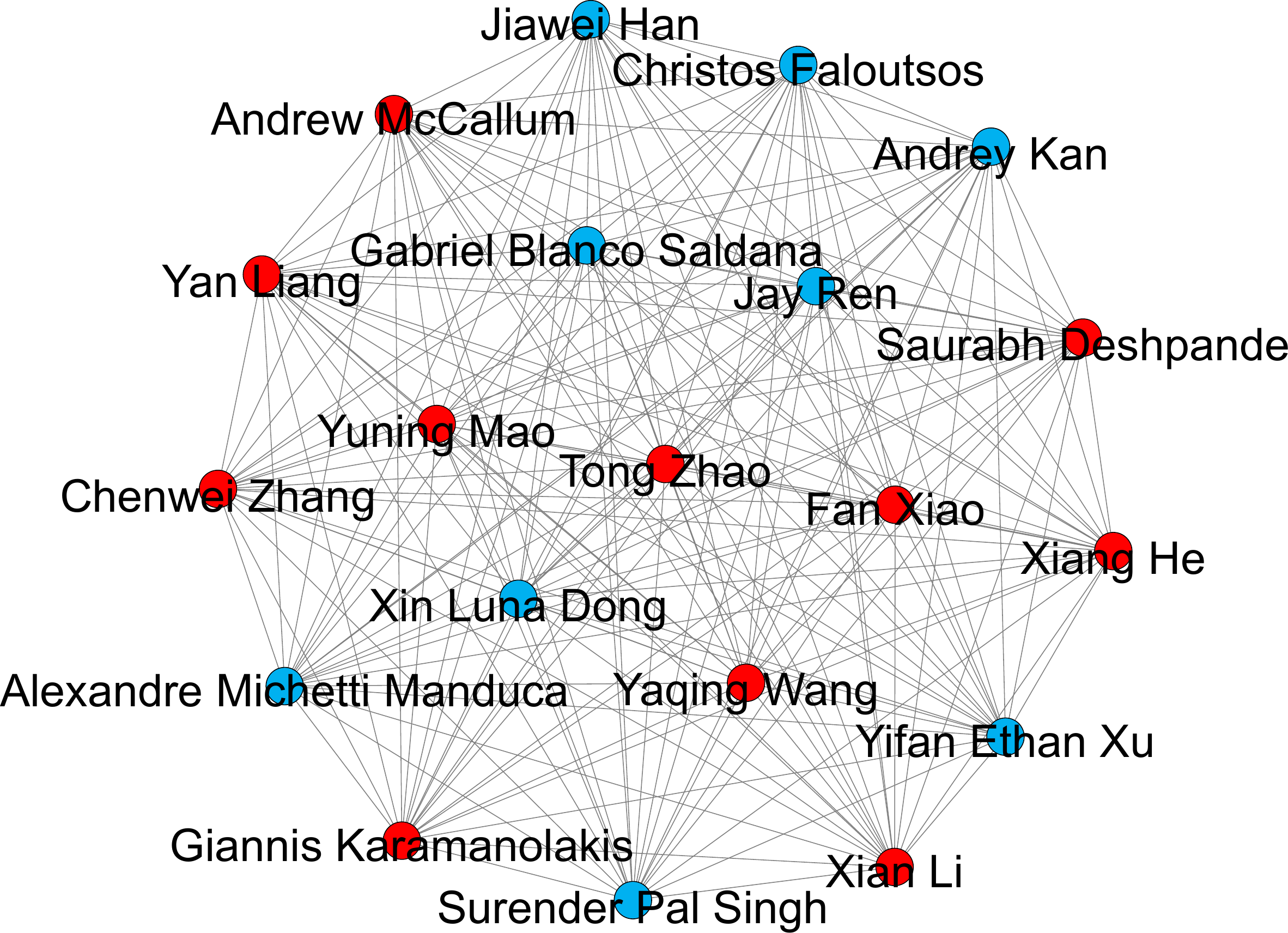}
			\end{minipage}
		}
	\end{center}
 \vspace*{-0.4cm}
	\caption{Case studies on \aminer and \dbai}
	\vspace*{-0.5cm}
	\label{fig:exp:casestudytotal1}
\end{figure}
}

\begin{figure*}[t]\vspace*{-0.4cm}
	\begin{center}		
        \subfigure[\aminer]{
			\label{fig:caseaminer}
			\begin{minipage}[b]{5.4cm}%{2.5cm}
				\centering
				\includegraphics[width=\textwidth]{figures/aminer3.pdf}
			\end{minipage}
		}
    \hspace{-0.7cm}
		\subfigure[\dbai]{
			\label{fig:casedbai}
			\begin{minipage}[b]{5.15cm}%{2.5cm}
				\centering
				\includegraphics[width=\textwidth]{figures/dbai3.pdf}
			\end{minipage}
		}
  \hspace{-0.7cm}
        \subfigure[\nba]{
			\label{fig:casenba}
			\begin{minipage}[b]{4.15cm}%{2.5cm}
				\centering
				\includegraphics[width=\textwidth]{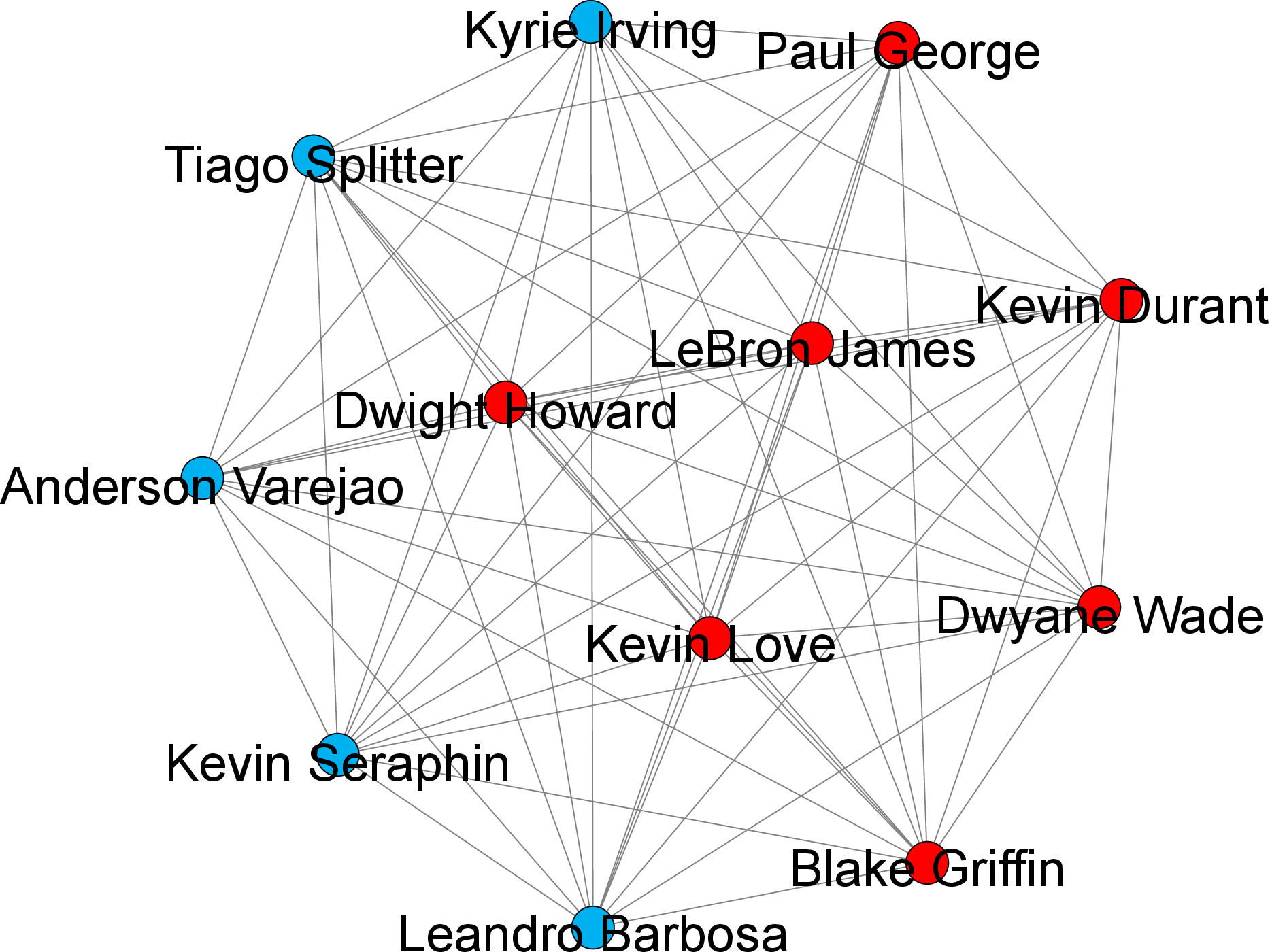}
			\end{minipage}
		}
  \hspace{-0.7cm}
		\subfigure[\imdb]{
			\label{fig:caseimdb}
			\begin{minipage}[b]{3.9cm}%{2.5cm}
				\centering
				\includegraphics[width=\textwidth]{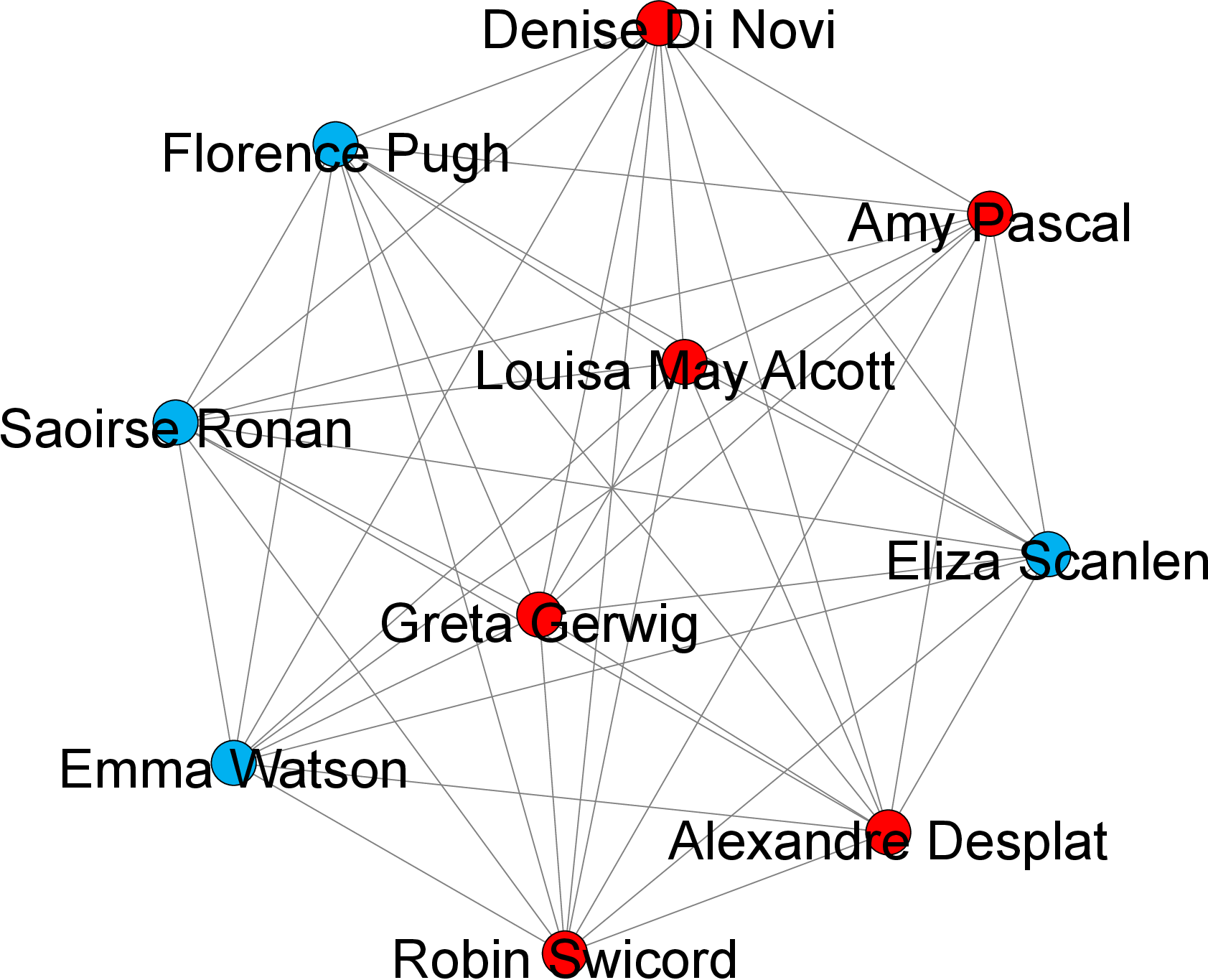}
			\end{minipage}
		}
	\end{center}
 \vspace*{-0.4cm}
	\caption{Case studies on \aminer, \dbai, \nba and \imdb}
	\vspace*{-0.5cm}
	\label{fig:exp:casestudytotal1}
\end{figure*}

 \comment{
\begin{figure}[t]
	\centering
	\includegraphics[width=0.38\textwidth]{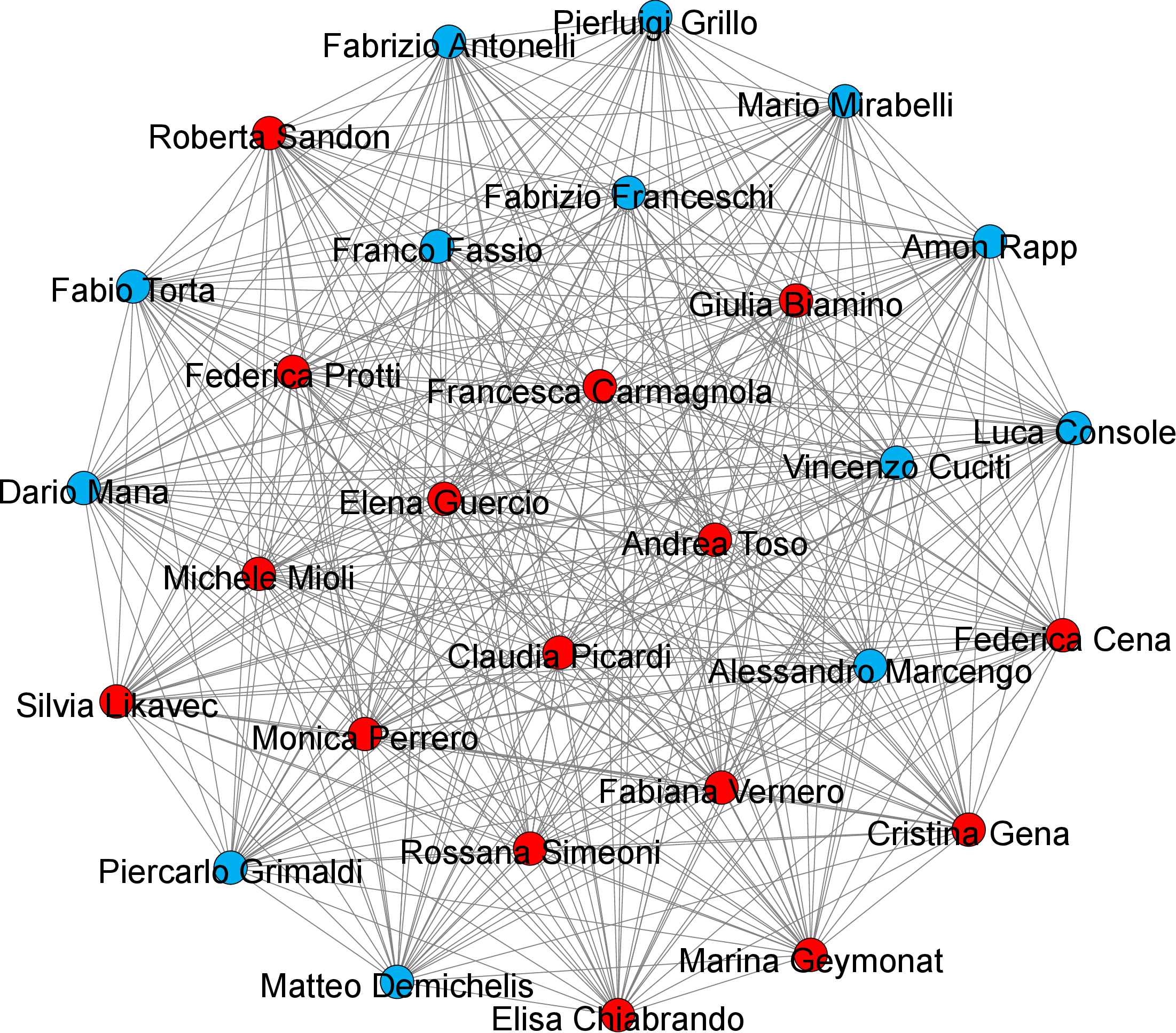}	
	\vspace*{-0.2cm}
	\caption{Case study on \aminer}
	\label{fig:caseaminer}
	\vspace*{-0.4cm}
\end{figure}

\begin{figure}[t]
	\centering
	\includegraphics[width=0.38\textwidth]{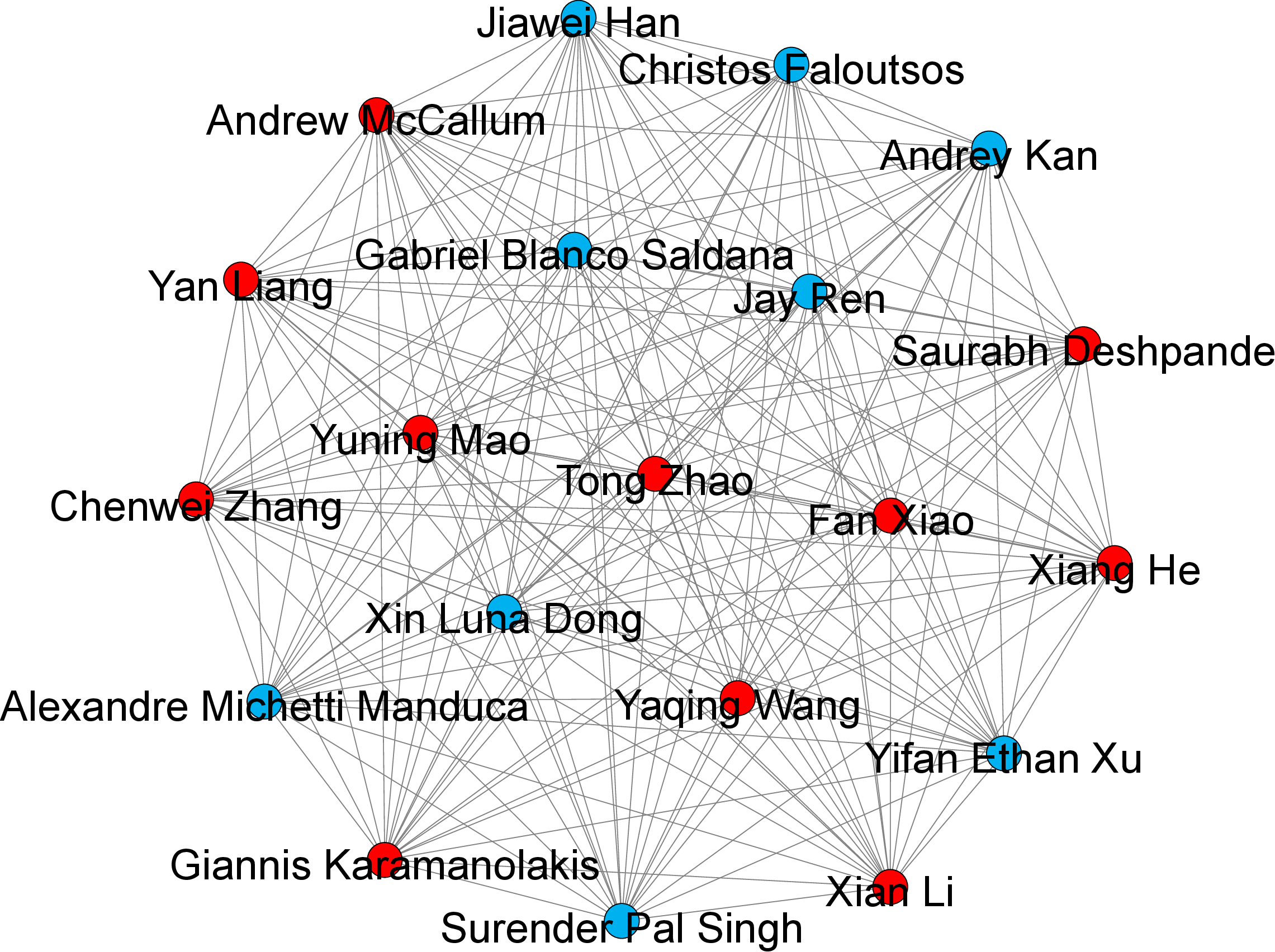}	
	\vspace*{-0.2cm}
	\caption{Case study on \dbai}
	\label{fig:casedbai}
	\vspace*{-0.4cm}
\end{figure}
}

\stitle{Case study on \dbai.} We conduct a case study on a collaboration network \dbai. The \dbai dataset is a subgraph of \DBLP downloaded from \url{dblp.uni-trier.de/xml/}, which contains the authors who had published at least one paper in the database ($DB$) and artificial intelligence ($AI$) related conferences. The subgraph contains 139,675 vertices and 975,722 undirected edges. The attribute $A$ represents the author's main research area, i.e., $A= \{DB, AI\}$. We assign the attribute for each vertex based on the maximum number of papers an author published in the related conferences. Performing our algorithms with $k=5$ and $\delta=3$, the maximum fair clique is depicted in \figref{fig:casedbai}, which includes $9$ scholars specializing in $DB$ (colored blue) and $11$ in $AI$ (colored red), maintaining a difference within $\delta$ between the scholar counts of each research field. These scholars have garnered considerable recognition within databases and artificial intelligence. For instance, Prof. Jiawei Han focuses on knowledge discovery, data mining, and database systems, boasting an impressive h-index of 200. Similarly, Prof. Andrew McCallum's expertise lies in statistical machine learning, natural language processing, and information retrieval, reflected in his h-index of 117. When embarking on a research project that demands a blend of database and machine learning expertise, our algorithms come to the fore. They identify the largest and most specialized cohort, ensuring equilibrium in participant numbers across the two distinct research directions.

Additionally, the maximum fair clique size can illuminate the intersecting degree between these two different research directions. The minuscule size of the maximum fair clique implies limited linkage between the two directions, while a larger maximum fair clique suggests a robust interconnection. Insights derived from our algorithms can guide interdisciplinary collaborations and research initiatives.

\stitle{Case study on \nba.} The \nba dataset, sourced from \url{https://github.com/yushundong/PyGDebias}, contains 403 basketball players and 21,242 relationships. Players' nationalities serve as attributes, i.e., $A=\{U.S., Oversea\}$. Invoking specified parameters of $k=5$ and $\delta=3$, our algorithms determine a maximum fair clique, illustrated in \figref{fig:casenba}. Red vertices represent 7 U.S. players, while blue vertices denote 5 players from overseas. All these individuals are widely renowned NBA stars, connected either through shared team histories or robust personal friendships. For instance, LeBron James, Kyrie Irving, and Kevin Love were core players for the Cavaliers, contributing to their 2016 NBA championship win. Dwyane Wade and LeBron James formed a dynamic partnership while playing together for the Miami Heat, securing two NBA championships. Anderson Varejao, Leandro Barbosa, and Tiago Splitter, representing Brazil, have collectively competed in prestigious international basketball events like the Olympics and World Cup, fostering a strong camaraderie through national team participation. Discovering a dense organization with a large size that encompasses a nearly equivalent count of foreign and local stars by our algorithms holds significant potential for sports clubs, athletes, and brands. This potential extends to attracting a broader fan base, expanding exposure, enhancing brand recognition, and ultimately amplifying the impact of their social media marketing endeavors.

\comment{
\begin{figure}[t] \vspace*{-0.2cm}
	\begin{center}		
	 \subfigure[\nba]{
			\label{fig:casenba}
			\begin{minipage}[b]{3.25cm}%{2.5cm}
				\centering
				\includegraphics[width=\textwidth]{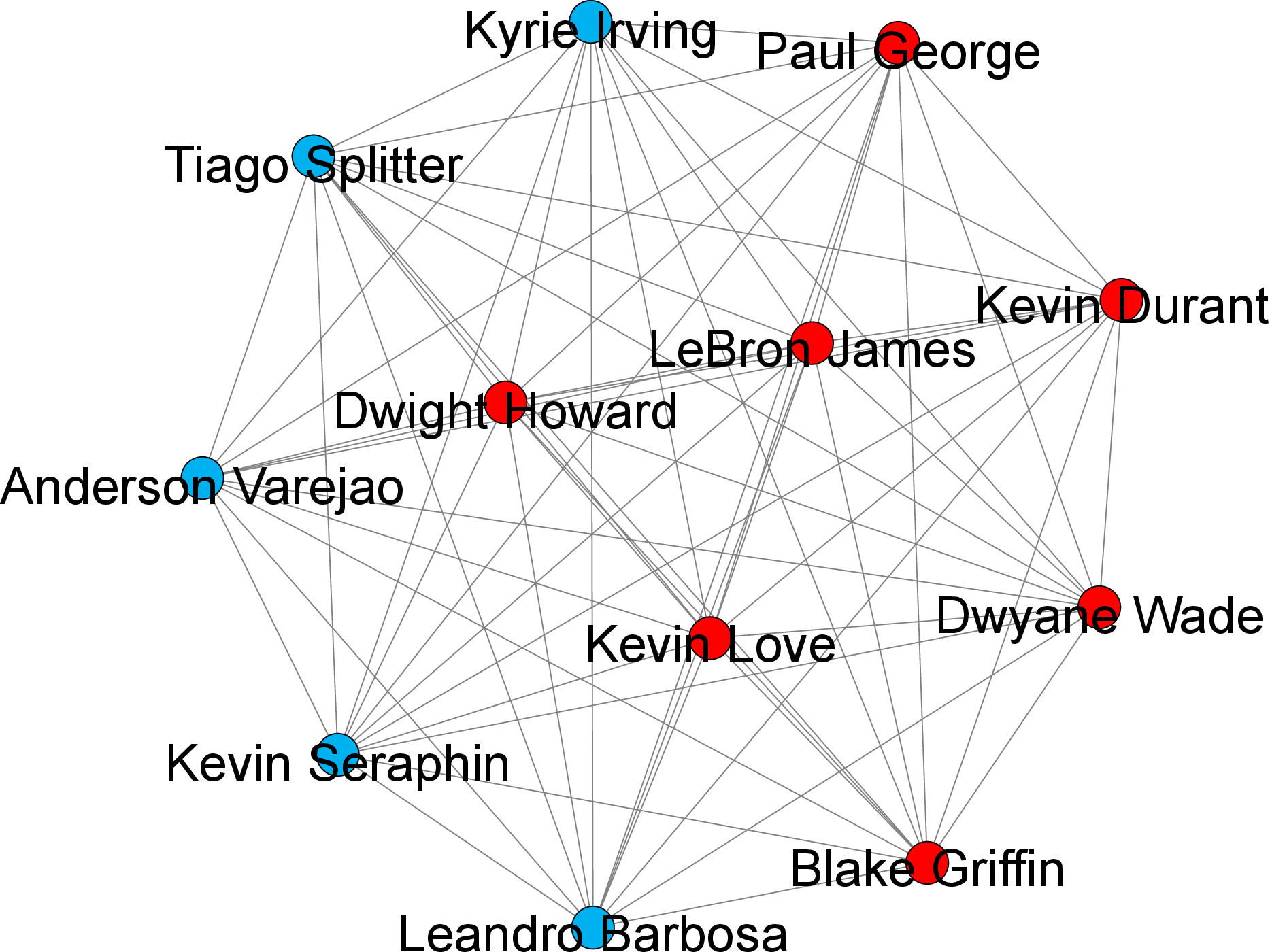}
			\end{minipage}
		}
  \hspace{0.25cm}
		\subfigure[\imdb]{
			\label{fig:caseimdb}
			\begin{minipage}[b]{3.25cm}%{2.5cm}
				\centering
				\includegraphics[width=\textwidth]{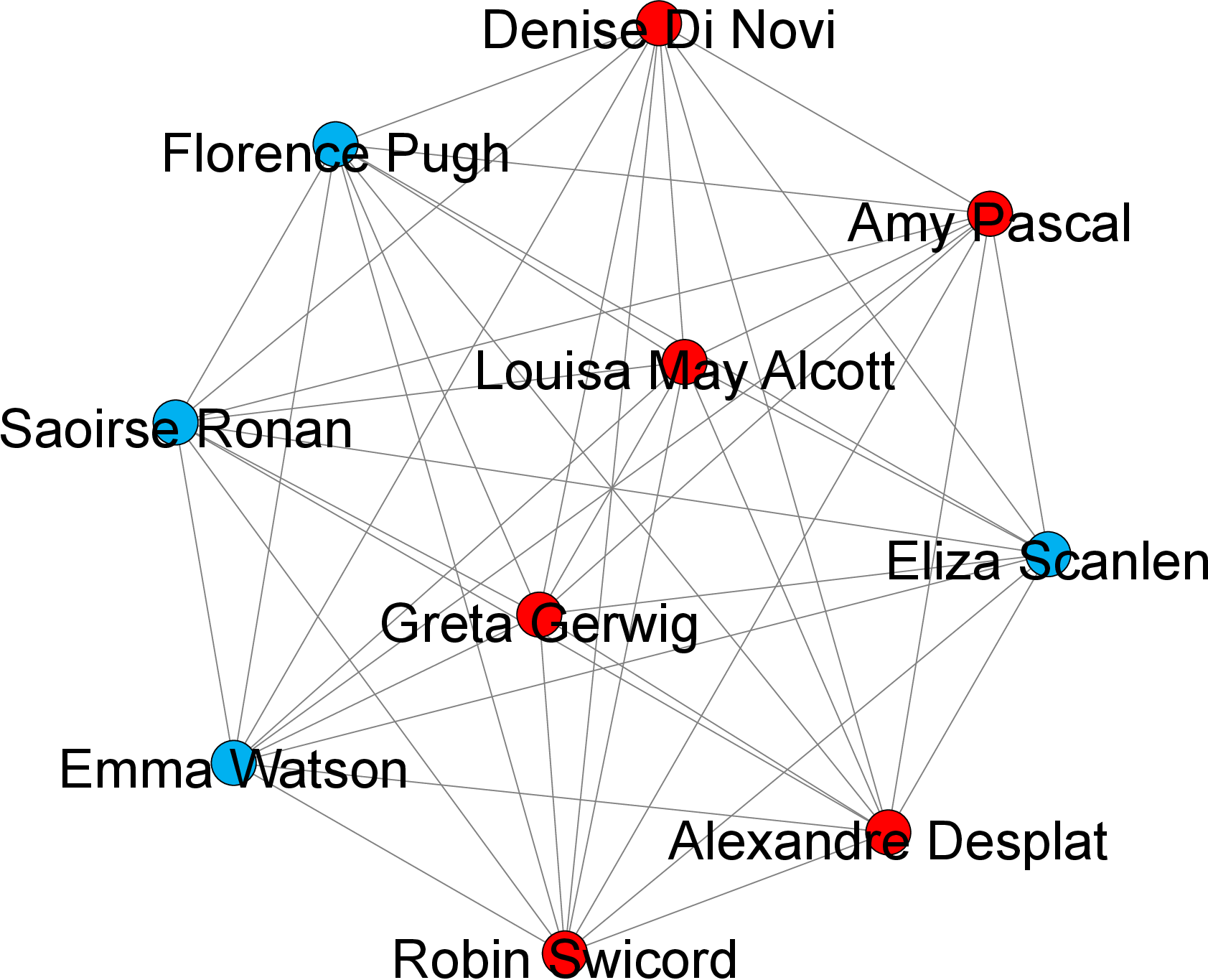}
			\end{minipage}
		}
	\end{center}
 \vspace*{-0.4cm}
	\caption{Case studies on \nba and \imdb}
	\vspace*{-0.5cm}
	\label{fig:exp:casestudytotal2}
\end{figure}
}

\comment{
\begin{figure}[t]
	\centering
	\includegraphics[width=0.3\textwidth]{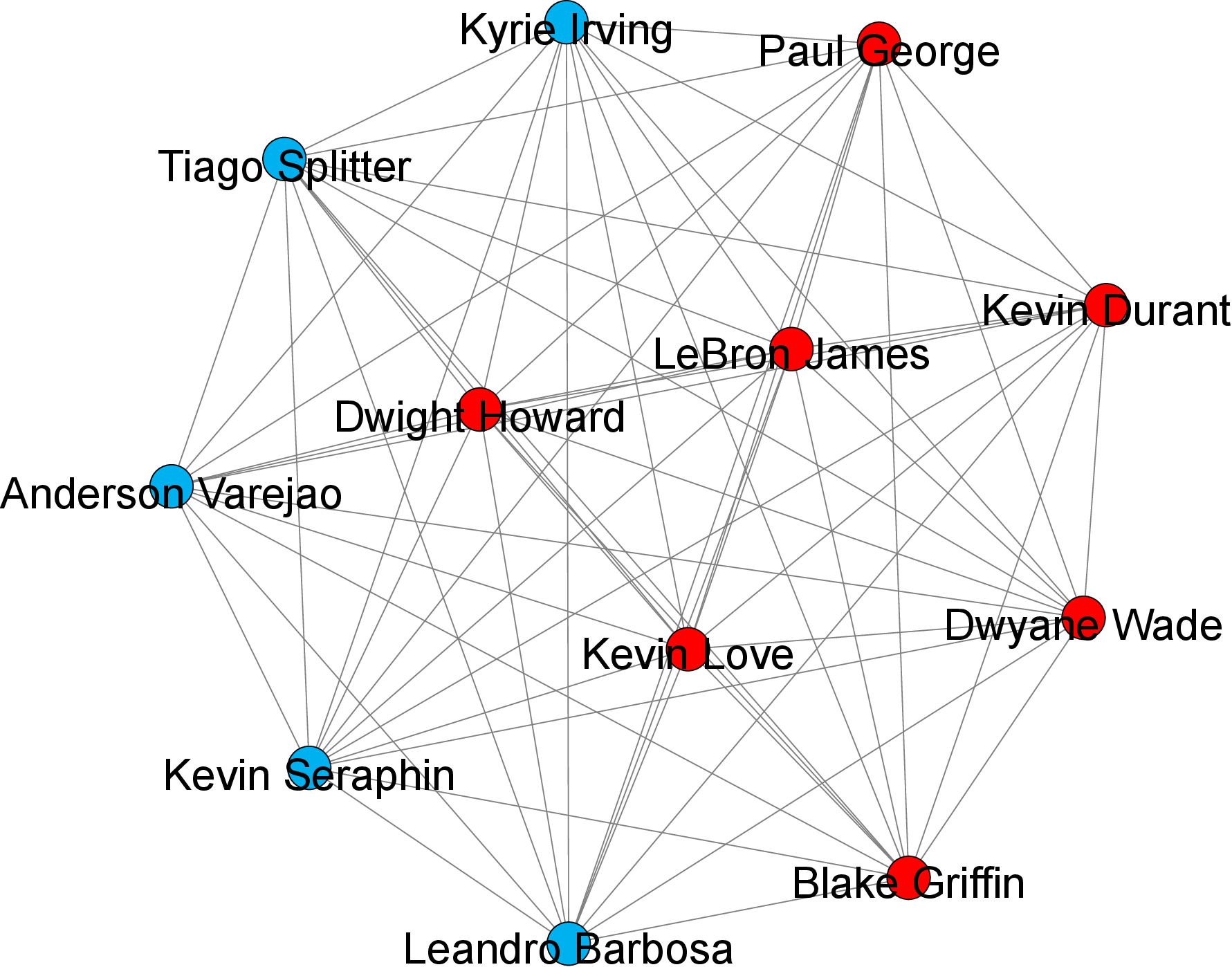}	
	\vspace*{-0.2cm}
	\caption{Case study on \nba}
	\label{fig:casenba}
	\vspace*{-0.4cm}
\end{figure}
}

\stitle{Case study on \imdb.} We conduct a case study on a movie dataset obtained from \url{https://developer.imdb.com}. Filtering out movies categorized as $titleType=movie$ and $isAdult=0$, we create a graph \imdb. This graph comprises 583,933 vertices representing actors, directors, writers, and others, connected by 29,332,894 edges indicating their collaborations. Each vertex is associated with an attribute from $A=\{S, J\}$, where $S$ represents a senior artist and $J$ denotes a junior artist. This categorization is based on birth year: with individuals born before 1990 classified as $S$ and those born after as $J$. Using our algorithms with parameters $k=5$ and $\delta=3$, we identify the maximum fair clique as depicted in \figref{fig:caseimdb}. The team connected to the film ``Little Women'' intricately combines $4$ junior artists (colored blue) and $6$ senior artists (colored red). Among them, Louisa May Alcott is the novelist behind the film's source material, and Greta Gerwig takes on the directorial role. Denise Di Novi, Robin Swicord, and Amy Pascal manage production aspects. Alexandre Desplat contributes his musical talents to compose the soundtrack, and the others are accomplished actors. This movie boasted an IMDB rating of 7.8 and earned a place among the top 10 movies of the year according to the American Film Institute. It also secured nominations at esteemed award ceremonies like the Academy Awards, BAFTAs, and Golden Globes. This serves as evidence that a diverse team comprising both young and seasoned artists can blend creativity, expertise, and experience to elevate the quality of cinematic production. Identifying such a team through our algorithms and investing in it can yield substantial returns.

\comment{
\begin{figure}[t]
	\centering
	\includegraphics[width=0.3\textwidth]{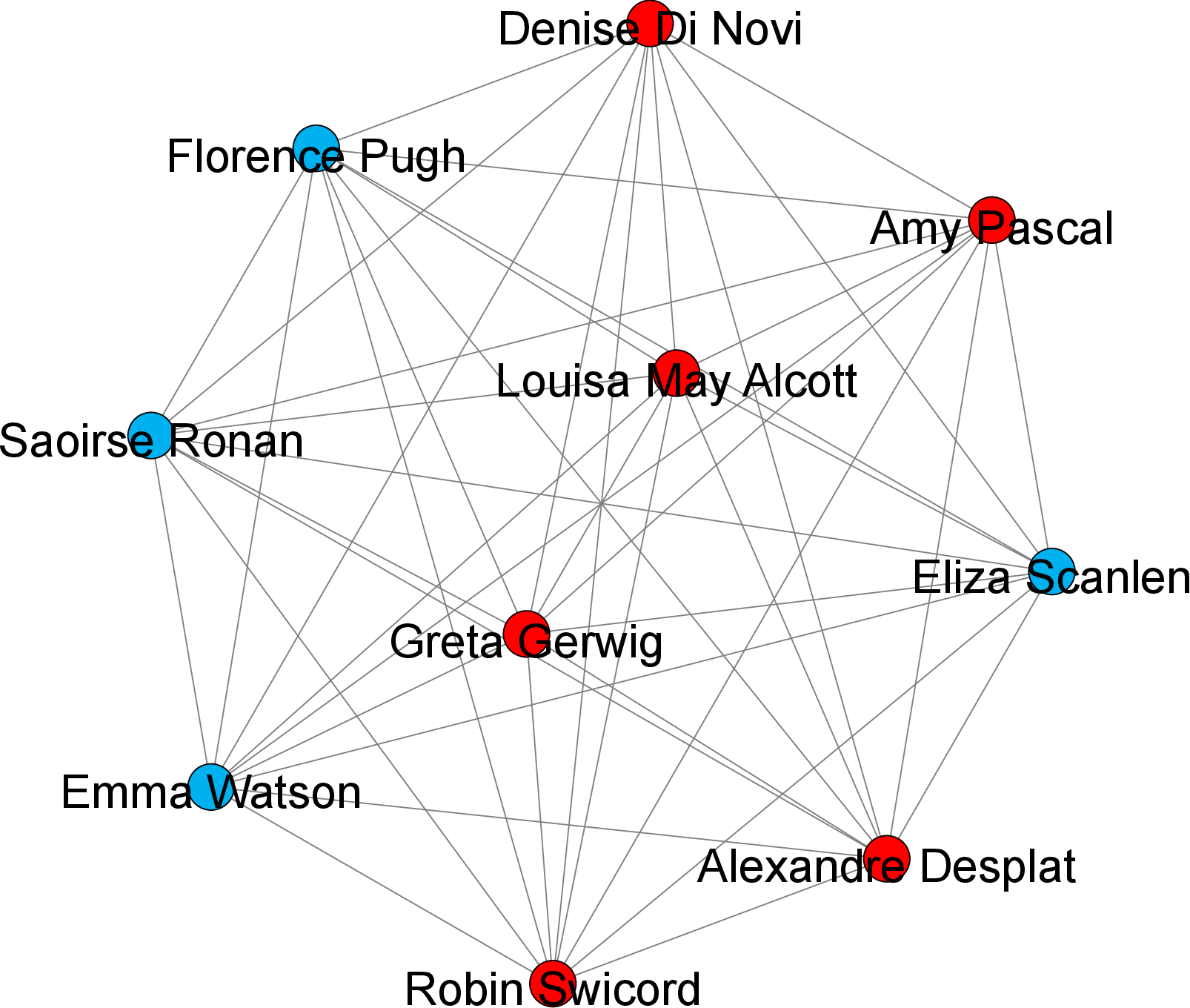}	
	\vspace*{-0.2cm}
	\caption{Case study on \imdb}
	\label{fig:caseimdb}
	\vspace*{-0.4cm}
\end{figure}
}

\section{Related work} \label{sec:relatedwork}
\stitle{Maximum clique computation.} Our work is closely related to the Maximum Clique Computation (MCC) problem, aiming to find the clique with the largest number of nodes. The MCC problem falls into the domain of NP-hard problems \cite{DBLP:conf/coco/Karp72}. Existing research primarily centers on devising heuristic algorithms that approximate solutions close to the maximum clique size. These heuristic algorithms iteratively augment the partial clique $R$ by adding vertices from the candidate set $C$ based on specific greedy strategies until $C$ is empty. For example, the maximum degree-based heuristic greedily selects the vertex with the highest degree to extend $R$ in each iterative step \cite{DBLP:journals/im/PattabiramanPGL15}, while the degeneracy order-based heuristic prioritizes vertices with the largest degeneracy for inclusion into $R$ \cite{DBLP:journals/siamsc/RossiGG15}. The ego-centric degeneracy-based heuristic extends the degeneracy order-based approach to each vertex's ego network and identifies the largest one as the result \cite{DBLP:conf/kdd/Chang19, DBLP:journals/vldb/Chang20}. On the other hand, effective exact methods for the MCC problem are also extensively studied, primarily based on the branch and bound framework. These exact methods consider every possible vertex addition to the partial clique $R$ to form a new search branch and often employ upper bound-based pruning techniques to improve search efficiency \cite{DBLP:conf/ictai/LiFX13, DBLP:journals/cor/LiJM17,DBLP:conf/walcom/Tomita17,DBLP:conf/walcom/TomitaSHTW10,DBLP:conf/kdd/Chang19, DBLP:journals/vldb/Chang20}. Chang \etal presented a state-of-the-art algorithm for the MCC problem, transforming the MCC problem on sparse graphs into multiple dense graphs. They also provided a branch-reduce-bound framework to compute the maximum clique on dense graphs \cite{DBLP:conf/kdd/Chang19, DBLP:journals/vldb/Chang20}. In this paper, we focus on the fair clique model and study the maximum fair clique search problem. Due to the inherent differences between clique and fair clique concepts, all the aforementioned algorithms cannot be directly applied to address our problem.

%For instance, Li \etal ingeniously combine the incremental upper bound and the MaxSAT reasoning technique to propose a branch and bound method tailored to the MCC problem \cite{DBLP:conf/ictai/LiFX13}. Li \etal further derive a static vertex search ordering, which can be combined with a dynamic vertex ordering to minimize the number of search branches \cite{DBLP:journals/cor/LiJM17}. 

%It's important to note that the above algorithms are tailored to the clique model. 

\stitle{Fairness-aware data mining.} Our work is motivated by the concept of fairness. It has attracted much attention in the machine learning research area, such as the classification task \cite{DBLP:conf/alt/CotterJS19, DBLP:conf/aistats/Narasimhan18, DBLP:conf/colt/WoodworthGOS17, DBLP:conf/icml/ZemelWSPD13,DBLP:conf/innovations/DworkHPRZ12,DBLP:conf/nips/HardtPNS16} and the recommendation task \cite{DBLP:conf/sigir/BiegaGW18, DBLP:conf/kdd/SinghJ18, DBLP:conf/nips/SinghJ19, DBLP:conf/www/Zehlike020, DBLP:conf/sigmod/AsudehJS019, DBLP:conf/kdd/BeutelCDQWWHZHC19}. Within the field of data mining, Pan \etal blazed a trail by introducing fairness into the clique model, proposing both weak and strong fair clique models, as well as a suite of enumeration algorithms \cite{DBLP:conf/icde/PanLZDTW22}. Based on this foundation, Zhang \etal introduced the relative fair clique model, offering a compromise between weak and strong fair clique models \cite{zhang2023fairness}. Hao \etal defined the absolute fair clique model and studied the problem of finding absolute fair cliques from attributed social networks \cite{hao2023afcminer}. Qiao \etal incorporated fairness into the KPcore model, formulating the maximum core mining problem on heterogeneous information networks \cite{qlp2022community}. In addition, Yin \etal focused on fairness within bipartite graphs, introducing the single-side and bi-side fair bicliques, and studied the problem of fairness-aware biclique enumeration. \cite{DBLP:journals/corr/abs-2303-03705}. This paper, for the first time, investigates the problem of finding the relative fair clique with the largest size. Among the mentioned studies, only the relative fair clique enumeration algorithms introduced in \cite{zhang2023fairness} possess adaptability for solving our problem. However, these algorithms tend to exhibit inefficiency, especially when dealing with large graphs. In light of this, we propose efficient graph reduction techniques and deploy a series of upper-bounding pruning techniques to enhance the efficiency of finding the maximum fair clique. 

\section{Conclusion} \label{sec:conclusion}
This paper studies the problem of finding the maximum fair clique in large graphs. Two novel graph reduction techniques grounded in colorful support are presented, aimed at shrinking graph size. Then, we propose a series of upper-bounding techniques to prune needless search space during the branch-and-bound procedure. Adding to this, a linear time complexity heuristic algorithm based on degree and colorful degree greedy strategies is presented for finding a larger fair clique, which can also be used to prune branches to further improve search efficiency. Comprehensive experiments on six real-life graphs demonstrate the efficiency, scalability and effectiveness of the proposed algorithms.

%\clearpage
\balance
\bibliographystyle{IEEEtran}
\bibliography{maximumfairclique}

% Generated by IEEEtran.bst, version: 1.14 (2015/08/26)
\begin{thebibliography}{10}
\providecommand{\url}[1]{#1}
\csname url@samestyle\endcsname
\providecommand{\newblock}{\relax}
\providecommand{\bibinfo}[2]{#2}
\providecommand{\BIBentrySTDinterwordspacing}{\spaceskip=0pt\relax}
\providecommand{\BIBentryALTinterwordstretchfactor}{4}
\providecommand{\BIBentryALTinterwordspacing}{\spaceskip=\fontdimen2\font plus
\BIBentryALTinterwordstretchfactor\fontdimen3\font minus
  \fontdimen4\font\relax}
\providecommand{\BIBforeignlanguage}[2]{{%
\expandafter\ifx\csname l@#1\endcsname\relax
\typeout{** WARNING: IEEEtran.bst: No hyphenation pattern has been}%
\typeout{** loaded for the language `#1'. Using the pattern for}%
\typeout{** the default language instead.}%
\else
\language=\csname l@#1\endcsname
\fi
#2}}
\providecommand{\BIBdecl}{\relax}
\BIBdecl

\bibitem{chang2018cohesive}
L.~Chang and L.~Qin, \emph{Cohesive subgraph computation over large sparse
  graphs: algorithms, data structures, and programming techniques}, 2018.

\bibitem{bron1973finding}
C.~Bron and J.~Kerbosch, ``Finding all cliques of an undirected graph
  (algorithm 457),'' \emph{Commun. ACM}, vol.~16, no.~9, pp. 575--576, 1973.

\bibitem{chang2019efficient}
L.~Chang, ``Efficient maximum clique computation over large sparse graphs,'' in
  \emph{KDD}, 2019, pp. 529--538.

\bibitem{eppstein2013listing}
D.~Eppstein, M.~L{\"o}ffler, and D.~Strash, ``Listing all maximal cliques in
  large sparse real-world graphs,'' \emph{Journal of Experimental
  Algorithmics}, vol.~18, pp. 3--1, 2013.

\bibitem{li2017minimization}
C.-M. Li, H.~Jiang, and F.~Many{\`a}, ``On minimization of the number of
  branches in branch-and-bound algorithms for the maximum clique problem,''
  \emph{Computers \& Operations Research}, vol.~84, pp. 1--15, 2017.

\bibitem{san2016new}
P.~San~Segundo, A.~Lopez, and P.~M. Pardalos, ``A new exact maximum clique
  algorithm for large and massive sparse graphs,'' \emph{Computers \&
  Operations Research}, vol.~66, pp. 81--94, 2016.

\bibitem{DBLP:conf/alt/CotterJS19}
A.~Cotter, H.~Jiang, and K.~Sridharan, ``Two-player games for efficient
  non-convex constrained optimization,'' in \emph{ALT}, ser. Proceedings of
  Machine Learning Research, vol.~98, 2019, pp. 300--332.

\bibitem{DBLP:conf/aistats/Narasimhan18}
H.~Narasimhan, ``Learning with complex loss functions and constraints,'' in
  \emph{AISTATS}, ser. Proceedings of Machine Learning Research, vol.~84, 2018,
  pp. 1646--1654.

\bibitem{DBLP:conf/colt/WoodworthGOS17}
B.~E. Woodworth, S.~Gunasekar, M.~I. Ohannessian, and N.~Srebro, ``Learning
  non-discriminatory predictors,'' in \emph{COLT}, ser. Proceedings of Machine
  Learning Research, vol.~65, 2017, pp. 1920--1953.

\bibitem{DBLP:conf/icml/ZemelWSPD13}
R.~S. Zemel, Y.~Wu, K.~Swersky, T.~Pitassi, and C.~Dwork, ``Learning fair
  representations,'' in \emph{ICML}, ser. {JMLR} Workshop and Conference
  Proceedings, vol.~28, 2013, pp. 325--333.

\bibitem{DBLP:conf/kdd/SinghJ18}
A.~Singh and T.~Joachims, ``Fairness of exposure in rankings,'' in \emph{KDD},
  2018, pp. 2219--2228.

\bibitem{DBLP:conf/nips/SinghJ19}
{Ashudeep Singh and Thorsten Joachims}, ``Policy learning for fairness in
  ranking,'' in \emph{NeurIPS}, 2019, pp. 5427--5437.

\bibitem{DBLP:conf/sigmod/AsudehJS019}
A.~Asudeh, H.~V. Jagadish, J.~Stoyanovich, and G.~Das, ``Designing fair ranking
  schemes,'' in \emph{SIGMOD}, 2019, pp. 1259--1276.

\bibitem{DBLP:conf/kdd/BeutelCDQWWHZHC19}
A.~Beutel, J.~Chen, T.~Doshi, H.~Qian, L.~Wei, Y.~Wu, L.~Heldt, Z.~Zhao,
  L.~Hong, E.~H. Chi, and C.~Goodrow, ``Fairness in recommendation ranking
  through pairwise comparisons,'' in \emph{KDD}, 2019, pp. 2212--2220.

\bibitem{mehrabi2019debiasing}
N.~Mehrabi, F.~Morstatter, N.~Peng, and A.~Galstyan, ``Debiasing community
  detection: the importance of lowly connected nodes,'' in \emph{ASONAM}, 2019,
  pp. 509--512.

\bibitem{lipton2018does}
Z.~Lipton, J.~McAuley, and A.~Chouldechova, ``Does mitigating ml's impact
  disparity require treatment disparity?'' \emph{Advances in neural information
  processing systems}, vol.~31, 2018.

\bibitem{louizos2015variational}
C.~Louizos, K.~Swersky, Y.~Li, M.~Welling, and R.~Zemel, ``The variational fair
  autoencoder,'' \emph{arXiv preprint arXiv:1511.00830}, 2015.

\bibitem{du2019learning}
M.~Du, N.~Liu, F.~Yang, and X.~Hu, ``Learning credible deep neural networks
  with rationale regularization,'' in \emph{ICDM}, 2019, pp. 150--159.

\bibitem{ross2017right}
A.~S. Ross, M.~C. Hughes, and F.~Doshi-Velez, ``Right for the right reasons:
  Training differentiable models by constraining their explanations,''
  \emph{arXiv preprint arXiv:1703.03717}, 2017.

\bibitem{elazar2018adversarial}
Y.~Elazar and Y.~Goldberg, ``Adversarial removal of demographic attributes from
  text data,'' \emph{arXiv preprint arXiv:1808.06640}, 2018.

\bibitem{zhang2018mitigating}
B.~H. Zhang, B.~Lemoine, and M.~Mitchell, ``Mitigating unwanted biases with
  adversarial learning,'' in \emph{AAAI}, 2018, pp. 335--340.

\bibitem{wang2019balanced}
T.~Wang, J.~Zhao, M.~Yatskar, K.-W. Chang, and V.~Ordonez, ``Balanced datasets
  are not enough: Estimating and mitigating gender bias in deep image
  representations,'' in \emph{ICCV}, 2019, pp. 5310--5319.

\bibitem{DBLP:conf/icde/PanLZDTW22}
M.~Pan, R.~Li, Q.~Zhang, Y.~Dai, Q.~Tian, and G.~Wang, ``Fairness-aware maximal
  clique enumeration,'' in \emph{ICDE}, 2022, pp. 259--271.

\bibitem{zhang2023fairness}
Q.~Zhang, R.-H. Li, M.~Pan, Y.~Dai, Q.~Tian, and G.~Wang, ``Fairness-aware
  maximal clique in large graphs: Concepts and algorithms,'' \emph{IEEE TKDE},
  2023.

\bibitem{hao2023afcminer}
F.~Hao, Y.~Yang, J.~Shang, and D.-S. Park, ``Afcminer: Finding absolute fair
  cliques from attributed social networks for responsible computational social
  systems,'' \emph{IEEE TCSS}, 2023.

\bibitem{DBLP:journals/corr/abs-2303-03705}
Z.~Yin, Q.~Zhang, W.~Zhang, R.~Li, and G.~Wang, ``Fairness-aware maximal
  biclique enumeration on bipartite graphs,'' \emph{CoRR}, vol. abs/2303.03705,
  2023.

\bibitem{qlp2022community}
L.~Qiao, H.~Hou, and G.~Wang, ``Community search algorithm on heterogeneous
  information networks based on attribute fairness,'' \emph{Journal of
  Software}, vol.~34, no.~3, pp. 0--0, 2022.

\bibitem{matula1972graph}
D.~W. Matula, G.~Marble, and J.~D. Isaacson, ``Graph coloring algorithms,'' in
  \emph{Graph theory and computing}, 1972, pp. 109--122.

\bibitem{jensen2011graph}
T.~R. Jensen and B.~Toft, \emph{Graph coloring problems}, 2011, vol.~39.

\bibitem{DBLP:conf/spaa/HasenplaughKSL14}
W.~Hasenplaugh, T.~Kaler, T.~B. Schardl, and C.~E. Leiserson, ``Ordering
  heuristics for parallel graph coloring,'' in \emph{SPAA}, 2014, pp. 166--177.

\bibitem{lick1970k}
D.~R. Lick and A.~T. White, ``k-degenerate graphs,'' \emph{Canadian Journal of
  Mathematics}, vol.~22, no.~5, pp. 1082--1096, 1970.

\bibitem{seidman1983network}
S.~B. Seidman, ``Network structure and minimum degree,'' \emph{Social
  networks}, vol.~5, no.~3, pp. 269--287, 1983.

\bibitem{DBLP:journals/pnas/Hirsch05}
J.~E. Hirsch, ``An index to quantify an individual's scientific research
  output,'' \emph{Proc. Natl. Acad. Sci. {USA}}, vol. 102, no.~46, pp.
  16\,569--16\,572, 2005.

\bibitem{DBLP:conf/kdd/WangCF13}
J.~Wang, J.~Cheng, and A.~W. Fu, ``Redundancy-aware maximal cliques,'' in
  \emph{KDD}, 2013, pp. 122--130.

\bibitem{DBLP:conf/stoc/EdenRS18}
T.~Eden, D.~Ron, and C.~Seshadhri, ``On approximating the number of k-cliques
  in sublinear time,'' in \emph{{STOC}}, 2018, pp. 722--734.

\bibitem{DBLP:conf/coco/Karp72}
R.~M. Karp, ``Reducibility among combinatorial problems,'' in \emph{CCC}, ser.
  The {IBM} Research Symposia Series, 1972, pp. 85--103.

\bibitem{DBLP:journals/im/PattabiramanPGL15}
B.~Pattabiraman, M.~M.~A. Patwary, A.~H. Gebremedhin, W.~Liao, and A.~N.
  Choudhary, ``Fast algorithms for the maximum clique problem on massive graphs
  with applications to overlapping community detection,'' \emph{Internet
  Math.}, vol.~11, no. 4-5, pp. 421--448, 2015.

\bibitem{DBLP:journals/siamsc/RossiGG15}
R.~A. Rossi, D.~F. Gleich, and A.~H. Gebremedhin, ``Parallel maximum clique
  algorithms with applications to network analysis,'' \emph{{SIAM} J. Sci.
  Comput.}, vol.~37, no.~5, 2015.

\bibitem{DBLP:conf/kdd/Chang19}
L.~Chang, ``Efficient maximum clique computation over large sparse graphs,'' in
  \emph{{KDD}}, 2019, pp. 529--538.

\bibitem{DBLP:journals/vldb/Chang20}
{Lijun Chang}, ``Efficient maximum clique computation and enumeration over
  large sparse graphs,'' \emph{{VLDB} J.}, vol.~29, no.~5, pp. 999--1022, 2020.

\bibitem{DBLP:conf/ictai/LiFX13}
C.~Li, Z.~Fang, and K.~Xu, ``Combining maxsat reasoning and incremental upper
  bound for the maximum clique problem,'' in \emph{ICTAI}, 2013, pp. 939--946.

\bibitem{DBLP:journals/cor/LiJM17}
C.~Li, H.~Jiang, and F.~Many{\`{a}}, ``On minimization of the number of
  branches in branch-and-bound algorithms for the maximum clique problem,''
  \emph{Comput. Oper. Res.}, vol.~84, pp. 1--15, 2017.

\bibitem{DBLP:conf/walcom/Tomita17}
E.~Tomita, ``Efficient algorithms for finding maximum and maximal cliques and
  their applications,'' in \emph{WALCOM}, ser. Lecture Notes in Computer
  Science, vol. 10167, 2017, pp. 3--15.

\bibitem{DBLP:conf/walcom/TomitaSHTW10}
E.~Tomita, Y.~Sutani, T.~Higashi, S.~Takahashi, and M.~Wakatsuki, ``A simple
  and faster branch-and-bound algorithm for finding a maximum clique,'' in
  \emph{WALCOM}, ser. Lecture Notes in Computer Science, vol. 5942, 2010, pp.
  191--203.

\bibitem{DBLP:conf/innovations/DworkHPRZ12}
C.~Dwork, M.~Hardt, T.~Pitassi, O.~Reingold, and R.~S. Zemel, ``Fairness
  through awareness,'' in \emph{ITCS}, 2012, pp. 214--226.

\bibitem{DBLP:conf/nips/HardtPNS16}
M.~Hardt, E.~Price, and N.~Srebro, ``Equality of opportunity in supervised
  learning,'' in \emph{NeurIPS}, 2016, pp. 3315--3323.

\bibitem{DBLP:conf/sigir/BiegaGW18}
A.~J. Biega, K.~P. Gummadi, and G.~Weikum, ``Equity of attention: Amortizing
  individual fairness in rankings,'' in \emph{{SIGIR}}, 2018, pp. 405--414.

\bibitem{DBLP:conf/www/Zehlike020}
M.~Zehlike and C.~Castillo, ``Reducing disparate exposure in ranking: {A}
  learning to rank approach,'' in \emph{{WWW}}, 2020, pp. 2849--2855.

\end{thebibliography}

\end{document}